\documentclass[sn-aps]{sn-jnl}% American Physical Society (APS) Reference Style
\usepackage{preamble}
\theoremstyle{thmstyleone}%
\newtheorem{theorem}{Theorem}%  meant for continuous numbers
\theoremstyle{thmstyletwo}%

\begin{document}

\title[Article Title]{Proof of completeness of the local conserved quantities in the one-dimensional Hubbard model}

\author*[1]{\fnm{Kohei} \sur{Fukai}}\email{k.fukai@issp.u-tokyo.ac.jp}

\affil*[1]{\orgdiv{The Institute for Solid State Physics}, \orgname{The University of Tokyo}, \orgaddress{\city{Kashiwa}, \postcode{277-8581}, \state{Chiba}, \country{Japan}}}

\abstract{We rigorously prove that the local conserved quantities in the one-dimensional Hubbard model are uniquely determined for each locality up to the freedom to add lower-order ones. From this, we can conclude that the local conserved quantities are exhausted by those obtained from the expansion of the transfer matrix.}


\maketitle

\section{Introduction}
%%% Integrability and local charges %%%
%%%%%%%%%%%%%%%%%%%%%%%%%%%%%%%%%%%%%%%%%%
Quantum many-body systems are called ``integrable'' when their dynamics or thermodynamics can be computed exactly.
%%%%%%%%%%%%%%%%%%%%%%%%%%%%%%%%%%%%%%%%%%
Although the rigorous criterion of quantum integrability has yet to be formulated, it is widely accepted that quantum integrability is characterized by the existence of an extensive number of local conserved quantities, which bears an analogy to the Liouville-Arnold definition of classical integrability~\cite{Grabowski_1995, Caux_2011}.

%%% transfer matrix and QISM %%%
%%%%%%%%%%%%%%%%%%%%%%%%%%%%%%%%%%%%%%%%%%
Yang-Baxter integrable systems are exactly solvable via the Bethe Ansatz~\cite{Bethe1931, Baxter1982}, and the existence of local conserved quantities is assured from the quantum inverse scattering method~\cite{korepin_bogoliubov_izergin_1993}: the local charges are generated by the expansion of the logarithm of the transfer matrix, though it is practically difficult to obtain their general explicit expression.
%%%%%%%%%%%%%%%%%%%%%%%%%%%%%%%%%%%%%%%%%%
For interacting integrable models, the local charges generated from the transfer matrix have been conjectured to be complete, meaning there are no additional local charges independent of them. 
%%%%%%%%%%%%%%%%%%%%%%%%%%%%%%%%%%%%%%%%%%
Rigorous proofs for this completeness have been established for some models: the spin-1/2 XXX chain~\cite{Babbitt-1979} and the XYZ chain~\cite{Nozawa2020}. 
%%%%%%%%%%%%%%%%%%%%%%%%%%%%%%%%%%%%%%%%%%
However, for almost all other cases, the rigorous proof of completeness of transfer matrix charges has yet to be presented.

%%% non-integrability %%%
%%%%%%%%%%%%%%%%%%%%%%%%%%%%%%%%%%%%%%%%%%
On the other hand, rigorous proof of non-integrability has also been demonstrated for the spin-$1/2$ XYZ chain with a magnetic field~\cite{Shiraishi2019}, the mixed-field Ising chain~\cite{chiba2023proof}, and the PXP model~\cite{pxp-nonintegrability-2024}, by providing rigorous proof of the absence of local charges beyond the Hamiltonian. 
%%%%%%%%%%%%%%%%%%%%%%%%%%%%%%%%%%%%%%%%%%
Their strategy is straightforward: first, write down the linear combination, which consists of a candidate of local conserved quantities using all local operators acting on finite range site, and demonstrate that if it is conserved, then all of the coefficients in this linear combination must be zero.

%%%%%%%% New findings %%%%%%%%
%%%%%%%%%%%%%%%%%%%%%%%%%%%%%%%%%%%%%%%%%%
In this article, we present rigorous proof that there are no other local conserved quantities in the one-dimensional Hubbard model independent of those obtained in~\cite{fukai-hubbard-charge-2023}, in the spirit of the proof of non-integrability~\cite{Shiraishi2019,chiba2023proof,pxp-nonintegrability-2024}.
%%%%%%%%%%%%%%%%%%%%%%%%%%%%%%%%%%%%%%%%%%
This proof does not need detailed knowledge of the structure of the local conserved quantities obtained in~\cite{fukai-hubbard-charge-2023}.
%%%%%%%%%%%%%%%%%%%%%%%%%%%%%%%%%%%%%%%%%%
From our result, it immediately follows that local charges generated from the expansion of the logarithm of the transfer matrix~\cite{shastry1988BF01022987,Olmedilla1988,Ramos1997,essler2005one} are written as a linear combination of those obtained in~\cite{fukai-hubbard-charge-2023}.
%%%%%%%%%%%%%%%%%%%%%%%%%%%%%%%%%%%%%%%%%%
We note that our considerations are solely focused on ultra-local charges, excluding quasi-local ones~\cite{Ilievski2015, Ilievski_2016}.

%%%%%%% The contents of this work %%%%%%%
This paper is organized as follows.
%%%%%%%%%%%%%%%%%%%%%%%%%%%%%%%%%%%%%%%%%%
In Section~\ref{sec:main_result}, we explain our setup and present the main theorem of this article.
%%%%%%%%%%%%%%%%%%%%%%%%%%%%%%%%%%%%%%%%%%
Section~\ref{sec:notation} introduces notations helpful in the proof of the main theorem.
%%%%%%%%%%%%%%%%%%%%%%%%%%%%%%%%%%%%%%%%%%
In Section~\ref{sec:proof}, we provide rigorous proof of the main theorem.
%%%%%%%%%%%%%%%%%%%%%%%%%%%%%%%%%%%%%%%%%%
Section~\ref{sec:summary} contains a summary of our work.
%%%%%%%%%%%%%%%%%%%%%%%%%%%%%%%%%%%%%%%%%%
In Appendix~\ref{app:one-sup}, we give all the one-support local charges in the one-dimensional Hubbard model.

\section{Main result}
\label{sec:main_result}
%%%%%%%%%%%%%%%%%%%%%%%%%%%%%%
In this section, we present the main result of this article in Theorem~\ref{theorem} and the subsequent Corollary~\ref{coro}.
%%%%%%%%%%%%%%%%%%%%%%%%%%%%%%
%%%%%%%%%%%%%%%%%%%%%%%%%%%%%%
The Hamiltonian of the one-dimensional Hubbard model is 
\begin{align}
    H
    =&
    -
    2t
    \sum_{j=1}^L
    \sum_{\sigma=\uparrow,\downarrow}
    \paren{
    c_{j,\sigma}^\dag c_{j+1,\sigma}+\text{h.c.}
    }
    +4U 
    \sum_{j=1}^L
    \paren{n_{j,\uparrow}-\frac{1}{2}}
    \paren{n_{j,\downarrow}-\frac{1}{2}}
    ,
    \label{eq:Hamiltonian}
\end{align}
where the periodic boundary condition is imposed, $n_{j\sigma}\equiv c_{j,\sigma}^\dag c_{j,\sigma}$, and $U$ is the coupling constant.
%%%%%%%%%%%%%%%%%%%%%%%%%%%%%%
In the following, we assume $U\neq 0$ and set $t=1$.
%%%%%%%%%%%%%%%%%%%%%%%%%%%%%%
Let the first (second) term in~\eqref{eq:Hamiltonian} be denoted by $H_0(H_{\mathrm{int}})$.
%%%%%%%%%%%%%%%%%%%%%%%%%%%%%%
The Hamiltonian~\eqref{eq:Hamiltonian} is known to be integrable with the (nested) Bethe ansatz~\cite{PhysRevLett.20.1445} and has an extensive number of local conserved quantities~\cite{shastry1988BF01022987,Olmedilla1988}.

\def\sumsym{e}
%%%%%%%%%%%%%%%%%%%%%%%%%%%%%%
We define \textit{$k$-support basis} starting from the $i$ th site by
\begin{align}
    \bm{\sumsym}_{i}^{k}
    &\equiv
    \bm{\sumsym}_{i,\uparrow}^{k}
    \bm{\sumsym}_{i,\downarrow}^{k}
    ,
    \\
    \bm{\sumsym}_{i,\sigma}^{k}
    &\equiv
    \sumsym_{i,\sigma}
    \sumsym_{i+1,\sigma}
    \cdots
    \sumsym_{i+k-1,\sigma}
    ,
\end{align}
where $\sigma\in \bce{\uparrow, \downarrow}$, $\sumsym_{j, \sigma} \in \{c_{j, \sigma}, c^\dag_{j, \sigma}, z_{j, \sigma}, I\}$, $z_{l,\sigma} = 2n_{l,\sigma} - 1$, $I$ is the identity operator, $\{\sumsym_{i,\uparrow}, \sumsym_{i,\downarrow}\} \neq \{I, I\}$, and $\{\sumsym_{i+k-1,\uparrow}, \sumsym_{i+k-1,\downarrow}\} \neq \{I,I\}$.
%%%%%%%%%%%%%%%%%%%%%%%%%%%%%%
We refer to a linear combination of $k$-support basis elements as a \textit{$k$-support operator}.

%%%%%%%%%%%%%%%%%%%%%%%%%%%%%%
We define a \textit{$k$-local conserved quantity} by
\begin{align}
    F_k
    =
    \sum_{l=1}^{k}
    \sum_{i=1}^{L}
    \sum_{\bm{\sumsym}_{i}^{l}}
    c_{\bm{\sumsym}_{i}^{l}}
    \bm{\sumsym}_{i}^{l}
    ,
\end{align} 
where the sum of $\bm{\sumsym}_{i}^{l}$ runs over all $l$-support basis elements starting from the $i$ th site, $c_{\bm{\sumsym}_{i}^{l}}$ is the coefficient of $\bm{\sumsym}_{i}^{l}$ which may depend on $U$, there exists a $k$-support basis element $\bm{\sumsym}_{i}^{k}$ such that $c_{\bm{\sumsym}_{i}^{k}} \neq 0$, and $F_k$ commutes with the Hamiltonian~\eqref{eq:Hamiltonian}: $\bck{F_k, H}=0$.

%%%%%%%%%%%%%%%%%%%%%%%%%%%%%%
When referring to a \textit{less-than-or-equal-to-k-local conserved quantity}, we exclude the constraint of the existence of $\bm{\sumsym}_{i}^{k}$ such that $c_{\bm{\sumsym}_{i}^{k}} \neq 0$ from the above definition.

%%%%%%%%%%%%%%%%%%%%%%%%%%%%%%
Let the $k$-local conserved quantity obtained in~\cite{fukai-hubbard-charge-2023} be denoted by $Q_k$, which has the following form:
%%%%%%%%%%%%%%%%%%%%%%%%%%%%%%
%$Q_k$ has the following form:
\begin{align}
    Q_k
    &
    =
    Q_k^0
    +
    \delta Q_{k-1}(U)
    ,
    \\
    Q_k^0
    &
    \equiv
    2
    \sum_{j=1}^L
    \sum_{\sigma=\uparrow,\downarrow}
    \paren{
        c_{j,\sigma} c_{j+k-1,\sigma}^\dag
        +(-1)^{k-1}
        c_{j,\sigma}^\dag c_{j+k-1,\sigma}
    }
    ,
\end{align}
where $\delta Q_{k-1}(U)$ is written as a linear combination of less-than-or-equal-to-$(k-1)$-support operators and depends on $U$.
%%%%%%%%%%%%%%%%%%%%%%%%%%%%%%
The exact form of $\delta Q_{k-1}(U)$ is explained in~\cite{fukai-hubbard-charge-2023}, which is not necessary for the following argument.
%%%%%%%%%%%%%%%%%%%%%%%%%%%%%%
$Q_k^0$ is the $k$-support operator in $Q_k$.

%%%%%%%%%%%%%%%%%%%%%%%%%%%%%%
Let a one-local conserved quantity be denoted by $Q_1$, which is not obtained through the expansion of the transfer matrix.
%%%%%%%%%%%%%%%%%%%%%%%%%%%%%%
$Q_1$ is a linear combination of the $\mathrm{SU}(2)$ charges and $\mathrm{U}(1)$ charge, and also the $\eta$-pairing charges for even $L$~\cite{PhysRevLett.63.2144, doi:10.1142/S0217984990000933}.
%%%%%%%%%%%%%%%%%%%%%%%%%%%%%%
We prove there are no other one-support conserved quantities in Appendix~\ref{app:one-sup}.

%%%%%%%%%%%%%%%%%%%%%%%%%%%%%%
Whether or not there exist local charges independent of $\bce{Q_k}_{k\geq 1}$ has been a mystery.
%%%%%%%%%%%%%%%%%%%%%%%%%%%%%%
We prove the family of the local conserved quantities $\{Q_k\}_{k\geq 1}$ is complete: there are no other local conserved quantities in the one-dimensional Hubbard model independent of $\{Q_k\}_{k\geq 1}$.
%%%%%%%%%%%%%%%%%%%%%%%%%%%%%%
The proof is based on the following theorem and its corollary:
%%%%%%%%%%%%%%%%%%%%%%%%%%%%%%
\begin{theorem}
\label{theorem}
Let $F_k$ be a $k$-local conserved quantity of the one-dimensional Hubbard model $(k < \floor{\frac{L}{2}})$.
%%%%%%%%%%%%%%%%%%%%%%%%%%%%%%
A constant $c_k (\neq 0)$ and a less-than-or-equal-to-$(k-1)$-local conserved quantity $\Delta_{k-1}$ exist such that $F_k = c_k Q_k + \Delta_{k-1}$.
\end{theorem}
%%%%%%%%%%%%%%%%%%%%%%%%%%%%%%
When we replace $F_k$ by a less-than-or-equal-to-$k$-local conserved quantity $F_k^\prime$, in Theorem~\ref{theorem}, $c_k=0$ is allowed.

%%%%%%%%%%%%%%%%%%%%%%%%%%%%%%
We have the following corollary immediately from Theorem~\ref{theorem}.
\begin{corollary}
\label{coro}
Let $F_k$ be a $k$-local conserved quantity of the one-dimensional Hubbard model $(k < \floor{\frac{L}{2}})$.
%%%%%%%%%%%%%%%%%%%%%%%%%%%%%%
A set of constants $\bce{c_l}_{1\leq l \leq k}$ exists such that $F_k = \sum_{l=1}^{k} c_l Q_l$.
\end{corollary}
\begin{proof}
%%%%%%%%%%%%%%%%%%%%%%%%%%%%%%
From Theorem~\ref{theorem}, $F_k$ is written as $F_k = c_k Q_k + F^\prime_{k-1}$, where $F^\prime_{k-1}$ is a less-than-or-equal-to-$(k-1)$-local conserved quantity.
%%%%%%%%%%%%%%%%%%%%%%%%%%%%%%
Using Theorem~\ref{theorem} again to $F^\prime_{k-1}$, we have $F_{k} = c_k Q_k + c_{k-1} Q_{k-1} + F^\prime_{k-2}$ where $F^\prime_{k-2}$ is a less-than-or-equal-to-$(k-2)$-local conserved quantity.
%%%%%%%%%%%%%%%%%%%%%%%%%%%%%%
In the same way, by repeatedly using Theorem~\ref{theorem}, we have $F_k = \sum_{l=1}^{k} c_l Q_l$ ($c_k\neq 0$).
\end{proof}

%%%%%%%%%%%%%%%%%%%%%%%%%%%%%%
From Corollary~\ref{coro}, all local conserved quantities are written as a linear combination of $\bce{Q_k}_{k\geq 1}$, and we can see the completeness of $\bce{Q_k}_{k\geq 1}$.
%%%%%%%%%%%%%%%%%%%%%%%%%%%%%%
We can also see that the local charges generated by the transfer matrix are written as a linear combination of $\bce{Q_k}_{k\geq 1}$.
%%%%%%%%%%%%%%%%%%%%%%%%%%%%%%
All we have to do is the proof of Theorem~\ref{theorem}.
%%%%%%%%%%%%%%%%%%%%%%%%%%%%%%
We prove Theorem~\ref{theorem} in the following sections.

\section{Notations}
\label{sec:notation}
In this section, we introduce useful notations in the proof of Theorem~\ref{theorem}.
%%%%%%%%%%%%%%%%%%%%%%%%%%%%%%
In the following, the symbol for the commutator has the additional factor $1/2$: $\bck{A, B} \equiv \frac{1}{2}\paren{AB-BA}$.

\subsection{Notation for the commutator with Hamiltonian}
%%%%%%%%%%%%%%%%%%%%%%%%%%%%%%
We introduce a notation for the $k$-support basis element:
\begin{align}
    \label{eq:notation-def}
    \begin{tikzpicture}[baseline=-0.5*\completevertical]
        \stringud{0}{$\Bar{a}_1$}{$\Bar{b}_1$}
        \stringud{1}{$\Bar{a}_2$}{$\Bar{b}_2$}
        \stringud{2}{$\cdots$}{$\cdots$}
        \stringud{3}{$\Bar{a}_k$}{$\Bar{b}_k$}
    \end{tikzpicture}
    (i)
    \equiv
    \bm{\sumsym}_{i}^{k}
    ,
\end{align}
where $\Bar{a}_l(\Bar{b}_l) \equiv \circled{-}, \circled{+}, \circled{z}, \circled{\phantom{\pm}}$ for $\sumsym_{i+l-1, \uparrow(\downarrow)} = c_{i+l-1, \uparrow(\downarrow)}, c^\dag_{i+l-1, \uparrow(\downarrow)}, z_{i+l-1, \uparrow(\downarrow)}, I$ respectively.
%%%%%%%%%%%%%%%%%%%%%%%%%%%%%%
We note that at least one of $\Bar{a}_1$ and $\Bar{b}_1$ should not be $I$, and the same is also true for $\Bar{a}_k$ and $\Bar{b}_k$.
%We note that at least one of $\Bar{a}_1$ or $\Bar{b}_1$ ($\Bar{a}_k$ or $\Bar{b}_k$) are not $\circled{}$.
%%%%%%%%%%%%%%%%%%%%%%%%%%%%%%
We give examples of $5$-support basis elements in this notation:
\begin{align}
    \begin{tikzpicture}[baseline=-0.5*\completevertical]
        \stringud{0}{\circled{+}}{\circled{}}
        \stringud{1}{\circled{-}}{\circled{}}
        \stringud{2}{\circled{}}{\circled{}}
        \stringud{3}{\circled{}}{\circled{+}}
        \stringud{4}{\circled{}}{\circled{-}}
    \end{tikzpicture}
    (i)
    &=
    c^\dag_{i,\uparrow}c_{i+1,\uparrow}
    c^\dag_{i+3,\downarrow}c_{i+4,\downarrow}
    ,
    \\
    \begin{tikzpicture}[baseline=-0.5*\completevertical]
        \stringud{0}{\circled{+}}{\circled{}}
        \stringud{1}{\circled{}}{\circled{-}}
        \stringud{2}{\circled{-}}{\circled{}}
        \stringud{3}{\circled{}}{\circled{+}}
        \stringud{4}{\circled{z}}{\circled{}}
    \end{tikzpicture}
    (i)
    &=
    c^\dag_{i,\uparrow}c_{i+2,\uparrow}z_{i+4, \uparrow}
    c_{i+1,\downarrow}c^\dag_{i+3,\downarrow}
    .
\end{align}
%%%%%%%%%%%%%%%%%%%%%%%%%%%%%%
We refer to $i$ in~\eqref{eq:notation-def} as the starting site and to the two-row sequence 
$
\begin{tikzpicture}[baseline=-0.5*\completevertical]
    \stringud{0}{$\Bar{a}_1$}{$\Bar{b}_1$}
    \stringud{1}{$\Bar{a}_2$}{$\Bar{b}_2$}
    \stringud{2}{$\cdots$}{$\cdots$}
    \stringud{3}{$\Bar{a}_k$}{$\Bar{b}_k$}
\end{tikzpicture}
$
as a $k$-support \textit{configuration}.

%%%%%%%%%%%%%%%%%%%%%%%%%%%%%%
We introduce a notation for the commutator with the hopping term in the Hamiltonian.
%%%%%%%%%%%%%%%%%%%%%%%%%%%%%%
We first denote the density of the hopping term as
\begin{align}
    \label{eq:hopping_dens}
    h_{l}^{\uparrow} &= 2(c_{l, \uparrow}c^\dag_{l+1, \uparrow} - c^\dag_{l, \uparrow}c_{l+1, \uparrow} )
    ,
    \\
    h_{l}^{\downarrow} &= 2(c_{l, \downarrow}c^\dag_{l+1, \downarrow} - c^\dag_{l, \downarrow}c_{l+1, \downarrow} )
    ,
\end{align}
and the hopping term is written as $H_0 = \sum_{\sigma \in \{\uparrow, \downarrow\}}\sum_{l=1}^{L} h_{l}^{\sigma}$.
%%%%%%%%%%%%%%%%%%%%%%%%%%%%%%
We represent the commutator of a basis element and $h_{l}^{\sigma}$ ($\sigma \in \{\uparrow, \downarrow\}$) by
\begin{align}
    \label{eq:commuhup}
    \begin{tikzpicture}[baseline=-0.5*\completevertical]
        \stringud{0}{$\Bar{a}_1$}{$\Bar{b}_1$}
        \stringud{1}{$\cdots$}{$\cdots$}
        \stringud{2}{$\Bar{a}_l$}{$\Bar{b}_l$}
        \stringud{3.3}{$\Bar{a}_{l+1}$}{$\Bar{b}_{l+1}$}
        \stringud{5}{$\cdots$}{$\cdots$}
        \stringud{6}{$\Bar{a}_k$}{$\Bar{b}_k$}
        \commutatorhopu{2}{3}
    \end{tikzpicture}
    (i)
    &
    \equiv
    %\frac{1}{2}
    \bck{
        \begin{tikzpicture}[baseline=-0.5*\completevertical]
            \stringud{0}{$\Bar{a}_1$}{$\Bar{b}_1$}
            \stringud{1}{$\cdots$}{$\cdots$}
            \stringud{2}{$\Bar{a}_l$}{$\Bar{b}_l$}
            \stringud{3.3}{$\Bar{a}_{l+1}$}{$\Bar{b}_{l+1}$}
            \stringud{5}{$\cdots$}{$\cdots$}
            \stringud{6}{$\Bar{a}_k$}{$\Bar{b}_k$}
        \end{tikzpicture}
        (i)
        ,
        h_{i+l-1}^{\uparrow}
    }
    ,
    \\
    \label{eq:commuhdown}
    \begin{tikzpicture}[baseline=-0.5*\completevertical]
        \stringud{0}{$\Bar{a}_1$}{$\Bar{b}_1$}
        \stringud{1}{$\cdots$}{$\cdots$}
        \stringud{2}{$\Bar{a}_l$}{$\Bar{b}_l$}
        \stringud{3.3}{$\Bar{a}_{l+1}$}{$\Bar{b}_{l+1}$}
        \stringud{5}{$\cdots$}{$\cdots$}
        \stringud{6}{$\Bar{a}_k$}{$\Bar{b}_k$}
        \commutatorhopd{2}{3}
    \end{tikzpicture}
    (i)
    &
    \equiv
    %\frac{1}{2}
    \bck{
        \begin{tikzpicture}[baseline=-0.5*\completevertical]
            \stringud{0}{$\Bar{a}_1$}{$\Bar{b}_1$}
            \stringud{1}{$\cdots$}{$\cdots$}
            \stringud{2}{$\Bar{a}_l$}{$\Bar{b}_l$}
            \stringud{3.3}{$\Bar{a}_{l+1}$}{$\Bar{b}_{l+1}$}
            \stringud{5}{$\cdots$}{$\cdots$}
            \stringud{6}{$\Bar{a}_k$}{$\Bar{b}_k$}
        \end{tikzpicture}
        (i)
        ,
        h_{i+l-1}^{\downarrow}
    }
    .
\end{align}
%%%%%%%%%%%%%%%%%%%%%%%%%%%%%%
When acting on the left or right end of the basis element, we write the action as
\begin{align}
    \label{eq:commuhuprightend}
    \begin{tikzpicture}[baseline=-0.5*\completevertical]
        \stringud{0}{$\Bar{a}_1$}{$\Bar{b}_1$}
        \stringud{1}{$\cdots$}{$\cdots$}
        \stringud{2}{$\Bar{a}_k$}{$\Bar{b}_k$}
        \stringud{3}{\circleddotted{}}{\circleddotted{}}
        \commutatorhopu{1.9}{3}
    \end{tikzpicture}
    (i)
    &
    =
    %\frac{1}{2}
    \bck{
        \begin{tikzpicture}[baseline=-0.5*\completevertical]
            \stringud{0}{$\Bar{a}_1$}{$\Bar{b}_1$}
            \stringud{1}{$\cdots$}{$\cdots$}
            \stringud{2}{$\Bar{a}_k$}{$\Bar{b}_k$}
            %\stringud{3}{\circled{}}{\circled{}}
            %\commutatorhopu{2}{3}
        \end{tikzpicture}
        (i)
        ,
        h_{i+k-1}^{\uparrow}
    }
    ,
    \\
    \label{eq:commuhupleftend}
    \begin{tikzpicture}[baseline=-0.5*\completevertical]
        \stringud{0}{$\Bar{a}_1$}{$\Bar{b}_1$}
        \stringud{1}{$\cdots$}{$\cdots$}
        \stringud{2}{$\Bar{a}_k$}{$\Bar{b}_k$}
        \stringud{-1}{\circleddotted{}}{\circleddotted{}}
        \commutatorhopu{-1}{-0.1}
    \end{tikzpicture}
    (i)
    &
    =
    %\frac{1}{2}
    \bck{
        \begin{tikzpicture}[baseline=-0.5*\completevertical]
            \stringud{0}{$\Bar{a}_1$}{$\Bar{b}_1$}
            \stringud{1}{$\cdots$}{$\cdots$}
            \stringud{2}{$\Bar{a}_k$}{$\Bar{b}_k$}
            %\stringud{3}{\circled{}}{\circled{}}
            %\commutatorhopu{2}{3}
        \end{tikzpicture}
        (i)
        ,
        h_{i-1}^{\uparrow}
    }
    ,
\end{align}
where we add the additional columns with $\circleddotted{}$ to the right or left end, and the support can be increased by one in this case.
%%%%%%%%%%%%%%%%%%%%%%%%%%%%%%
The actions of $h_{i-1}^{\downarrow}$ and $h_{i+k-1}^{\downarrow}$ are represented in the same manner.

%%%%%%%%%%%%%%%%%%%%%%%%%%%%%%
We show all the possible actions of~\eqref{eq:commuhup}:
\begin{align}
    \label{eq:commu-plus-b-up}
    \begin{tikzpicture}[baseline=0.5*\completevertical]
        \stringsu{0}{1}{$\cdots$}
        \stringsu{1}{2}{$\circled{\pm}$}
        \stringsu{2}{3}{$\circled{}$}
        \commutatorhopu{1}{2}
        \stringsu{3}{4}{$\cdots$}
    \end{tikzpicture}
    &
    =
    \pm
    \begin{tikzpicture}[baseline=0.5*\completevertical]
        \stringsu{0}{1}{$\cdots$}
        \stringsu{2}{3}{$\circled{\pm}$}
        \stringsu{1}{2}{$\circled{}$}
        %\commutatorhopu{1}{2}
        \stringsu{3}{4}{$\cdots$}
    \end{tikzpicture}
    ,
    \\
    \begin{tikzpicture}[baseline=0.5*\completevertical]
        \stringsu{0}{1}{$\cdots$}
        \stringsu{2}{3}{$\circled{\pm}$}
        \stringsu{1}{2}{$\circled{}$}
        \commutatorhopu{1}{2}
        \stringsu{3}{4}{$\cdots$}
    \end{tikzpicture}
    &
    =
    \pm
    \begin{tikzpicture}[baseline=0.5*\completevertical]
        \stringsu{0}{1}{$\cdots$}
        \stringsu{1}{2}{$\circled{\pm}$}
        \stringsu{2}{3}{$\circled{}$}
        %\commutatorhopu{1}{2}
        \stringsu{3}{4}{$\cdots$}
    \end{tikzpicture}
    ,
    \\
    \label{eq:commu-pm-mp}
    \begin{tikzpicture}[baseline=0.5*\completevertical]
        \stringsu{0}{1}{$\cdots$}
        \stringsu{1}{2}{$\circled{\pm}$}
        \stringsu{2}{3}{$\circled{\mp}$}
        \commutatorhopu{1}{2}
        \stringsu{3}{4}{$\cdots$}
    \end{tikzpicture}
    &
    =
    \frac{1}{2}
    \paren{
        \begin{tikzpicture}[baseline=0.5*\completevertical]
            \stringsu{0}{1}{$\cdots$}
            \stringsu{1}{2}{$\circled{}$}
            \stringsu{2}{3}{$\circled{z}$}
            %\commutatorhopu{1}{2}
            \stringsu{3}{4}{$\cdots$}
        \end{tikzpicture}
        -
        \begin{tikzpicture}[baseline=0.5*\completevertical]
            \stringsu{0}{1}{$\cdots$}
            \stringsu{1}{2}{$\circled{z}$}
            \stringsu{2}{3}{$\circled{}$}
            %\commutatorhopu{1}{2}
            \stringsu{3}{4}{$\cdots$}
        \end{tikzpicture}
    }
    ,
    \\
    \label{eq:commu-z-i}
    \begin{tikzpicture}[baseline=0.5*\completevertical]
        \stringsu{0}{1}{$\cdots$}
        \stringsu{1}{2}{$\circled{z}$}
        \stringsu{2}{3}{$\circled{}$}
        \commutatorhopu{1}{2}
        \stringsu{3}{4}{$\cdots$}
    \end{tikzpicture}
    &
    =
    -
    2
    \paren{
        \begin{tikzpicture}[baseline=0.5*\completevertical]
            \stringsu{0}{1}{$\cdots$}
            \stringsu{1}{2}{$\circled{+}$}
            \stringsu{2}{3}{$\circled{-}$}
            %\commutatorhopu{1}{2}
            \stringsu{3}{4}{$\cdots$}
        \end{tikzpicture}
        +
        \begin{tikzpicture}[baseline=0.5*\completevertical]
            \stringsu{0}{1}{$\cdots$}
            \stringsu{1}{2}{$\circled{-}$}
            \stringsu{2}{3}{$\circled{+}$}
            %\commutatorhopu{1}{2}
            \stringsu{3}{4}{$\cdots$}
        \end{tikzpicture} 
    }
    ,
    \\
    \label{eq:commu-i-z}
    \begin{tikzpicture}[baseline=0.5*\completevertical]
        \stringsu{0}{1}{$\cdots$}
        \stringsu{1}{2}{$\circled{}$}
        \stringsu{2}{3}{$\circled{z}$}
        \commutatorhopu{1}{2}
        \stringsu{3}{4}{$\cdots$}
    \end{tikzpicture}
    &
    =
    2
    \paren{
        \begin{tikzpicture}[baseline=0.5*\completevertical]
            \stringsu{0}{1}{$\cdots$}
            \stringsu{1}{2}{$\circled{+}$}
            \stringsu{2}{3}{$\circled{-}$}
            %\commutatorhopu{1}{2}
            \stringsu{3}{4}{$\cdots$}
        \end{tikzpicture}
        +
        \begin{tikzpicture}[baseline=0.5*\completevertical]
            \stringsu{0}{1}{$\cdots$}
            \stringsu{1}{2}{$\circled{-}$}
            \stringsu{2}{3}{$\circled{+}$}
            %\commutatorhopu{1}{2}
            \stringsu{3}{4}{$\cdots$}
        \end{tikzpicture}
    }
    ,
    \\
    \begin{tikzpicture}[baseline=0.5*\completevertical]
        \stringsu{0}{1}{$\cdots$}
        \stringsu{1}{2}{$\circled{z}$}
        \stringsu{2}{3}{$\circled{\pm}$}
        \commutatorhopu{1}{2}
        \stringsu{3}{4}{$\cdots$}
    \end{tikzpicture}
    &
    =
    \pm
    \begin{tikzpicture}[baseline=0.5*\completevertical]
        \stringsu{0}{1}{$\cdots$}
        \stringsu{1}{2}{$\circled{\pm}$}
        \stringsu{2}{3}{$\circled{z}$}
        %\commutatorhopu{1}{2}
        \stringsu{3}{4}{$\cdots$}
    \end{tikzpicture}
    ,
     \\
    \begin{tikzpicture}[baseline=0.5*\completevertical]
        \stringsu{0}{1}{$\cdots$}
        \stringsu{1}{2}{$\circled{\pm}$}
        \stringsu{2}{3}{$\circled{z}$}
        \commutatorhopu{1}{2}
        \stringsu{3}{4}{$\cdots$}
    \end{tikzpicture}
    &
    =
    \pm
    \begin{tikzpicture}[baseline=0.5*\completevertical]
        \stringsu{0}{1}{$\cdots$}
        \stringsu{1}{2}{$\circled{z}$}
        \stringsu{2}{3}{$\circled{\pm}$}
        %\commutatorhopu{1}{2}
        \stringsu{3}{4}{$\cdots$}
    \end{tikzpicture}
    ,
    \label{eq:commu-plus-b-up-end}
    \\
    \begin{tikzpicture}[baseline=0.5*\completevertical]
        \stringsu{0}{1}{$\cdots$}
        \stringsu{1}{2}{$\circled{}$}
        \stringsu{2}{3}{$\circled{}$}
        \commutatorhopu{1}{2}
        \stringsu{3}{4}{$\cdots$}
    \end{tikzpicture}
    &
    =
    \begin{tikzpicture}[baseline=0.5*\completevertical]
        \stringsu{0}{1}{$\cdots$}
        \stringsu{1}{2}{$\circled{\pm}$}
        \stringsu{2}{3}{$\circled{\pm}$}
        \commutatorhopu{1}{2}
        \stringsu{3}{4}{$\cdots$}
    \end{tikzpicture}
    =
     \begin{tikzpicture}[baseline=0.5*\completevertical]
        \stringsu{0}{1}{$\cdots$}
        \stringsu{1}{2}{$\circled{z}$}
        \stringsu{2}{3}{$\circled{z}$}
        \commutatorhopu{1}{2}
        \stringsu{3}{4}{$\cdots$}
    \end{tikzpicture}
    =0
    ,
\end{align}
where we omit the lower row and the starting site.
%%%%%%%%%%%%%%%%%%%%%%%%%%%%%%
%The action of~\eqref{eq:commuhdown} is also represented in the same way.
%%%%%%%%%%%%%%%%%%%%%%%%%%%%%%
When evaluating~\eqref{eq:commuhdown} case-by-case, the obtained results are the same.
%%%%%%%%%%%%%%%%%%%%%%%%%%%%%%
$\circleddotted{}$ in~\eqref{eq:commuhuprightend} and~\eqref{eq:commuhupleftend} is treated in the same manner as $\circled{}$ in the above actions.

%%%%%%%%%%%%%%%%%%%%%%%%%%%%%%
We also introduce a notation for the commutator with the interaction term $H_{\mathrm{int}}$.
%%%%%%%%%%%%%%%%%%%%%%%%%%%%%%
We first denote the density of the interaction term as
\begin{align}
    \label{eq:int_dens}
    h_{l}^{\mathrm{int}} = U z_{l, \uparrow}z_{l, \downarrow}
    ,
\end{align}
and the interaction term is written as
$H_{\mathrm{int}} 
=
\sum_{l=1}^{L} h_{l}^{\mathrm{int}}$.
%%%%%%%%%%%%%%%%%%%%%%%%%%%%%%
We represent the commutator of a basis element and the density of the interaction term by
\begin{align}
    \label{eq:commutatorint}
    \begin{tikzpicture}[baseline=-0.5*\completevertical]
        \stringud{0}{$\Bar{a}_1$}{$\Bar{b}_1$}
        \stringud{1}{$\cdots$}{$\cdots$}
        \stringud{2}{$\Bar{a}_l$}{$\Bar{b}_l$}
        \stringud{3}{$\cdots$}{$\cdots$}
        \stringud{4}{$\Bar{a}_k$}{$\Bar{b}_k$}
        \commutatorint{2}
    \end{tikzpicture}
    (i)
    \equiv
    %\frac{1}{2}
    \bck{
        \begin{tikzpicture}[baseline=-0.5*\completevertical]
            \stringud{0}{$\Bar{a}_1$}{$\Bar{b}_1$}
            \stringud{1}{$\cdots$}{$\cdots$}
            \stringud{2}{$\Bar{a}_l$}{$\Bar{b}_l$}
            \stringud{3}{$\cdots$}{$\cdots$}
            \stringud{4}{$\Bar{a}_k$}{$\Bar{b}_k$}
        \end{tikzpicture}
        (i)
        ,
        h_{i+l-1}^{\mathrm{int}}
    }
    .
\end{align}
%%%%%%%%%%%%%%%%%%%%%%%%%%%%%%
We show the non-zero action of~\eqref{eq:commutatorint}:
\begin{align}
    \label{eq:commutator-int-start}
    \begin{tikzpicture}[baseline=-0.5*\completevertical]
        \stringsud{0}{1}{$\cdots$}{$\cdots$}
        \stringsud{1}{2}{$\circled{\pm}$}{$\circled{}$}
        \commutatorint{1}
        \stringsud{2}{3}{$\cdots$}{$\cdots$}
    \end{tikzpicture}
    &
    =
    \mp U
    \begin{tikzpicture}[baseline=-0.5*\completevertical]
        \stringsud{0}{1}{$\cdots$}{$\cdots$}
        \stringsud{1}{2}{$\circled{\pm}$}{$\circled{z}$}
        %\commutatorint{1}
        \stringsud{2}{3}{$\cdots$}{$\cdots$}
    \end{tikzpicture}
    ,
    \quad\quad
    \begin{tikzpicture}[baseline=-0.5*\completevertical]
        \stringsud{0}{1}{$\cdots$}{$\cdots$}
        \stringsud{1}{2}{$\circled{}$}{$\circled{\pm}$}
        \commutatorint{1}
        \stringsud{2}{3}{$\cdots$}{$\cdots$}
    \end{tikzpicture}
    =
    \mp U
    \begin{tikzpicture}[baseline=-0.5*\completevertical]
        \stringsud{0}{1}{$\cdots$}{$\cdots$}
        \stringsud{1}{2}{$\circled{z}$}{$\circled{\pm}$}
        %\commutatorint{1}
        \stringsud{2}{3}{$\cdots$}{$\cdots$}
    \end{tikzpicture}
    ,
    \\
    \begin{tikzpicture}[baseline=-0.5*\completevertical]
        \stringsud{0}{1}{$\cdots$}{$\cdots$}
        \stringsud{1}{2}{$\circled{\pm}$}{$\circled{z}$}
        \commutatorint{1}
        \stringsud{2}{3}{$\cdots$}{$\cdots$}
    \end{tikzpicture}
    &
    =
    \mp U
    \begin{tikzpicture}[baseline=-0.5*\completevertical]
        \stringsud{0}{1}{$\cdots$}{$\cdots$}
        \stringsud{1}{2}{$\circled{\pm}$}{$\circled{}$}
        %\commutatorint{1}
        \stringsud{2}{3}{$\cdots$}{$\cdots$}
    \end{tikzpicture}
    ,
    \quad\quad
    \begin{tikzpicture}[baseline=-0.5*\completevertical]
        \stringsud{0}{1}{$\cdots$}{$\cdots$}
        \stringsud{1}{2}{$\circled{z}$}{$\circled{\pm}$}
        \commutatorint{1}
        \stringsud{2}{3}{$\cdots$}{$\cdots$}
    \end{tikzpicture}
    =
    \mp U
    \begin{tikzpicture}[baseline=-0.5*\completevertical]
        \stringsud{0}{1}{$\cdots$}{$\cdots$}
        \stringsud{1}{2}{$\circled{}$}{$\circled{\pm}$}
        %\commutatorint{1}
        \stringsud{2}{3}{$\cdots$}{$\cdots$}
    \end{tikzpicture}
    ,
    \label{eq:commutator-int-end}
\end{align}
and all other cases are zero.
%%%%%%%%%%%%%%%%%%%%%%%%%%%%%%
The action of~\eqref{eq:commutatorint} does not change the support.

%%%%%%%%%%%%%%%%%%%%%%%%%%%%%%
The commutator of a basis element and the Hamiltonian can be calculated by using the notations introduced above, for example:
\begin{align}
    \bck{
        \begin{tikzpicture}[baseline=-0.5*\completevertical]
            \stringud{0}{\circled{+}}{\circled{}}
            \stringud{1}{\circled{-}}{\circled{z}}
        \end{tikzpicture}
        (i)
        ,
        H
    }
    =&
    \begin{tikzpicture}[baseline=-0.5*\completevertical]
        \stringud{-1}{\circleddotted{}}{\circleddotted{}}
        \stringud{0}{\circled{+}}{\circled{}}
        \stringud{1}{\circled{-}}{\circled{z}}
        \commutatorhopu{-1}{0}
    \end{tikzpicture}
    (i) 
    +
    \begin{tikzpicture}[baseline=-0.5*\completevertical]
        \stringud{0}{\circled{+}}{\circled{}}
        \stringud{1}{\circled{-}}{\circled{z}}
        \commutatorhopu{0}{1}
    \end{tikzpicture}
    (i)
    +
    \begin{tikzpicture}[baseline=-0.5*\completevertical]
        \stringud{0}{\circled{+}}{\circled{}}
        \stringud{1}{\circled{-}}{\circled{z}}
        \stringud{2}{\circleddotted{}}{\circleddotted{}}
        \commutatorhopu{1}{2}
    \end{tikzpicture}
    (i)
    +
    \begin{tikzpicture}[baseline=-0.5*\completevertical]
        \stringud{0}{\circled{+}}{\circled{}}
        \stringud{1}{\circled{-}}{\circled{z}}
        %\stringud{2}{\circleddotted{}}{\circleddotted{}}
        \commutatorhopd{0}{1}
    \end{tikzpicture}
    (i)
    +
    \begin{tikzpicture}[baseline=-0.5*\completevertical]
        \stringud{0}{\circled{+}}{\circled{}}
        \stringud{1}{\circled{-}}{\circled{z}}
        \stringud{2}{\circleddotted{}}{\circleddotted{}}
        \commutatorhopd{1}{2}
    \end{tikzpicture}
    (i)
    \nonumber\\
    &
    +
    \begin{tikzpicture}[baseline=-0.5*\completevertical]
        \stringud{0}{\circled{+}}{\circled{}}
        \stringud{1}{\circled{-}}{\circled{z}}
        \commutatorint{0}
    \end{tikzpicture}
    (i)
    +
    \begin{tikzpicture}[baseline=-0.5*\completevertical]
        \stringud{0}{\circled{+}}{\circled{}}
        \stringud{1}{\circled{-}}{\circled{z}}
        \commutatorint{1}
    \end{tikzpicture}
    (i)
    ,
\end{align}
where each term on the right-hand side is calculated as follows:
\begin{align}
    \begin{tikzpicture}[baseline=-0.5*\completevertical]
        \stringud{-1}{\circleddotted{}}{\circleddotted{}}
        \stringud{0}{\circled{+}}{\circled{}}
        \stringud{1}{\circled{-}}{\circled{z}}
        \commutatorhopu{-1}{0}
    \end{tikzpicture}
    (i) 
    &
    =
    \begin{tikzpicture}[baseline=-0.5*\completevertical]
        \stringud{-1}{\circled{+}}{\circled{}}
        \stringud{0}{\circled{}}{\circled{}}
        \stringud{1}{\circled{-}}{\circled{z}}
    \end{tikzpicture}
    (i-1)
    ,
    \\
    \begin{tikzpicture}[baseline=-0.5*\completevertical]
        \stringud{0}{\circled{+}}{\circled{}}
        \stringud{1}{\circled{-}}{\circled{z}}
        \commutatorhopu{0}{1}
    \end{tikzpicture}
    (i)
    &
    =
    \frac{1}{2}
    \paren{
        \begin{tikzpicture}[baseline=-0.5*\completevertical]
            \stringud{1}{\circled{z}}{\circled{z}}
        \end{tikzpicture}
        (i+1)
        -
        \begin{tikzpicture}[baseline=-0.5*\completevertical]
            \stringud{0}{\circled{z}}{\circled{}}
            \stringud{1}{\circled{}}{\circled{z}}
        \end{tikzpicture}
        (i)
    }
    ,
    \\
    \begin{tikzpicture}[baseline=-0.5*\completevertical]
        \stringud{0}{\circled{+}}{\circled{}}
        \stringud{1}{\circled{-}}{\circled{z}}
        \stringud{2}{\circleddotted{}}{\circleddotted{}}
        \commutatorhopu{1}{2}
    \end{tikzpicture}
    (i)
    &=
    \begin{tikzpicture}[baseline=-0.5*\completevertical]
        \stringud{0}{\circled{+}}{\circled{}}
        \stringud{1}{\circled{}}{\circled{z}}
        \stringud{2}{\circled{-}}{\circled{}}
        %\commutatorhopu{1}{2}
    \end{tikzpicture}
    \times (-1)
    ,
    \\
    \begin{tikzpicture}[baseline=-0.5*\completevertical]
        \stringud{0}{\circled{+}}{\circled{}}
        \stringud{1}{\circled{-}}{\circled{z}}
        %\stringud{2}{\circleddotted{}}{\circleddotted{}}
        \commutatorhopd{0}{1}
    \end{tikzpicture}
    (i)
    &=
    2
    \paren{
        \begin{tikzpicture}[baseline=-0.5*\completevertical]
            \stringud{0}{\circled{+}}{\circled{+}}
            \stringud{1}{\circled{-}}{\circled{-}}
            %\stringud{2}{\circleddotted{}}{\circleddotted{}}
            %\commutatorhopd{0}{1}
        \end{tikzpicture}
        +
        \begin{tikzpicture}[baseline=-0.5*\completevertical]
            \stringud{0}{\circled{+}}{\circled{-}}
            \stringud{1}{\circled{-}}{\circled{+}}
            %\stringud{2}{\circleddotted{}}{\circleddotted{}}
            %\commutatorhopd{0}{1}
        \end{tikzpicture}
    }
    ,
    \\
    \begin{tikzpicture}[baseline=-0.5*\completevertical]
        \stringud{0}{\circled{+}}{\circled{}}
        \stringud{1}{\circled{-}}{\circled{z}}
        \stringud{2}{\circleddotted{}}{\circleddotted{}}
        \commutatorhopd{1}{2}
    \end{tikzpicture}
    (i)
    &=
    -2
    \paren{
        \begin{tikzpicture}[baseline=-0.5*\completevertical]
            \stringud{0}{\circled{+}}{\circled{}}
            \stringud{1}{\circled{-}}{\circled{+}}
            \stringud{2}{\circled{}}{\circled{-}}
            %\commutatorhopd{1}{2}
        \end{tikzpicture}
        (i)
        +
        \begin{tikzpicture}[baseline=-0.5*\completevertical]
            \stringud{0}{\circled{+}}{\circled{}}
            \stringud{1}{\circled{-}}{\circled{-}}
            \stringud{2}{\circled{}}{\circled{+}}
            %\commutatorhopd{1}{2}
        \end{tikzpicture}
        (i)
    }
    ,
    \\
    \begin{tikzpicture}[baseline=-0.5*\completevertical]
        \stringud{0}{\circled{+}}{\circled{}}
        \stringud{1}{\circled{-}}{\circled{z}}
        \commutatorint{0}
    \end{tikzpicture}
    (i)
    &=
    \begin{tikzpicture}[baseline=-0.5*\completevertical]
        \stringud{0}{\circled{+}}{\circled{z}}
        \stringud{1}{\circled{-}}{\circled{z}}
        %\commutatorint{0}
    \end{tikzpicture}
    \times (-U)
    ,
    \\
    \begin{tikzpicture}[baseline=-0.5*\completevertical]
        \stringud{0}{\circled{+}}{\circled{}}
        \stringud{1}{\circled{-}}{\circled{z}}
        \commutatorint{1}
    \end{tikzpicture}
    (i)
    &=
    \begin{tikzpicture}[baseline=-0.5*\completevertical]
        \stringud{0}{\circled{+}}{\circled{}}
        \stringud{1}{\circled{-}}{\circled{}}
        %\commutatorint{1}
    \end{tikzpicture}
    (i)
    \times U
    .
\end{align}

\subsection{Graphical notation for cancellation}
%%%%%%%%%%%%%%%%%%%%%%%%%%%%%%
We assume $F_{k}$ is a $k$-local conserved quantity of the one-dimensional Hubbard model.
%%%%%%%%%%%%%%%%%%%%%%%%%%%%%%
Configurations are denoted by the symbol $q$ or $p$ in the following.
%%%%%%%%%%%%%%%%%%%%%%%%%%%%%%
$F_k$ is written as a linear combination of less-than-or-equal-to-$k$-support basis elements: 
\begin{align}
    F_k
    =
    \sum_{l=1}^{k} 
    \sum_{q \in \mathcal{C}_l}
    \sum_{i=1}^{L}
    c_i(q)
    q(i)
    ,
\end{align}
where $\mathcal{C}_k$ denotes the set of all $k$-support configurations, and $c_i(q)$ is the coefficients of $q(i)$.
%%%%%%%%%%%%%%%%%%%%%%%%%%%%%%%%%
Because $F_k$ is a conserved quantity, $F_k$ commutes with the Hamiltonian $\bck{F_k, H} = 0$, which gives the equations for $c_i(q)$.

%%%%%%%%%%%%%%%%%%%%%%%%%%%%%%%%%
We represent the commutator $\bck{F_k, H}$ by
\begin{align}
    \label{eq:commu-and-d}
    \bck{F_k, H} 
    =
    \sum_{l=1}^{k+1} 
    \sum_{\widetilde{q} \in \mathcal{C}_l}
    \sum_{i=1}^{L}
    d_{i}(\widetilde{q}) \widetilde{q}(i)
    ,
\end{align}
where $\widetilde{q}(i)$ is a basis element starting from the $i$ th site, and $d_{i}(\widetilde{q})$ is a linear combination of $\{c_i(q)\}$, which is determined from the commutation relation in~\eqref{eq:commu-plus-b-up}--\eqref{eq:commu-plus-b-up-end} and~\eqref{eq:commutator-int-start} and~\eqref{eq:commutator-int-end}.
%%%%%%%%%%%%%%%%%%%%%%%%%%%%%%%%%
The upper limit of the summation over $l$ in the right-hand side of~\eqref{eq:commu-and-d} is $k+1$ because the maximum support of the basis elements that can be generated by the commutator of a $k$-local charge and the Hamiltonian is $k+1$.
%%%%%%%%%%%%%%%%%%%%%%%%%%%%%%%%%
To ensure the conservation of $F_k$, it is necessary that $d_{i}(\widetilde{q}) = 0$ for all possible $\widetilde{q}$ and $i$. 
%%%%%%%%%%%%%%%%%%%%%%%%%%%%%%%%%
This yields the equations which $\{c_i(q)\}$ must satisfy.

%%%%%%%%%%%%%%%%%%%%%%%%%%%%%%%%%
We consider the cancellation of $\widetilde{q}(i)$ in $\bck{F_k, H}$~\eqref{eq:commu-and-d}, i.e., the equation $d_{i}(\widetilde{q}) = 0$.
%%%%%%%%%%%%%%%%%%%%%%%%%%%%%%%%%
We first explain how we obtain $d_{i}(\widetilde{q})$ in terms of $\{c_i(q)\}$.
%%%%%%%%%%%%%%%%%%%%%%%%%%%%%%%%%
Let $q_{1}(i_1),\ldots,q_m(i_m)$ be the basis elements in $F_k$ such that $\bck{q_{l}(i_l), H}$, for $l = 1,\ldots, m$, generate $\widetilde{q}(i)$ and there is no other contribution to $d_{i}(\widetilde{q})$.
%%%%%%%%%%%%%%%%%%%%%%%%%%%%%%%%%
%There exist $s_l\in\bce{\uparrow, \downarrow, \mathrm{int}}$ and $v_l\in\bce{1,\ldots,L}$ such that $\bck{q_{l}(i_l), h_{v_l}^{s_l}} = f_{l} \widetilde{q}(i) + (\text{rest})$ where $f_{l}$ is the non-zero factor determined from the commutation relation in~\eqref{eq:commu-plus-b-up}--\eqref{eq:commu-plus-b-up-end} and~\eqref{eq:commutator-int-start} and~\eqref{eq:commutator-int-end}.
%%%%%%%%%%%%%%%%%%%%%%%%%%%%%%%%%
There exists only one $h_{v_l}^{s_l}$~($s_l\in\bce{\uparrow, \downarrow, \mathrm{int}}$, $v_l\in\bce{1,\ldots,L}$) for $q_{l}(i_l)$ that satisfies
\begin{align}
    \bck{q_{l}(i_l), h_{v_l}^{s_l}} = f_{l} \widetilde{q}(i) + (\text{rest})
    ,    
\end{align}
where $f_{l}$ is the non-zero factor determined from the commutation relation in~\eqref{eq:commu-plus-b-up}--\eqref{eq:commu-plus-b-up-end} and~\eqref{eq:commutator-int-start} and~\eqref{eq:commutator-int-end}.
%%%%%%%%%%%%%%%%%%%%%%%%%%%%%%%%%
%The $h_{v_l}^{s_l}$ for $q_{l}(i_l)$ in the cancellation of $\widetilde{q}(i)$ is unique because 
%%%%%%%%%%%%%%%%%%%%%%%%%%%%%%%%%
$(\text{rest})$ is the other basis element generated by the commutator in the case of~\eqref{eq:commu-pm-mp}--\eqref{eq:commu-i-z}, where the commutator generates two basis elements, and $(\text{rest}) = 0$ for the other cases.
%%%%%%%%%%%%%%%%%%%%%%%%%%%%%%%%%
Then, we have
\begin{align}
    d_{i}(\widetilde{q})
    =
    \sum_{l=1}^{m} 
    f_l c_{i_l}(q_l)    
    ,
\end{align}
and the equation for the cancellation of $\widetilde{q}(i)$ becomes: 
\begin{align}
    \label{eq:cancellation-general}
    \sum_{l=1}^{m} 
    f_l c_{i_l}(q_l)
    =0
    .
\end{align}
%%%%%%%%%%%%%%%%%%%%%%%%%%%%%%%%%
%We next introduce the graphical notation that represents the cancellation~\eqref{eq:cancellation-general} as
%%%%%%%%%%%%%%%%%%%%%%%%%%%%%%%%%
The cancellation condition~\eqref{eq:cancellation-general} follows from, and will be associated with, the following diagram, which encodes the basis elements $q_{1}(i_1),\ldots,q_m(i_m)$ that yield $\widetilde{q}(i)$ after commutation with the Hamiltonian:
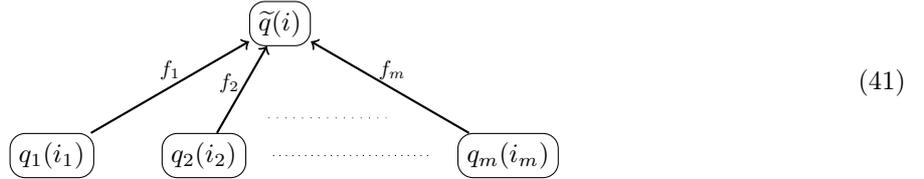
\begin{align}
    \label{eq:graphical-cancellation-general}
    \begin{tikzpicture}[baseline=(current bounding box.center)]
        % First TikZ picture in a node
        \node[draw, rectangle, rounded corners=2mm] (A1) at (0,0) {
            %$\bck{q_{1}(i_1), h_{v_1}^{s_1}}$
            $q_{1}(i_1)$
        };
        \node[draw, rectangle, rounded corners=2mm] (A2) at (2,0) {
            %$\bck{q_{2}(i_2), h_{v_2}^{s_2}}$
            $q_{2}(i_2)$
        };
        \node[draw, rectangle, rounded corners=2mm] (A3) at (6,0) {
            %$\bck{q_{m}(i_m), h_{v_m}^{s_m}}$
            $q_{m}(i_m)$
        };
        % Second TikZ picture in a node
        \node[draw, rectangle, rounded corners=2mm] (B) at (3,1.75) {
            $\widetilde{q}(i)$
        };
        % Arrow from first node to the second
        \draw[->,thick] (A1) -- (B) node[midway,above, scale=0.8] {$f_1$};
        \draw[->,thick] (A2) -- (B) node[midway,above, scale=0.8, xshift=-6pt, yshift=-5pt] {$f_2$};
        \draw[->,thick] (A3) -- (B) node[midway,above, scale=0.8,xshift=1pt] {$f_m$};
        \draw[dotted, shorten >=1em, shorten <=1em] (A2) -- (A3);
        \draw[dash pattern=on 0.1mm off 1mm, shorten >=1em, shorten <=1em] (2.5,0.5) -- (4.75,0.5);
    \end{tikzpicture}
    \quad
\end{align}
where the nodes at the starting point of arrows denote the operators that contribute the cancellation of $\widetilde{q}(i)$, and the factor $f_l$ is depicted in the middle of the arrows.
%%%%%%%%%%%%%%%%%%%%%%%%%%%%%%%%%
The $f_l$ in~\eqref{eq:graphical-cancellation-general} is often omitted in the following.

%%%%%%%%%%%%%%%%%%%%%%%%%%%%%%%%%
In the case of $m=1$,~\eqref{eq:cancellation-general} becomes $f_1 c_{i_1}(q_1)=0$, and we have $c_{i_1}(q_1)=0$ because of $f_1 \neq 0$. 
%%%%%%%%%%%%%%%%%%%%%%%%%%%%%%%%%
The $m=1$ case is represented graphically as
%%%%%%%%%%%%%%%%%%%%%%%%%%%%%%%%%
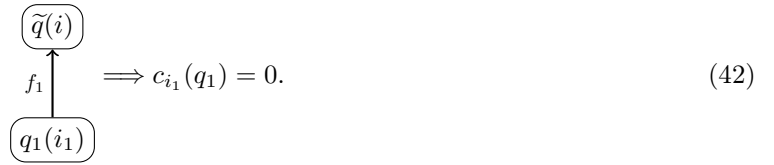
\begin{align}
    \label{eq:graphical-cancellation-one-m}
    \begin{tikzpicture}[baseline=(current bounding box.center)]
        % First TikZ picture in a node
        \node[draw, rectangle, rounded corners=2mm] (A1) at (0,0) {
            %$\bck{q_{1}(i_1), h_{v_1}^{s_1}}$
            $q_{1}(i_1)$
        };
        % Second TikZ picture in a node
        \node[draw, rectangle, rounded corners=2mm] (B) at (0,1.5) {
            $\widetilde{q}(i)$
        };
        % Arrow from first node to the second
        \draw[->,thick] (A1) -- (B) node[midway,left, scale=0.8] {$f_1$};
    \end{tikzpicture}
    \Longrightarrow
    c_{i_1}(q_1)=0.
\end{align}
%%%%%%%%%%%%%%%%%%%%%%%%%%%%%%%%%
In the case of  $m=2$,~\eqref{eq:cancellation-general} becomes $f_1 c_{i_1}(q_1) + f_2 c_{i_2}(q_2)=0$, and we have $c_{i_1}(q_1) = f c_{i_2}(q_2)$ ($f=-f_2/f_1$) because of $f_1 ,f_2 \neq 0$. %%%%%%%%%%%%%%%%%%%%%%%%%%%%%%%%%
The $m=2$ case is represented graphically as
\begin{align}
    \label{eq:graphical-cancellation-two-m}
    \begin{tikzpicture}[baseline=(current bounding box.center)]
        % First TikZ picture in a node
        \node[draw, rectangle, rounded corners=2mm] (A1) at (0,0) {
            %$\bck{q_{1}(i_1), h_{v_1}^{s_1}}$
            $q_{1}(i_1)$
        };
        \node[draw, rectangle, rounded corners=2mm] (A2) at (4,0) {
            %$\bck{q_{1}(i_1), h_{v_1}^{s_1}}$
            $q_{2}(i_2)$
        };
        % Second TikZ picture in a node
        \node[draw, rectangle, rounded corners=2mm] (B) at (2,1.25) {
            $\widetilde{q}(i)$
        };
        % Arrow from first node to the second
        \draw[->,thick] (A1) -- (B) node[midway,above, scale=0.8, xshift=-7pt, yshift=-4pt] {$f_1$};
        \draw[->,thick] (A2) -- (B) node[midway,above, scale=0.8, xshift=7pt, yshift=-4pt] {$f_2$};
    \end{tikzpicture}
    \Longrightarrow
    c_{i_1}(q_1) = f c_{i_2}(q_2)
    .
\end{align} 
%%%%%%%%%%%%%%%%%%%%%%%%%%%%%%%%%
In the following, we represent an equation for the cancellation of the $m=2$ case $c_{i_1}(q_1) = f c_{i_2}(q_2)$ simply as $c_{i_1}(q_1) \propto c_{i_2}(q_2)$.
%%%%%%%%%%%%%%%%%%%%%%%%%%%%%%%%%
Note that if $c_{i_1}(q_1) \propto c_{i_2}(q_2)$ and $c_{i_2}(q_2)=0$ holds, then we have $c_{i_1}(q_1)=0$.

%%%%%%%%%%%%%%%%%%%%%%%%%%%%%%%%%
We also introduce an additional notation that represents the node in the graphical notation:
\begin{align}
    \begin{tikzpicture}[baseline=-0.25\completevertical]
        % First TikZ picture in a node
        \node[draw, rectangle, rounded corners=2mm] (A1) at (0,0) {
            %$\bck{q_{1}(i_1), h_{v_1}^{s_1}}$
            \begin{tikzpicture}[baseline=-0.5*\completevertical]
                \stringud{0}{$\Bar{a}_1$}{$\Bar{b}_1$}
                \stringud{1}{$\Bar{a}_2$}{$\Bar{b}_2$}
                \stringud{2}{$\cdots$}{$\cdots$}
                \stringud{3}{$\Bar{a}_k$}{$\Bar{b}_k$}
            \end{tikzpicture}
        };
        \node[above=0em  of A1, scale=1] {$q(i)$};
    \end{tikzpicture}
    \equiv
    \begin{tikzpicture}[baseline=(current bounding box.center)]
        % First TikZ picture in a node
        \node[draw, rectangle, rounded corners=2mm] (A1) at (0,0) {
            %$\bck{q_{1}(i_1), h_{v_1}^{s_1}}$
            \begin{tikzpicture}[baseline=-0.5*\completevertical]
                \stringud{0}{$\Bar{a}_1$}{$\Bar{b}_1$}
                \stringud{1}{$\Bar{a}_2$}{$\Bar{b}_2$}
                \stringud{2}{$\cdots$}{$\cdots$}
                \stringud{3}{$\Bar{a}_k$}{$\Bar{b}_k$}
            \end{tikzpicture}
            (i)
        };
    \end{tikzpicture}
    \quad
\end{align}
where $q$ represents the configuration:
$
q = 
\begin{tikzpicture}[baseline=-0.5*\completevertical]
    \stringud{0}{$\Bar{a}_1$}{$\Bar{b}_1$}
    \stringud{1}{$\Bar{a}_2$}{$\Bar{b}_2$}
    \stringud{2}{$\cdots$}{$\cdots$}
    \stringud{3}{$\Bar{a}_k$}{$\Bar{b}_k$}
\end{tikzpicture}
$.

\section{Proof of completeness of $Q_k$}
\label{sec:proof}
%%%%%%%%%%%%%%%%%%%%%%%%%%%%%%%%%
In this section, we prove Theorem~\ref{theorem}.
%%%%%%%%%%%%%%%%%%%%%%%%%%%%%%%%%
We write $k$-local conserved quantity $F_k$ as
\begin{align}
    \label{eq:F_k_decomp}
    F_k
    =
    F_k^{k}
    +
    F_k^{k-1}
    +
    \text{(rest)}
    ,
\end{align}
where $F_k^{k}$ is the $k$-support operator in $F_k$, and $F_k^{k-1}$ is the $(k-1)$-support operator in $F_k$, and $\text{(rest)}$ is a linear combination of the less-than-or-equal-to-$(k-2)$-support operators in $F_k$.
%%%%%%%%%%%%%%%%%%%%%%%%%%%%%%%%%
For the proof of Theorem~\ref{theorem}, it suffices to consider only $F_k^{k}$ and $F_k^{k-1}$.
%%%%%%%%%%%%%%%%%%%%%%%%%%%%%%%%%
We determine $F^k_k$ so that $F_k$ satisfies $\bck{F_k, H} =0$.

%%%%%%%%%%%%%%%%%%%%%%%%%%%%%%%%%
The cancellation of the operator in $\bck{F_k, H} = \bck{F_k, H_0} +  \bck{F_k, H_{\mathrm{int}}}$ occurs for each support.
%%%%%%%%%%%%%%%%%%%%%%%%%%%%%%%%%
The commutator with the $H_0$ increases or decreases the support by one or does not change the support:
$
    \bck{
        A^{l}
        ,
        H_0
    }
    =
    \bck{
        A^{l}
        ,
        H_0
    }
    \big|_{l+1}
    +
    \bck{
        A^{l}
        ,
        H_0
    }
    \big|_{l}
    +
    \bck{
        A^{l}
        ,
        H_0
    }
    \big|_{l-1}
$
where $A^{l}$ is a $l$-support operator, and $\mathcal{O}|_{l}$ denotes the $l$-support operator in $\mathcal{O}$.
%%%%%%%%%%%%%%%%%%%%%%%%%%%%%%%%%
The commutator with $H_{\mathrm{int}}$ does not change the support:
$
    \bck{A^{l}, H_{\mathrm{int}}}=\bck{A^{l} ,H_{\mathrm{int}}}\big|_{l}
$.

%%%%%%%%%%%%%%%%%%%%%%%%%%%%%%%%%
From $\bck{F_k, H}\big|_{k+1}=0$ and $\bck{F_k, H}\big|_{k} = 0$, we have the following equation for the cancellation of $k+1$ and $k$-support operator in $\bck{F_k, H}$:
%%%%%%%%%%%%%%%%%%%%%%%%%%%%%%%%%
\begin{align}
    \label{eq:k+1-cancellation}
    &\bck{F_k^{k}, H_0}\bigg|_{k+1}
    =0
    ,
    \\
    \label{eq:k-cancellation}
    &\bck{F_k^{k}, H_0}\bigg|_{k}
    +
    \bck{F_k^{k-1}, H_0}\bigg|_{k}
    +
    \bck{F_k^{k}, H_{\mathrm{int}}}
    =0
    .
\end{align}

%%%%%%%%%%%%%%%%%%%%%%%%%%%%%%%%%
We note that the following argument is applicable to the maximal support is less than half of the system size: $k<\floor{\frac{L}{2}}$, because otherwise there are other contributions to the cancellation of operators beyond what we give below.

\subsection{Cancellation of $(k+1)$-support operator}
%%%%%%%%%%%%%%%%%%%%%%%%%%%%%%
In the following, we consider the cancellation of $(k+1)$-support operators~\eqref{eq:k+1-cancellation}.
%%%%%%%%%%%%%%%%%%%%%%%%%%%%%%
$F_k^{ k}$ is written as
\begin{align}
    \label{eq:k-sup-sum}
    F_k^{ k}
    =
    \sum_{q\in \mathcal{C}_k}
    \sum_{i=1}^L
    c_{i}(q)
    q(i)
    .
\end{align}
%%%%%%%%%%%%%%%%%%%%%%%%%%%%%%
The configuration $q \in \mathcal{C}_k$ is represented as 
$
    q
    =
    \begin{tikzpicture}[baseline=-0.5*\completevertical]
        \stringud{0}{$\Bar{a}_1$}{$\Bar{b}_1$}
        \stringud{1}{$\Bar{a}_2$}{$\Bar{b}_2$}
        \stringud{2}{$\cdots$}{$\cdots$}
        \stringud{3}{$\Bar{a}_k$}{$\Bar{b}_k$}
    \end{tikzpicture}
$
where at least one of $\Bar{a}_1$ or $\Bar{b}_1$ and at least one of $\Bar{a}_k$ or $\Bar{b}_k$ is not $\circled{}$.
%%%%%%%%%%%%%%%%%%%%%%%%%%%%%%
We determine $c_i(q)$ to satisfy the cancellation of the $(k+1)$-support operators in~\eqref{eq:k+1-cancellation}.
%%%%%%%%%%%%%%%%%%%%%%%%%%%%%%
We note that the nodes at the starting point of the arrows below represent $k$-support operators, as long as we consider the cancellation of~\eqref{eq:k+1-cancellation}.

%%%%%%%%%%%%%%%%%%%%%%%%%%%%%%
We can see the coefficients are zero if the column at the right end or left end is filled for both upper and lower rows, such as $
c_i\paren{
    \begin{tikzpicture}[baseline=-0.5*\completevertical]
        \stringud{0}{\circled{+}}{\circled{-}}
        \stringud{1}{$\cdots$}{$\cdots$}
    \end{tikzpicture}
}
=
0
$ and $
c_i\paren{
    \begin{tikzpicture}[baseline=-0.5*\completevertical]
        \stringud{2}{\circled{z}}{\circled{+}}
        \stringud{1}{$\cdots$}{$\cdots$}
    \end{tikzpicture}
}
=
0
$, which is stated in the following Lemma.
\begin{lemma}
\label{lem:first}
%%%%%%%%%%%%%%%%%%%%%%%%%%%%%%
If $q\in \mathcal{C}_k$ satisfies $\circled{} \notin \bce{\Bar{a}_1, \Bar{b}_1}$ or $\circled{} \notin \bce{\Bar{a}_k, \Bar{b}_k}$, then $c_{i}(q)=0$ holds.
\end{lemma}
%%%%%%%%%%%%%%%%%%%%%%%%%%%%%%
%%%%%%%%%%%%%%%%%%%%%%%%%%%%%%
\begin{proof}
%%%%%%%%%%%%%%%%%%%%%%%%%%%%%%
Because at least one of $\Bar{a}_k$ or $\Bar{b}_k$ is not $\circled{}$, we consider the case of $\Bar{a}_k \neq \circled{}$ first.
%%%%%%%%%%%%%%%%%%%%%%%%%%%%%%
If $\circled{} \notin \bce{\Bar{a}_1, \Bar{b}_1}$, we have the following cancellations of $(k+1)$-support basis elements~$\widetilde{q}(i)$:
\begin{align}
    \label{eq:cancellation-1}
    \begin{tikzpicture}[baseline=(current bounding box.center)]
        % Second TikZ picture in a node
        \node[draw, rectangle, rounded corners=2mm] (B) at (0,2.5) {
            \begin{tikzpicture}[baseline=-0.5*\completevertical]
                \stringud{0}{$\Bar{a}_1$}{$\Bar{b}_1$}
                \stringud{1}{$\cdots$}{$\cdots$}
                \stringud{2.25}{$\Bar{a}_{k-1}$}{$\Bar{b}_{k-1}$}
                \stringud{3.5}{\circled{}}{$\Bar{b}_k$}
                \stringud{4.5}{\circled{\pm}}{\circled{}}
                %\commutatorhopu{3.5}{4.5}
            \end{tikzpicture}
        };
        % First TikZ picture in a node
        \node[draw, rectangle, rounded corners=2mm] (A) at (2,0) {
            \begin{tikzpicture}[baseline=-0.5*\completevertical, rounded corners=0mm]
                \stringud{0}{$\Bar{a}_1$}{$\Bar{b}_1$}
                \stringud{1}{$\cdots$}{$\cdots$}
                \stringud{2.25}{$\Bar{a}_{k-1}$}{$\Bar{b}_{k-1}$}
                \stringud{3.5}{\circled{\pm}}{$\Bar{b}_k$}
                \stringud{4.5}{\circleddotted{}}{\circleddotted{}}
                \commutatorhopudot{3.5}{4.5}
            \end{tikzpicture}
        };
        % Arrow from first node to the second
        \draw[->,thick] (A) -- (B);
        \node[above=0em  of B, scale=1] {$\widetilde{q}(i)$};
        \node[above=0em  of A, scale=1] {$q(i)$};
        \node[above=2em  of B, scale=1] {Case 1: $\Bar{a}_k = \circled{\pm}$};
    \end{tikzpicture}
    \ ,\quad\quad
    \begin{tikzpicture}[baseline=(current bounding box.center)]
        % Second TikZ picture in a node
        \node[draw, rectangle, rounded corners=2mm] (B) at (0,2.5) {
            \begin{tikzpicture}[baseline=-0.5*\completevertical]
                \stringud{0}{$\Bar{a}_1$}{$\Bar{b}_1$}
                \stringud{1}{$\cdots$}{$\cdots$}
                \stringud{2.25}{$\Bar{a}_{k-1}$}{$\Bar{b}_{k-1}$}
                \stringud{3.5}{\circled{+}}{$\Bar{b}_k$}
                \stringud{4.5}{\circled{-}}{\circled{}}
                %\commutatorhopu{3.5}{4.5}
            \end{tikzpicture}
        };
        % First TikZ picture in a node
        \node[draw, rectangle, rounded corners=2mm] (A) at (2,0) {
            \begin{tikzpicture}[baseline=-0.5*\completevertical, rounded corners=0mm]
                \stringud{0}{$\Bar{a}_1$}{$\Bar{b}_1$}
                \stringud{1}{$\cdots$}{$\cdots$}
                \stringud{2.25}{$\Bar{a}_{k-1}$}{$\Bar{b}_{k-1}$}
                \stringud{3.5}{\circled{z}}{$\Bar{b}_k$}
                \stringud{4.5}{\circleddotted{}}{\circleddotted{}}
                \commutatorhopudot{3.5}{4.5}
            \end{tikzpicture}
        };
        % Arrow from first node to the second
        \draw[->,thick] (A) -- (B);
        \node[above=0em  of B, scale=1] {$\widetilde{q}(i)$};
        \node[above=2em  of B, scale=1] {Case 2: $\Bar{a}_k = \circled{z}$};
        \node[above=0em  of A, scale=1] {$q(i)$};
    \end{tikzpicture}
    \quad
\end{align}
where the commutator that generates $\widetilde{q}$ is indicated by the dotted line as a guide.
%%%%%%%%%%%%%%%%%%%%%%%%%%%%%%
From~\eqref{eq:graphical-cancellation-one-m}, we have $c_i(q)=0$ in both cases.

%%%%%%%%%%%%%%%%%%%%%%%%%%%%%%
We have proved $c_{i}(q)=0$ if $q$ satisfies $\circled{} \notin \bce{\Bar{a}_1, \Bar{b}_1}$ and $\Bar{a}_k \neq \circled{}$.
%%%%%%%%%%%%%%%%%%%%%%%%%%%%%%
Given the same argument holds when the roles of $\Bar{a}$ and $\Bar{b}$ are interchanged, we have also proved $c_{i}(q)=0$ if $q$ satisfies $\circled{} \notin \bce{\Bar{a}_1, \Bar{b}_1}$ and $\Bar{b}_k \neq \circled{}$.
%%%%%%%%%%%%%%%%%%%%%%%%%%%%%%
Thus we have proved $c_{i}(q)=0$ if $q$ satisfies $\circled{} \notin \bce{\Bar{a}_1, \Bar{b}_1}$.
%%%%%%%%%%%%%%%%%%%%%%%%%%%%%%
In the same way, we can also prove $c_{i}(q)=0$ if $q$ satisfies $\circled{} \notin \bce{\Bar{a}_k, \Bar{b}_k}$.
\end{proof}

%%%%%%%%%%%%%%%%%%%%%%%%%%%%%%
We can prove the coefficients are zero if $\circled{z}$ is included in the column at the right end or left end, such as $
c_i\paren{
    \begin{tikzpicture}[baseline=-0.5*\completevertical]
        \stringud{0}{\circled{z}}{\circled{}}
        \stringud{1}{$\cdots$}{$\cdots$}
    \end{tikzpicture}
}
=
0
$, which is stated in the following Lemma.
\begin{lemma}
\label{lem:second}
%%%%%%%%%%%%%%%%%%%%%%%%%%%%%%
If $q\in \mathcal{C}_k$  satisfies $\circled{z}\in\{\Bar{a}_1,\Bar{b}_1, \Bar{a}_k,\Bar{b}_k\}$, then $c_{i}(q)=0$ holds.
\end{lemma}
%%%%%%%%%%%%%%%%%%%%%%%%%%%%%%
%%%%%%%%%%%%%%%%%%%%%%%%%%%%%%
\begin{proof}
%%%%%%%%%%%%%%%%%%%%%%%%%%%%%%
Because at least one of $\Bar{a}_k$ or $\Bar{b}_k$ is not $\circled{}$, we consider the case of $\Bar{a}_k \neq \circled{}$ first.
%%%%%%%%%%%%%%%%%%%%%%%%%%%%%%
If $q\in \mathcal{C}_k$ satisfies $\Bar{a}_1=\circled{z}$, the following relations for the cancellations of $(k+1)$-support basis elements~$\widetilde{q}(i)$ hold:
\begin{align}
    \label{eq:cancellation-2}
    \begin{tikzpicture}[baseline=(current bounding box.center)]
        % Second TikZ picture in a node
        \node[draw, rectangle, rounded corners=2mm] (B) at (0,2.5) {
            \begin{tikzpicture}[baseline=-0.5*\completevertical]
                \stringud{0}{\circled{z}}{$\Bar{b}_1$}
                \stringud{1}{$\cdots$}{$\cdots$}
                \stringud{2.25}{$\Bar{a}_{k-1}$}{$\Bar{b}_{k-1}$}
                \stringud{3.5}{\circled{}}{$\Bar{b}_k$}
                \stringud{4.5}{\circled{\pm}}{\circled{}}
                %\commutatorhopu{3.5}{4.5}
            \end{tikzpicture}
        };
        % First TikZ picture in a node
        \node[draw, rectangle, rounded corners=2mm] (A) at (2,0) {
            \begin{tikzpicture}[baseline=-0.5*\completevertical, rounded corners=0mm]
                \stringud{0}{\circled{z}}{$\Bar{b}_1$}
                \stringud{1}{$\cdots$}{$\cdots$}
                \stringud{2.25}{$\Bar{a}_{k-1}$}{$\Bar{b}_{k-1}$}
                \stringud{3.5}{\circled{\pm}}{$\Bar{b}_k$}
                \stringud{4.5}{\circleddotted{}}{\circleddotted{}}
                \commutatorhopudot{3.5}{4.5}
                %\commutatorhopu{3.5}{4.5}
            \end{tikzpicture}
        };
        % Arrow from first node to the second
        \draw[->,thick] (A) -- (B);
        \node[above=0em  of B, scale=1] {$\widetilde{q}(i)$};
        \node[above=0em  of A, scale=1] {$q(i)$};
        \node[above=2em  of B, scale=1] {Case 1: $\Bar{a}_k = \circled{\pm}$};
    \end{tikzpicture}
    \quad,\quad
    \begin{tikzpicture}[baseline=(current bounding box.center)]
        % Second TikZ picture in a node
        \node[draw, rectangle, rounded corners=2mm] (B) at (0,2.5) {
            \begin{tikzpicture}[baseline=-0.5*\completevertical]
                \stringud{0}{\circled{z}}{$\Bar{b}_1$}
                \stringud{1}{$\cdots$}{$\cdots$}
                \stringud{2.25}{$\Bar{a}_{k-1}$}{$\Bar{b}_{k-1}$}
                \stringud{3.5}{\circled{+}}{$\Bar{b}_k$}
                \stringud{4.5}{\circled{-}}{\circled{}}
                %\commutatorhopu{3.5}{4.5}
            \end{tikzpicture}
        };
        % First TikZ picture in a node
        \node[draw, rectangle, rounded corners=2mm] (A) at (2,0) {
            \begin{tikzpicture}[baseline=-0.5*\completevertical, rounded corners=0mm]
                \stringud{0}{\circled{z}}{$\Bar{b}_1$}
                \stringud{1}{$\cdots$}{$\cdots$}
                \stringud{2.25}{$\Bar{a}_{k-1}$}{$\Bar{b}_{k-1}$}
                \stringud{3.5}{\circled{z}}{$\Bar{b}_k$}
                %\commutatorhopu{3.5}{4.5}
                \stringud{4.5}{\circleddotted{}}{\circleddotted{}}
                \commutatorhopudot{3.5}{4.5}
            \end{tikzpicture}
        };
        % Arrow from first node to the second
        \draw[->,thick] (A) -- (B);
        \node[above=0em  of B, scale=1] {$\widetilde{q}(i)$};
        \node[above=0em  of A, scale=1] {$q(i)$};
        \node[above=2em  of B, scale=1] {Case 2: $\Bar{a}_k = \circled{z}$};
    \end{tikzpicture}
    \quad
\end{align}
and from~\eqref{eq:graphical-cancellation-one-m}, we have $c_i(q) = 0$ in both cases.
%%%%%%%%%%%%%%%%%%%%%%%%%%%%%%
In the same way, we can prove $c_i(q) = 0$ if $q$ satisfies $\Bar{a}_1=\circled{z}$ and $\Bar{b}_k \neq \circled{}$.
%%%%%%%%%%%%%%%%%%%%%%%%%%%%%%
Thus we have proved $c_{i}(q)=0$ if $q$ satisfies $\Bar{a}_1 = \circled{z}$.

%%%%%%%%%%%%%%%%%%%%%%%%%%%%%%
The same argument holds when the roles of $\Bar{a}$ and $\Bar{b}$ are interchanged, and then we can also prove $c_{i}(q)=0$ if $q$ satisfies $\Bar{b}_1 = \circled{z}$.
%%%%%%%%%%%%%%%%%%%%%%%%%%%%%%
Thus we have proved $c_{i}(q)=0$ if $q$ satisfies $\circled{z} \in \bce{\Bar{a}_1, \Bar{b}_1}$.
%%%%%%%%%%%%%%%%%%%%%%%%%%%%%%
In the same way, we can also prove $c_{i}(q)=0$ if $q$ satisfies $\circled{z} \in \bce{\Bar{a}_k, \Bar{b}_k}$.
\end{proof}

%%%%%%%%%%%%%%%%%%%%%%%%%%%%%%
We can prove the coefficients are zero if the second column from the left end is not empty for both rows, such as $
c_i\paren{
    \begin{tikzpicture}[baseline=-0.5*\completevertical]
        \stringud{0}{\circled{+}}{\circled{}}
        \stringud{1}{\circled{}}{\circled{-}}
        \stringud{2}{$\cdots$}{$\cdots$}
    \end{tikzpicture}
}
=
0
$, which is stated in the following Lemma.
\begin{lemma}
    \label{lem:third}
    If $q\in \mathcal{C}_k$($k\geq 3$) satisfies $\{\Bar{a}_2, \Bar{b}_2\} \neq \{ \circled{}, \circled{}\}$, then $c_{i}(q)=0$ holds.
\end{lemma}
\begin{proof}
%%%%%%%%%%%%%%%%%%%%%%%%%%%%%%
If $\circled{z}\in\{\Bar{a}_1,\Bar{b}_1, \Bar{a}_k,\Bar{b}_k\}$ or $\circled{} \notin \bce{\Bar{a}_1, \Bar{b}_1}$ or $\circled{} \notin \bce{\Bar{a}_k, \Bar{b}_k}$, we can see $c_i(q)=0$ holds from Lemma~\ref{lem:first} and~\ref{lem:second}.
%%%%%%%%%%%%%%%%%%%%%%%%%%%%%%
Thus, the nontrivial cases are
 $
\begin{tikzpicture}[baseline=-0.5*\completevertical]
    \stringud{0}{$\Bar{a}_1$}{$\Bar{b}_1$}
\end{tikzpicture}
, 
\begin{tikzpicture}[baseline=-0.5*\completevertical]
    \stringud{0}{$\Bar{a}_{k}$}{$\Bar{b}_{k}$}
\end{tikzpicture}
\in
\bce{
    \begin{tikzpicture}[baseline=-0.5*\completevertical]
        \stringud{0}{$\circled{+}$}{$\circled{}$}
    \end{tikzpicture}
    ,
    \begin{tikzpicture}[baseline=-0.5*\completevertical]
        \stringud{0}{$\circled{}$}{$\circled{+}$}
    \end{tikzpicture}
    ,
    \begin{tikzpicture}[baseline=-0.5*\completevertical]
        \stringud{0}{$\circled{-}$}{$\circled{}$}
    \end{tikzpicture}
    ,
    \begin{tikzpicture}[baseline=-0.5*\completevertical]
        \stringud{0}{$\circled{}$}{$\circled{-}$}
    \end{tikzpicture}
}
$
.
%%%%%%%%%%%%%%%%%%%%%%%%%%%%%%
We consider the cases of
$
\begin{tikzpicture}[baseline=-0.5*\completevertical]
    \stringud{0}{$\Bar{a}_1$}{$\Bar{b}_1$}
\end{tikzpicture}
=
\begin{tikzpicture}[baseline=-0.5*\completevertical]
    \stringud{0}{$\circled{+}$}{$\circled{}$}
\end{tikzpicture}
,\quad 
\begin{tikzpicture}[baseline=-0.5*\completevertical]
    \stringud{0}{$\Bar{a}_k$}{$\Bar{b}_k$}
\end{tikzpicture}
=
\begin{tikzpicture}[baseline=-0.5*\completevertical]
    \stringud{0}{$\circled{+}$}{$\circled{}$}
\end{tikzpicture}
$ first.
%%%%%%%%%%%%%%%%%%%%%%%%%%%%%%
If $q\in \mathcal{C}_k$ satisfies $\Bar{b}_2 \neq \circled{}$, we have the following cancellations of the $(k+1)$-support basis elements~$\widetilde{q}(i)$:
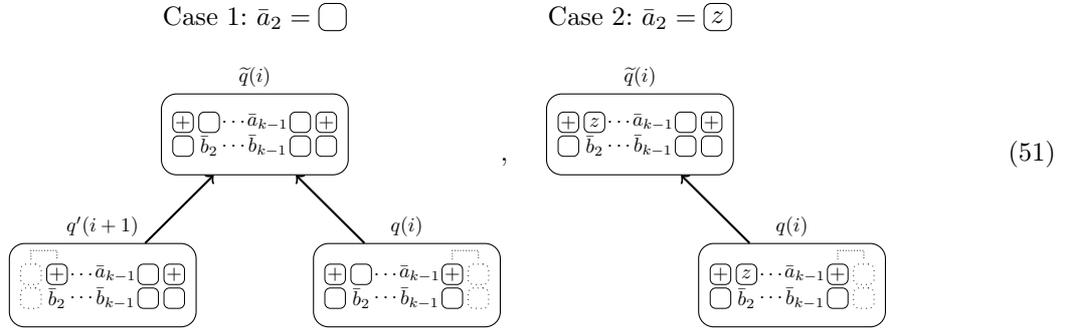
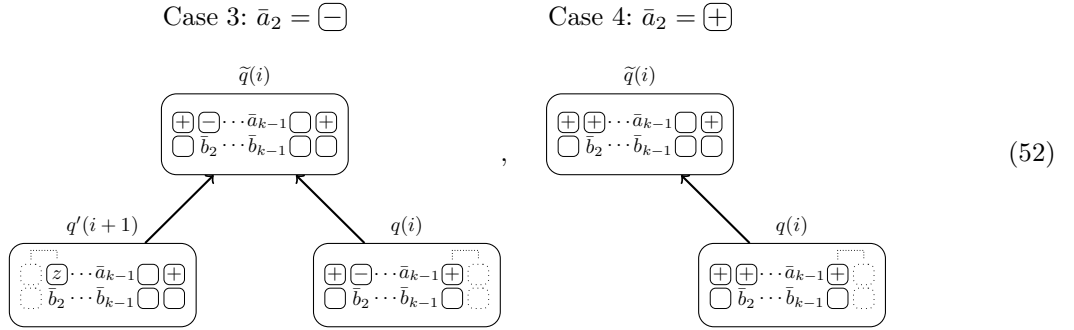
\begin{align}
    \begin{tikzpicture}[baseline=(current bounding box.center)]
        % Second TikZ picture in a node
        \node[draw, rectangle, rounded corners=2mm, scale=0.75] (B) at (2,2) {
            \begin{tikzpicture}[baseline=-0.5*\completevertical]
                \stringud{-1}{\circled{+}}{\circled{}}
                \stringud{0}{\circled{}}{$\Bar{b}_2$}
                \stringud{1}{$\cdots$}{$\cdots$}
                \stringud{2.25}{$\Bar{a}_{k-1}$}{$\Bar{b}_{k-1}$}
                \stringud{3.5}{\circled{}}{\circled{}}
                \stringud{4.5}{\circled{+}}{\circled{}}
                %\commutatorhopu{3.5}{4.5}
                %\stringlongu{2.}{3}{$\widetilde{q}(i-1)$}
            \end{tikzpicture}
        };
        \node[draw, rectangle, rounded corners=2mm, scale=0.75] (A1) at (0,0) {
            \begin{tikzpicture}[baseline=-0.5*\completevertical, rounded corners=0mm]
                \stringud{-1}{\circleddotted{}}{\circleddotted{}}
                \stringud{0}{\circled{+}}{$\Bar{b}_2$}
                \stringud{1}{$\cdots$}{$\cdots$}
                \stringud{2.25}{$\Bar{a}_{k-1}$}{$\Bar{b}_{k-1}$}
                \stringud{3.5}{\circled{}}{\circled{}}
                \stringud{4.5}{\circled{+}}{\circled{}}
                \commutatorhopudot{-1}{0}
                %\stringlongu{2.25}{3}{$q_1(i)$}
            \end{tikzpicture}
        };
        % First TikZ picture in a node
        \node[draw, rectangle, rounded corners=2mm, scale=0.75] (A2) at (4,0) {
            \begin{tikzpicture}[baseline=-0.5*\completevertical, rounded corners=0mm]
                \stringud{-1}{\circled{+}}{\circled{}}
                \stringud{0}{\circled{}}{$\Bar{b}_2$}
                \stringud{1}{$\cdots$}{$\cdots$}
                \stringud{2.25}{$\Bar{a}_{k-1}$}{$\Bar{b}_{k-1}$}
                \stringud{3.5}{\circled{+}}{\circled{}}
                \stringud{4.5}{\circleddotted{}}{\circleddotted{}}
                \commutatorhopudot{3.5}{4.5}
                %\stringlongu{1}{3}{$q_2(i-1)$}
            \end{tikzpicture}
        };
        % Arrow from first node to the second
        \draw[->,thick] (A1) -- (B);
        \draw[->,thick] (A2) -- (B);
        \node[above=0em  of B, scale=0.75] {$\widetilde{q}(i)$};
        \node[above=2em  of B, scale=1] {Case 1: $\Bar{a}_2 = \circled{}$};
        \node[above=0em  of A1, scale=0.75] {$q^\prime(i+1)$};
        \node[above=0em  of A2, scale=0.75] {$q(i)$};
    \end{tikzpicture}
    ,
    \quad
    \begin{tikzpicture}[baseline=(current bounding box.center)]
        % Second TikZ picture in a node
        \node[draw, rectangle, rounded corners=2mm, scale=0.75] (B) at (0,2) {
            \begin{tikzpicture}[baseline=-0.5*\completevertical]
                \stringud{-1}{\circled{+}}{\circled{}}
                \stringud{0}{\circled{z}}{$\Bar{b}_2$}
                \stringud{1}{$\cdots$}{$\cdots$}
                \stringud{2.25}{$\Bar{a}_{k-1}$}{$\Bar{b}_{k-1}$}
                \stringud{3.5}{\circled{}}{\circled{}}
                \stringud{4.5}{\circled{+}}{\circled{}}
                %\commutatorhopu{3.5}{4.5}
                %\stringlongu{2.}{3}{$\widetilde{q}(i-1)$}
            \end{tikzpicture}
        };
        % First TikZ picture in a node
        \node[draw, rectangle, rounded corners=2mm, scale=0.75] (A2) at (2,0) {
            \begin{tikzpicture}[baseline=-0.5*\completevertical, rounded corners=0mm]
                \stringud{-1}{\circled{+}}{\circled{}}
                \stringud{0}{\circled{z}}{$\Bar{b}_2$}
                \stringud{1}{$\cdots$}{$\cdots$}
                \stringud{2.25}{$\Bar{a}_{k-1}$}{$\Bar{b}_{k-1}$}
                \stringud{3.5}{\circled{+}}{\circled{}}
                \stringud{4.5}{\circleddotted{}}{\circleddotted{}}
                \commutatorhopudot{3.5}{4.5}
                %\stringlongu{1}{3}{$q_2(i-1)$}
            \end{tikzpicture}
        };
        % Arrow from first node to the second
        \draw[->,thick] (A2) -- (B);
        \node[above=0em  of B, scale=0.75] {$\widetilde{q}(i)$};
        \node[above=2em  of B, scale=1] {Case 2: $\Bar{a}_2 = \circled{z}$};
        \node[above=0em  of A2, scale=0.75] {$q(i)$};
    \end{tikzpicture}
    \\[25pt]
    \begin{tikzpicture}[baseline=(current bounding box.center)]
        % Second TikZ picture in a node
        \node[draw, rectangle, rounded corners=2mm, scale=0.75] (B) at (2,2) {
            \begin{tikzpicture}[baseline=-0.5*\completevertical]
                \stringud{-1}{\circled{+}}{\circled{}}
                \stringud{0}{\circled{-}}{$\Bar{b}_2$}
                \stringud{1}{$\cdots$}{$\cdots$}
                \stringud{2.25}{$\Bar{a}_{k-1}$}{$\Bar{b}_{k-1}$}
                \stringud{3.5}{\circled{}}{\circled{}}
                \stringud{4.5}{\circled{+}}{\circled{}}
                %\commutatorhopu{3.5}{4.5}
                %\stringlongu{2.}{3}{$\widetilde{q}(i-1)$}
            \end{tikzpicture}
        };
        \node[draw, rectangle, rounded corners=2mm, scale=0.75] (A1) at (0,0) {
            \begin{tikzpicture}[baseline=-0.5*\completevertical, rounded corners=0mm]
                \stringud{-1}{\circleddotted{}}{\circleddotted{}}
                \stringud{0}{\circled{z}}{$\Bar{b}_2$}
                \stringud{1}{$\cdots$}{$\cdots$}
                \stringud{2.25}{$\Bar{a}_{k-1}$}{$\Bar{b}_{k-1}$}
                \stringud{3.5}{\circled{}}{\circled{}}
                \stringud{4.5}{\circled{+}}{\circled{}}
                \commutatorhopudot{-1}{0}
                %\stringlongu{2.25}{3}{$q_1(i)$}
            \end{tikzpicture}
        };
        % First TikZ picture in a node
        \node[draw, rectangle, rounded corners=2mm, scale=0.75] (A2) at (4,0) {
            \begin{tikzpicture}[baseline=-0.5*\completevertical, rounded corners=0mm]
                \stringud{-1}{\circled{+}}{\circled{}}
                \stringud{0}{\circled{-}}{$\Bar{b}_2$}
                \stringud{1}{$\cdots$}{$\cdots$}
                \stringud{2.25}{$\Bar{a}_{k-1}$}{$\Bar{b}_{k-1}$}
                \stringud{3.5}{\circled{+}}{\circled{}}
                \stringud{4.5}{\circleddotted{}}{\circleddotted{}}
                \commutatorhopudot{3.5}{4.5}
                %\stringlongu{1}{3}{$q_2(i-1)$}
            \end{tikzpicture}
        };
        % Arrow from first node to the second
        \draw[->,thick] (A1) -- (B);
        \draw[->,thick] (A2) -- (B);
        \node[above=0em  of B, scale=0.75] {$\widetilde{q}(i)$};
        \node[above=2em  of B, scale=1] {Case 3: $\Bar{a}_2 = \circled{-}$};
        \node[above=0em  of A1, scale=0.75] {$q^\prime(i+1)$};
        \node[above=0em  of A2, scale=0.75] {$q(i)$};
    \end{tikzpicture}
    ,
    \quad
    \begin{tikzpicture}[baseline=(current bounding box.center)]
        % Second TikZ picture in a node
        \node[draw, rectangle, rounded corners=2mm, scale=0.75] (B) at (0,2) {
            \begin{tikzpicture}[baseline=-0.5*\completevertical]
                \stringud{-1}{\circled{+}}{\circled{}}
                \stringud{0}{\circled{+}}{$\Bar{b}_2$}
                \stringud{1}{$\cdots$}{$\cdots$}
                \stringud{2.25}{$\Bar{a}_{k-1}$}{$\Bar{b}_{k-1}$}
                \stringud{3.5}{\circled{}}{\circled{}}
                \stringud{4.5}{\circled{+}}{\circled{}}
                %\commutatorhopu{3.5}{4.5}
                %\stringlongu{2.}{3}{$\widetilde{q}(i-1)$}
            \end{tikzpicture}
        };
        % First TikZ picture in a node
        \node[draw, rectangle, rounded corners=2mm, scale=0.75] (A2) at (2,0) {
            \begin{tikzpicture}[baseline=-0.5*\completevertical, rounded corners=0mm]
                \stringud{-1}{\circled{+}}{\circled{}}
                \stringud{0}{\circled{+}}{$\Bar{b}_2$}
                \stringud{1}{$\cdots$}{$\cdots$}
                \stringud{2.25}{$\Bar{a}_{k-1}$}{$\Bar{b}_{k-1}$}
                \stringud{3.5}{\circled{+}}{\circled{}}
                \stringud{4.5}{\circleddotted{}}{\circleddotted{}}
                \commutatorhopudot{3.5}{4.5}
                %\stringlongu{1}{3}{$q_2(i-1)$}
            \end{tikzpicture}
        };
        % Arrow from first node to the second
        \draw[->,thick] (A2) -- (B);
        \node[above=0em  of B, scale=0.75] {$\widetilde{q}(i)$};
        \node[above=2em  of B, scale=1] {Case 4: $\Bar{a}_2 = \circled{+}$};
        \node[above=0em  of A2, scale=0.75] {$q(i)$};
    \end{tikzpicture}
\end{align}
and for Case 2 and Case 4, we have $c_i(q) = 0$ from~\eqref{eq:graphical-cancellation-one-m}.
%%%%%%%%%%%%%%%%%%%%%%%%%%%%%%
For Case 1 and Case 3, we have $c_i(q) \propto c_{i+1}(q^\prime)=0$ from~\eqref{eq:graphical-cancellation-two-m} and Lemma~\ref{lem:first}.

%%%%%%%%%%%%%%%%%%%%%%%%%%%%%%
We can also do the same argument for the other choices of 
$
\begin{tikzpicture}[baseline=-0.5*\completevertical]
    \stringud{0}{$\Bar{a}_1$}{$\Bar{b}_1$}
\end{tikzpicture}
$
and
$
\begin{tikzpicture}[baseline=-0.5*\completevertical]
    \stringud{0}{$\Bar{a}_k$}{$\Bar{b}_k$}
\end{tikzpicture}
$.
%%%%%%%%%%%%%%%%%%%%%%%%%%%%%%
Then we have proved  $c_i(q) = 0$ if $\Bar{b}_2\neq \circled{}$.
%%%%%%%%%%%%%%%%%%%%%%%%%%%%%%
Given the same argument holds when the roles of $\Bar{a}$ and $\Bar{b}$ are interchanged, we have also proved $c_i(q) = 0$ if $\Bar{a}_2\neq \circled{}$.
\end{proof}

%%%%%%%%%%%%%%%%%%%%%%%%%%%%%%
We can prove the coefficients are zero if a middle column is not empty for both rows, such as $
c_i\paren{
    \begin{tikzpicture}[baseline=-0.5*\completevertical]
        \stringud{0}{\circled{+}}{\circled{}}
        \stringud{1}{$\cdots$}{$\cdots$}
        \stringud{2}{\circled{}}{\circled{-}}
        \stringud{3}{$\cdots$}{$\cdots$}
        \stringud{4}{\circled{-}}{\circled{}}
    \end{tikzpicture}
}
=
0
$, which is stated in the following Lemma.
\begin{lemma}
    \label{lem:forth}
    If there exists an integer \( l \) with \( 2 \leq l \leq k-1 \) such that $\bce{\Bar{a}_l, \Bar{b}_l} \neq \bce{\circled{}, \circled{}}$ for $q \in \mathcal{C}_k$, then \( c_{i}(q) = 0 \) holds.
\end{lemma}
\begin{proof}
    %%%%%%%%%%%%%%%%%%%%%%%%%%%%%%
If $\circled{z}\in\{\Bar{a}_1,\Bar{b}_1, \Bar{a}_k,\Bar{b}_k\}$ or $\circled{} \notin \bce{\Bar{a}_1, \Bar{b}_1}$ or $\circled{} \notin \bce{\Bar{a}_k, \Bar{b}_k}$, we can see $c_i(q)=0$ holds from Lemma~\ref{lem:first} and~\ref{lem:second}.
%%%%%%%%%%%%%%%%%%%%%%%%%%%%%%
Thus, the nontrivial cases are
 $
\begin{tikzpicture}[baseline=-0.5*\completevertical]
    \stringud{0}{$\Bar{a}_1$}{$\Bar{b}_1$}
\end{tikzpicture}
, 
\begin{tikzpicture}[baseline=-0.5*\completevertical]
    \stringud{0}{$\Bar{a}_{k}$}{$\Bar{b}_{k}$}
\end{tikzpicture}
\in
\bce{
    \begin{tikzpicture}[baseline=-0.5*\completevertical]
        \stringud{0}{$\circled{+}$}{$\circled{}$}
    \end{tikzpicture}
    ,
    \begin{tikzpicture}[baseline=-0.5*\completevertical]
        \stringud{0}{$\circled{}$}{$\circled{+}$}
    \end{tikzpicture}
    ,
    \begin{tikzpicture}[baseline=-0.5*\completevertical]
        \stringud{0}{$\circled{-}$}{$\circled{}$}
    \end{tikzpicture}
    ,
    \begin{tikzpicture}[baseline=-0.5*\completevertical]
        \stringud{0}{$\circled{}$}{$\circled{-}$}
    \end{tikzpicture}
}
$
.
%%%%%%%%%%%%%%%%%%%%%%%%%%%%%%
We first consider the cases of
$
\begin{tikzpicture}[baseline=-0.5*\completevertical]
    \stringud{0}{$\Bar{a}_1$}{$\Bar{b}_1$}
\end{tikzpicture}
=
\begin{tikzpicture}[baseline=-0.5*\completevertical]
    \stringud{0}{$\circled{+}$}{$\circled{}$}
\end{tikzpicture}
,\quad 
\begin{tikzpicture}[baseline=-0.5*\completevertical]
    \stringud{0}{$\Bar{a}_k$}{$\Bar{b}_k$}
\end{tikzpicture}
=
\begin{tikzpicture}[baseline=-0.5*\completevertical]
    \stringud{0}{$\circled{+}$}{$\circled{}$}
\end{tikzpicture}
$.
%%%%%%%%%%%%%%%%%%%%%%%%%%%%%%
The $l=2$ case corresponds to the Lemma~\ref{lem:third}, and thus we consider the $l>2$ case below.
%%%%%%%%%%%%%%%%%%%%%%%%%%%%%%
If there exists $l(2< l \leq k-1)$ such that $\bce{\Bar{a}_l, \Bar{b}_l} \neq \bce{\circled{}, \circled{}}$ and $\Bar{a}_m, \Bar{b}_m = \circled{}$ for $(2\leq m\leq l-1)$, we have the following  cancellations of $(k+1)$-support basis element $\widetilde{q}_m(i+m-1)\ (1\leq m \leq l-2)$:
%%%%%%%%%%%%%%%%%%%%%%%%%%%%%%
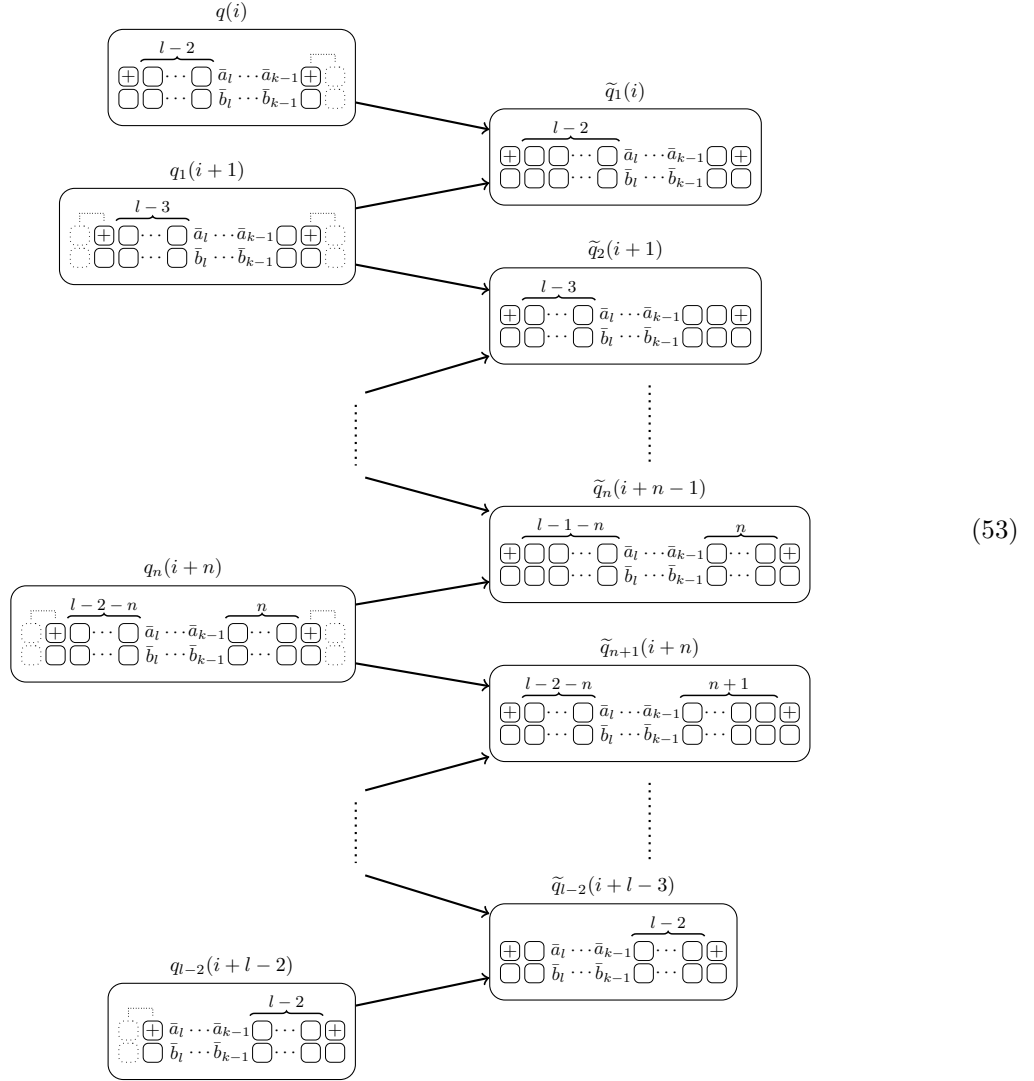
\begin{align}
    \begin{tikzpicture}[baseline=(current bounding box.center), x=1.25em, y=3.em]
        % First TikZ picture in a node
        \node[draw, rectangle, rounded corners=2mm, scale=0.7, anchor=east] (A1) at (0,0) {
            \begin{tikzpicture}[baseline=-0.5*\completevertical, rounded corners=0mm]
                \stringud{0}{\circled{+}}{\circled{}}
                \stringud{1}{\circled{}}{\circled{}}
                \stringud{2}{$\cdots$}{$\cdots$}
                \stringud{3}{\circled{}}{\circled{}}
                \stringud{4}{$\Bar{a}_{l}$}{$\Bar{b}_{l}$}
                \stringud{5}{$\cdots$}{$\cdots$}
                \stringud{6.25}{$\Bar{a}_{k-1}$}{$\Bar{b}_{k-1}$}
                \stringud{7.5}{\circled{+}}{\circled{}}
                \stringud{8.5}{\circleddotted{}}{\circleddotted{}}
                \commutatorhopudot{7.5}{8.5}
                %\stringud{8.5}{\circleddotted{}}{\circleddotted{}}
                %\commutatorhopudot{7.5}{8.5}
                \ubraceplusstring{1}{4}{$l-2$}{0}
                %\stringlongu{1}{3}{$q_2(i-1)$}
            \end{tikzpicture}
        };
         % Second TikZ picture in a node
         \node[draw, rectangle, rounded corners=2mm, scale=0.7, anchor=west] (B1) at (4,-1) {
            \begin{tikzpicture}[baseline=-0.5*\completevertical, rounded corners=0mm]
                \stringud{-1}{\circled{+}}{\circled{}}
                \stringud{0}{\circled{}}{\circled{}}
                \stringud{1}{\circled{}}{\circled{}}
                \stringud{2}{$\cdots$}{$\cdots$}
                \stringud{3}{\circled{}}{\circled{}}
                \stringud{4}{$\Bar{a}_{l}$}{$\Bar{b}_{l}$}
                \stringud{5}{$\cdots$}{$\cdots$}
                \stringud{6.25}{$\Bar{a}_{k-1}$}{$\Bar{b}_{k-1}$}
                \stringud{7.5}{\circled{}}{\circled{}}
                \stringud{8.5}{\circled{+}}{\circled{}}
                %\commutatorhopudot{7.5}{8.5}
                \ubraceplusstring{0}{4}{$l-2$}{0}
                %\stringlongu{1}{3}{$q_2(i-1)$}
            \end{tikzpicture}
        };
        \node[draw, rectangle, rounded corners=2mm, scale=0.7, anchor=east] (A2) at (0,-2) {
            \begin{tikzpicture}[baseline=-0.5*\completevertical, rounded corners=0mm]
                \commutatorhopudot{-1}{0}
                \stringud{-1}{\circleddotted{}}{\circleddotted{}}
                \stringud{0}{\circled{+}}{\circled{}}
                \stringud{1}{\circled{}}{\circled{}}
                \stringud{2}{$\cdots$}{$\cdots$}
                \stringud{3}{\circled{}}{\circled{}}
                \stringud{4}{$\Bar{a}_{l}$}{$\Bar{b}_{l}$}
                \stringud{5}{$\cdots$}{$\cdots$}
                \stringud{6.25}{$\Bar{a}_{k-1}$}{$\Bar{b}_{k-1}$}
                \stringud{7.5}{\circled{}}{\circled{}}
                \stringud{8.5}{\circled{+}}{\circled{}}
                \stringud{9.5}{\circleddotted{}}{\circleddotted{}}
                \commutatorhopudot{8.5}{9.5}
                %\stringud{8.5}{\circleddotted{}}{\circleddotted{}}
                %\commutatorhopudot{7.5}{8.5}
                \ubraceplusstring{1}{4}{$l-3$}{0}
                %\stringlongu{1}{3}{$q_2(i-1)$}
            \end{tikzpicture}
        };
        \node[draw, rectangle, rounded corners=2mm, scale=0.7, anchor=west] (B2) at (4,-3) {
            \begin{tikzpicture}[baseline=-0.5*\completevertical, rounded corners=0mm]
                \stringud{0}{\circled{+}}{\circled{}}
                \stringud{1}{\circled{}}{\circled{}}
                \stringud{2}{$\cdots$}{$\cdots$}
                \stringud{3}{\circled{}}{\circled{}}
                \stringud{4}{$\Bar{a}_{l}$}{$\Bar{b}_{l}$}
                \stringud{5}{$\cdots$}{$\cdots$}
                \stringud{6.25}{$\Bar{a}_{k-1}$}{$\Bar{b}_{k-1}$}
                \stringud{7.5}{\circled{}}{\circled{}}
                \stringud{8.5}{\circled{}}{\circled{}}
                \stringud{9.5}{\circled{+}}{\circled{}}
                %\commutatorhopudot{7.5}{8.5}
                \ubraceplusstring{1}{4}{$l-3$}{0}
                %\stringlongu{1}{3}{$q_2(i-1)$}
            \end{tikzpicture}
        };
        \node (A3a) at (0,-4) {};
        \node (A3b) at (0,-5) {};
        \node[draw, rectangle, rounded corners=2mm, scale=0.7, anchor=west] (B3) at (4,-6) {
            \begin{tikzpicture}[baseline=-0.5*\completevertical, rounded corners=0mm]
                \stringud{-1}{\circled{+}}{\circled{}}
                \stringud{0}{\circled{}}{\circled{}}
                \stringud{1}{\circled{}}{\circled{}}
                \stringud{2}{$\cdots$}{$\cdots$}
                \stringud{3}{\circled{}}{\circled{}}
                \stringud{4}{$\Bar{a}_{l}$}{$\Bar{b}_{l}$}
                \stringud{5}{$\cdots$}{$\cdots$}
                \stringud{6.25}{$\Bar{a}_{k-1}$}{$\Bar{b}_{k-1}$}
                \stringud{7.5}{\circled{}}{\circled{}}
                \stringud{8.5}{$\cdots$}{$\cdots$}
                \stringud{9.5}{\circled{}}{\circled{}}
                \stringud{10.5}{\circled{+}}{\circled{}}
                %\commutatorhopudot{7.5}{8.5}
                \ubraceplusstring{0}{4}{$l-1-n$}{0}
                \ubraceplusstring{7.5}{10.5}{$n$}{0}
                %\stringlongu{1}{3}{$q_2(i-1)$}
            \end{tikzpicture}
        };
        \node[draw, rectangle, rounded corners=2mm, scale=0.7, anchor=east] (A4) at (0,-7) {
            \begin{tikzpicture}[baseline=-0.5*\completevertical, rounded corners=0mm]
                \commutatorhopudot{-1}{0}
                \stringud{-1}{\circleddotted{}}{\circleddotted{}}
                \stringud{0}{\circled{+}}{\circled{}}
                \stringud{1}{\circled{}}{\circled{}}
                \stringud{2}{$\cdots$}{$\cdots$}
                \stringud{3}{\circled{}}{\circled{}}
                \stringud{4}{$\Bar{a}_{l}$}{$\Bar{b}_{l}$}
                \stringud{5}{$\cdots$}{$\cdots$}
                \stringud{6.25}{$\Bar{a}_{k-1}$}{$\Bar{b}_{k-1}$}
                \stringud{7.5}{\circled{}}{\circled{}}
                \stringud{8.5}{$\cdots$}{$\cdots$}
                \stringud{9.5}{\circled{}}{\circled{}}
                \stringud{10.5}{\circled{+}}{\circled{}}
                \stringud{11.5}{\circleddotted{}}{\circleddotted{}}
                \commutatorhopudot{10.5}{11.5}
                %\commutatorhopudot{7.5}{8.5}
                \ubraceplusstring{1}{4}{$l-2-n$}{0}
                \ubraceplusstring{7.5}{10.5}{$n$}{0}
                %\stringlongu{1}{3}{$q_2(i-1)$}
            \end{tikzpicture}
        };
        \node[draw, rectangle, rounded corners=2mm, scale=0.7, anchor=west] (B4) at (4,-8) {
            \begin{tikzpicture}[baseline=-0.5*\completevertical, rounded corners=0mm]
                \stringud{0}{\circled{+}}{\circled{}}
                \stringud{1}{\circled{}}{\circled{}}
                \stringud{2}{$\cdots$}{$\cdots$}
                \stringud{3}{\circled{}}{\circled{}}
                \stringud{4}{$\Bar{a}_{l}$}{$\Bar{b}_{l}$}
                \stringud{5}{$\cdots$}{$\cdots$}
                \stringud{6.25}{$\Bar{a}_{k-1}$}{$\Bar{b}_{k-1}$}
                \stringud{7.5}{\circled{}}{\circled{}}
                \stringud{8.5}{$\cdots$}{$\cdots$}
                \stringud{9.5}{\circled{}}{\circled{}}
                \stringud{10.5}{\circled{}}{\circled{}}
                \stringud{11.5}{\circled{+}}{\circled{}}
                %\commutatorhopudot{7.5}{8.5}
                \ubraceplusstring{1}{4}{$l-2-n$}{0}
                \ubraceplusstring{7.5}{11.5}{$n+1$}{0}
                %\stringlongu{1}{3}{$q_2(i-1)$}
            \end{tikzpicture}
        };
        \node (A5a) at (0,-9) {};
        \node (A5b) at (0,-10) {};
        \node[draw, rectangle, rounded corners=2mm, scale=0.7, anchor=west] (B5) at (4,-11) {
            \begin{tikzpicture}[baseline=-0.5*\completevertical, rounded corners=0mm]
                \stringud{2}{\circled{+}}{\circled{}}
                \stringud{3}{\circled{}}{\circled{}}
                \stringud{4}{$\Bar{a}_{l}$}{$\Bar{b}_{l}$}
                \stringud{5}{$\cdots$}{$\cdots$}
                \stringud{6.25}{$\Bar{a}_{k-1}$}{$\Bar{b}_{k-1}$}
                \stringud{7.5}{\circled{}}{\circled{}}
                \stringud{8.5}{$\cdots$}{$\cdots$}
                \stringud{9.5}{\circled{}}{\circled{}}
                \stringud{10.5}{\circled{+}}{\circled{}}
                %\commutatorhopudot{7.5}{8.5}
                %\ubraceplusstring{1}{4}{$l-2-m$}{0}
                \ubraceplusstring{7.5}{10.5}{$l-2$}{0}
                %\stringlongu{1}{3}{$q_2(i-1)$}
            \end{tikzpicture}
        };
        \node[draw, rectangle, rounded corners=2mm, scale=0.7, anchor=east] (A6) at (0,-12) {
            \begin{tikzpicture}[baseline=-0.5*\completevertical, rounded corners=0mm]
                \commutatorhopudot{2}{3}
                \stringud{2}{\circleddotted{}}{\circleddotted{}}
                \stringud{3}{\circled{+}}{\circled{}}
                \stringud{4}{$\Bar{a}_{l}$}{$\Bar{b}_{l}$}
                \stringud{5}{$\cdots$}{$\cdots$}
                \stringud{6.25}{$\Bar{a}_{k-1}$}{$\Bar{b}_{k-1}$}
                \stringud{7.5}{\circled{}}{\circled{}}
                \stringud{8.5}{$\cdots$}{$\cdots$}
                \stringud{9.5}{\circled{}}{\circled{}}
                \stringud{10.5}{\circled{+}}{\circled{}}
                %\commutatorhopudot{7.5}{8.5}
                %\ubraceplusstring{1}{4}{$l-2-m$}{0}
                \ubraceplusstring{7.5}{10.5}{$l-2$}{0}
                %\stringlongu{1}{3}{$q_2(i-1)$}
            \end{tikzpicture}
        };
        % Arrow from first node to the second
        \draw[->,thick] (A1) -- (B1);
        \draw[->,thick] (A2) -- (B1);
        \draw[->,thick] (A2) -- (B2);
        \draw[->,thick] (A4) -- (B3);
        \draw[->,thick] (A4) -- (B4);
        \draw[dotted,thick] (A3a) -- (A3b);
        \draw[dotted,thick] (A5a) -- (A5b);
        \draw[dotted,thick] ([xshift=6em, yshift=3.5em]B3.west) -- ([xshift=6em, yshift=-2.5em]B2.west);
        \draw[dotted,thick] ([xshift=6em, yshift=3.5em]B5.west) -- ([xshift=6em, yshift=-2.5em]B4.west);
        \draw[->,thick] (A3a) -- (B2);
        \draw[->,thick] (A3b) -- (B3);
        \draw[->,thick] (A5a) -- (B4);
        \draw[->,thick] (A5b) -- (B5);
        \draw[->,thick] (A6) -- (B5);
        \node[above=0em  of A6, scale=0.75] {$q_{l-2}(i+l-2)$};
        \node[above=0em of B5, scale=0.75] {$\widetilde{q}_{l-2}(i+l-3)$};
        \node[above=0em of B4, scale=0.75] {$\widetilde{q}_{n+1}(i+n)$};
        \node[above=0em of A4, scale=0.75] {$q_n(i+n)$};
        \node[above=0em  of B3, scale=0.75] {$\widetilde{q}_{n}(i+n-1)$};
        \node[above=0em  of B2, scale=0.75] {$\widetilde{q}_2(i+1)$};
        \node[above=0em  of B1, scale=0.75] {$\widetilde{q}_1(i)$};
        \node[above=0em  of A1, scale=0.75] {$q(i)$};
        \node[above=0em  of A2, scale=0.75] {$q_1(i+1)$};
    \end{tikzpicture}
\end{align}
where the parts omitted with a dot have a similar structure, and we have $c_i(q) \propto c_{i+1}(q_1) \propto \cdots\propto c_{i+l-2}(q_{l-2})$ from~\eqref{eq:graphical-cancellation-two-m} and $c_{i+l-2}(q_{l-2}) = 0$ from Lemma~\ref{lem:third} because of $\bce{\Bar{a}_l, \Bar{b}_l} \neq \bce{\circled{}, \circled{}}$.
%%%%%%%%%%%%%%%%%%%%%%%%%%%%%%%%%%%%%
Thus we have $c_i(q) = 0$.

%%%%%%%%%%%%%%%%%%%%%%%%%%%%%%
We can also follow the same argument for the other choices of 
$
\begin{tikzpicture}[baseline=-0.5*\completevertical]
    \stringud{0}{$\Bar{a}_1$}{$\Bar{b}_1$}
\end{tikzpicture}
$
and
$
\begin{tikzpicture}[baseline=-0.5*\completevertical]
    \stringud{0}{$\Bar{a}_k$}{$\Bar{b}_k$}
\end{tikzpicture}
$.
%%%%%%%%%%%%%%%%%%%%%%%%%%%%%%
Then we have proved  Lemma~\ref{lem:forth}.
\end{proof}

%%%%%%%%%%%%%%%%%%%%%%%%%%%%%
From Lemmas~\ref{lem:first}--\ref{lem:forth}, we can see $c_{i}(q) = 0\ (q \in \mathcal{C}_k)$ unless $q$ is the following form:
\begin{align}
    \label{eq:non-zero-coeff-q}
    q
    =
    \begin{tikzpicture}[baseline=-0.5*\completevertical]
        \stringud{0}{$\Bar{a}_1$}{$\Bar{b}_1$}
        \circledcol{1}
        \stringud{2}{$\cdots$}{$\cdots$}
        \circledcol{3}
        \stringud{4}{$\Bar{a}_k$}{$\Bar{b}_k$}
        \ubraceplusstring{1}{4}{$k-2$}{0}
    \end{tikzpicture}
    ,
\end{align}
where 
$
\begin{tikzpicture}[baseline=-0.5*\completevertical]
    \stringud{0}{$\Bar{a}_1$}{$\Bar{b}_1$}
\end{tikzpicture}
, 
\begin{tikzpicture}[baseline=-0.5*\completevertical]
    \stringud{0}{$\Bar{a}_{k}$}{$\Bar{b}_{k}$}
\end{tikzpicture}
\in
\bce{
    \begin{tikzpicture}[baseline=-0.5*\completevertical]
        \stringud{0}{$\circled{\pm}$}{$\circled{}$}
    \end{tikzpicture}
    ,
    \begin{tikzpicture}[baseline=-0.5*\completevertical]
        \stringud{0}{$\circled{}$}{$\circled{\pm}$}
    \end{tikzpicture}
}
$.
%%%%%%%%%%%%%%%%%%%%%%%%%%%%%%
We next study the coefficients of the basis elements of these configurations.

\begin{lemma}
    \label{lem:fifth}
    $
    c_{i}
    \biggl(
        \begin{tikzpicture}[baseline=-1.25*\completevertical]
            \node[scale=0.8]{
                \begin{tikzpicture}[baseline=0*\completevertical]
                    \stringud{-2}{\circled{\pm}}{\circled{}}
                    \stringud{-1}{\circled{}}{\circled{}}
                    \stringud{0}{$\cdots$}{$\cdots$}
                    \stringud{1}{\circled{}}{\circled{}}
                    \stringud{2}{\circled{\mp}}{\circled{}}
                    %\stringud{3}{\circleddotted{}}{\circleddotted{}}
                    \ubraceplusstring{-1}{2}{$k-2$}{0}
                    %\commutatorhopu{2}{3}
                \end{tikzpicture}
            };
        \end{tikzpicture}
    \biggr)
    $ 
    ,
    $
    c_{i}
    \biggl(
        \begin{tikzpicture}[baseline=-1.25*\completevertical]
            \node[scale=0.8]{
                \begin{tikzpicture}[baseline=0*\completevertical]
                    \stringud{-2}{\circled{\pm}}{\circled{}}
                    \stringud{-1}{\circled{}}{\circled{}}
                    \stringud{0}{$\cdots$}{$\cdots$}
                    \stringud{1}{\circled{}}{\circled{}}
                    \stringud{2}{\circled{}}{\circled{\mp}}
                    %\stringud{3}{\circleddotted{}}{\circleddotted{}}
                    \ubraceplusstring{-1}{2}{$k-2$}{0}
                    %\commutatorhopu{2}{3}
                \end{tikzpicture}
            };
        \end{tikzpicture}
    \biggr)
    $
    ,
    $
    c_{i}
    \biggl(
        \begin{tikzpicture}[baseline=-1.25*\completevertical]
            \node[scale=0.8]{
                \begin{tikzpicture}[baseline=0*\completevertical]
                    \stringud{-2}{\circled{}}{\circled{\pm}}
                    \stringud{-1}{\circled{}}{\circled{}}
                    \stringud{0}{$\cdots$}{$\cdots$}
                    \stringud{1}{\circled{}}{\circled{}}
                    \stringud{2}{\circled{}}{\circled{\mp}}
                    %\stringud{3}{\circleddotted{}}{\circleddotted{}}
                    \ubraceplusstring{-1}{2}{$k-2$}{0}
                    %\commutatorhopu{2}{3}
                \end{tikzpicture}
            };
        \end{tikzpicture}
    \biggr)
    $
    and 
    $
    c_{i}
    \biggl(
        \begin{tikzpicture}[baseline=-1.25*\completevertical]
            \node[scale=0.8]{
                \begin{tikzpicture}[baseline=0*\completevertical]
                    \stringud{-2}{\circled{}}{\circled{\pm}}
                    \stringud{-1}{\circled{}}{\circled{}}
                    \stringud{0}{$\cdots$}{$\cdots$}
                    \stringud{1}{\circled{}}{\circled{}}
                    \stringud{2}{\circled{\mp}}{\circled{}}
                    %\stringud{3}{\circleddotted{}}{\circleddotted{}}
                    \ubraceplusstring{-1}{2}{$k-2$}{0}
                    %\commutatorhopu{2}{3}
                \end{tikzpicture}
            };
        \end{tikzpicture}
    \biggr)
    $ are constants independent of~$i$.
\end{lemma}
\begin{proof}
We prove 
$
c_{i}
\biggl(
    \begin{tikzpicture}[baseline=-1.25*\completevertical]
        \node[scale=0.8]{
            \begin{tikzpicture}[baseline=0*\completevertical]
                \stringud{-2}{\circled{\pm}}{\circled{}}
                \stringud{-1}{\circled{}}{\circled{}}
                \stringud{0}{$\cdots$}{$\cdots$}
                \stringud{1}{\circled{}}{\circled{}}
                \stringud{2}{\circled{\mp}}{\circled{}}
                %\stringud{3}{\circleddotted{}}{\circleddotted{}}
                \ubraceplusstring{-1}{2}{$k-2$}{0}
                %\commutatorhopu{2}{3}
            \end{tikzpicture}
        };
    \end{tikzpicture}
\biggr)
$ is independent of $i$.
%%%%%%%%%%%%%%%%%%%%%%%%%%%%%%
We consider the following cancellation of the $(k+1)$-support basis element $\widetilde{q}(i)$:
%%%%%%%%%%%%%%%%%%%%%%%%%%%%%%
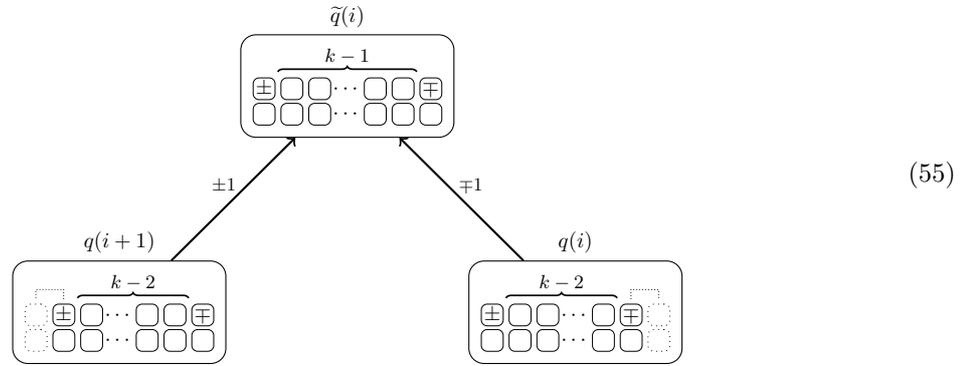
\begin{align}
    \label{eq:cancellation_eg_pmmp}
    \begin{tikzpicture}[baseline=(current bounding box.center)]
        % Second TikZ picture in a node
        \node[draw, rectangle, rounded corners=2mm, scale=0.8] (B) at (3,3) {
            \begin{tikzpicture}[baseline=-0.5*\completevertical, rounded corners=0mm]
                \stringud{-2}{\circled{\pm}}{\circled{}}
                \stringud{-1}{\circled{}}{\circled{}}
                \stringud{0}{\circled{}}{\circled{}}
                \stringud{1}{$\cdots$}{$\cdots$}
                \stringud{2}{\circled{}}{\circled{}}
                \stringud{3}{\circled{}}{\circled{}}
                \stringud{4}{\circled{\mp}}{\circled{}}
                %\stringud{3}{\circleddotted{}}{\circleddotted{}}
                \ubraceplusstring{-1}{4}{$k-1$}{0}
                %\commutatorhopu{2}{3}
            \end{tikzpicture}
        };
        \node[draw, rectangle, rounded corners=2mm, scale=0.8] (A1) at (0,0) {
            \begin{tikzpicture}[baseline=-0.5*\completevertical, rounded corners=0mm]
                \stringud{-3}{\circleddotted{}}{\circleddotted{}}
                \stringud{-2}{\circled{\pm}}{\circled{}}
                \stringud{-1}{\circled{}}{\circled{}}
                \stringud{0}{$\cdots$}{$\cdots$}
                \stringud{1}{\circled{}}{\circled{}}
                \stringud{2}{\circled{}}{\circled{}}
                \stringud{3}{\circled{\mp}}{\circled{}}
                \ubraceplusstring{-1}{3}{$k-2$}{0}
                \commutatorhopudot{-2}{-3}
                %\commutatorhopu{2}{3}
            \end{tikzpicture}
        };
        % First TikZ picture in a node
        \node[draw, rectangle, rounded corners=2mm, scale=0.8] (A2) at (6,0) {
            \begin{tikzpicture}[baseline=-0.5*\completevertical, rounded corners=0mm]
                \stringud{-3}{\circled{\pm}}{\circled{}}
                \stringud{-2}{\circled{}}{\circled{}}
                \stringud{-1}{\circled{}}{\circled{}}
                \stringud{0}{$\cdots$}{$\cdots$}
                \stringud{1}{\circled{}}{\circled{}}
                \stringud{2}{\circled{\mp}}{\circled{}}
                \stringud{3}{\circleddotted{}}{\circleddotted{}}
                \ubraceplusstring{-2}{2}{$k-2$}{0}
                \commutatorhopudot{2}{3}
            \end{tikzpicture}
        };
        % Arrow from first node to the second
        \draw[->,thick] (A1) -- (B) node[midway, scale=0.7,above, xshift=-5pt] {$\pm1$};
        \draw[->,thick] (A2) -- (B) node[midway, scale=0.7,above, xshift=5pt] {$\mp1$};;
        \node[above=0em  of B, scale=0.8] {$\widetilde{q}(i)$};
        \node[above=0em  of A1, scale=0.8] {$q(i+1)$};
        \node[above=0em  of A2, scale=0.8] {$q(i)$};
    \end{tikzpicture}
\end{align}
and we have 
$
c_{i}
\biggl(
    \begin{tikzpicture}[baseline=-1.25*\completevertical]
        \node[scale=0.8]{
            \begin{tikzpicture}[baseline=0*\completevertical]
                \stringud{-2}{\circled{\pm}}{\circled{}}
                \stringud{-1}{\circled{}}{\circled{}}
                \stringud{0}{$\cdots$}{$\cdots$}
                \stringud{1}{\circled{}}{\circled{}}
                \stringud{2}{\circled{\mp}}{\circled{}}
                %\stringud{3}{\circleddotted{}}{\circleddotted{}}
                \ubraceplusstring{-1}{2}{$k-2$}{0}
                %\commutatorhopu{2}{3}
            \end{tikzpicture}
        };
    \end{tikzpicture}
\biggr)
=
c_{i+1}
\biggl(
    \begin{tikzpicture}[baseline=-1.25*\completevertical]
        \node[scale=0.8]{
            \begin{tikzpicture}[baseline=0*\completevertical]
                \stringud{-2}{\circled{\pm}}{\circled{}}
                \stringud{-1}{\circled{}}{\circled{}}
                \stringud{0}{$\cdots$}{$\cdots$}
                \stringud{1}{\circled{}}{\circled{}}
                \stringud{2}{\circled{\mp}}{\circled{}}
                %\stringud{3}{\circleddotted{}}{\circleddotted{}}
                \ubraceplusstring{-1}{2}{$k-2$}{0}
                %\commutatorhopu{2}{3}
            \end{tikzpicture}
        };
    \end{tikzpicture}
\biggr)
$. 
%%%%%%%%%%%%%%%%%%%%%%%%%%%%%%
Thus, we have proved that 
$
c_{i}
\biggl(
    \begin{tikzpicture}[baseline=-1.25*\completevertical]
        \node[scale=0.8]{
            \begin{tikzpicture}[baseline=0*\completevertical]
                \stringud{-2}{\circled{\pm}}{\circled{}}
                \stringud{-1}{\circled{}}{\circled{}}
                \stringud{0}{$\cdots$}{$\cdots$}
                \stringud{1}{\circled{}}{\circled{}}
                \stringud{2}{\circled{\mp}}{\circled{}}
                %\stringud{3}{\circleddotted{}}{\circleddotted{}}
                \ubraceplusstring{-1}{2}{$k-2$}{0}
                %\commutatorhopu{2}{3}
            \end{tikzpicture}
        };
    \end{tikzpicture}
\biggr)
$ is independent of $i$.
%%%%%%%%%%%%%%%%%%%%%%%%%%%%%%
The other cases can also be proved in the same way.
\end{proof}

\begin{lemma}
    \label{lem:sixth}
    $
    c_{i}
    \biggl(
        \begin{tikzpicture}[baseline=-1.25*\completevertical]
            \node[scale=0.8]{
                \begin{tikzpicture}[baseline=0*\completevertical]
                    \stringud{-2}{\circled{\pm}}{\circled{}}
                    \stringud{-1}{\circled{}}{\circled{}}
                    \stringud{0}{$\cdots$}{$\cdots$}
                    \stringud{1}{\circled{}}{\circled{}}
                    \stringud{2}{\circled{\pm}}{\circled{}}
                    %\stringud{3}{\circleddotted{}}{\circleddotted{}}
                    \ubraceplusstring{-1}{2}{$k-2$}{0}
                    %\commutatorhopu{2}{3}
                \end{tikzpicture}
            };
        \end{tikzpicture}
    \biggr)
    $ 
    ,
    $
    c_{i}
    \biggl(
        \begin{tikzpicture}[baseline=-1.25*\completevertical]
            \node[scale=0.8]{
                \begin{tikzpicture}[baseline=0*\completevertical]
                    \stringud{-2}{\circled{\pm}}{\circled{}}
                    \stringud{-1}{\circled{}}{\circled{}}
                    \stringud{0}{$\cdots$}{$\cdots$}
                    \stringud{1}{\circled{}}{\circled{}}
                    \stringud{2}{\circled{}}{\circled{\pm}}
                    %\stringud{3}{\circleddotted{}}{\circleddotted{}}
                    \ubraceplusstring{-1}{2}{$k-2$}{0}
                    %\commutatorhopu{2}{3}
                \end{tikzpicture}
            };
        \end{tikzpicture}
    \biggr)
    $
    ,
    $
    c_{i}
    \biggl(
        \begin{tikzpicture}[baseline=-1.25*\completevertical]
            \node[scale=0.8]{
                \begin{tikzpicture}[baseline=0*\completevertical]
                    \stringud{-2}{\circled{}}{\circled{\pm}}
                    \stringud{-1}{\circled{}}{\circled{}}
                    \stringud{0}{$\cdots$}{$\cdots$}
                    \stringud{1}{\circled{}}{\circled{}}
                    \stringud{2}{\circled{}}{\circled{\pm}}
                    %\stringud{3}{\circleddotted{}}{\circleddotted{}}
                    \ubraceplusstring{-1}{2}{$k-2$}{0}
                    %\commutatorhopu{2}{3}
                \end{tikzpicture}
            };
        \end{tikzpicture}
    \biggr)
    $
    and 
    $
    c_{i}
    \biggl(
        \begin{tikzpicture}[baseline=-1.25*\completevertical]
            \node[scale=0.8]{
                \begin{tikzpicture}[baseline=0*\completevertical]
                    \stringud{-2}{\circled{}}{\circled{\pm}}
                    \stringud{-1}{\circled{}}{\circled{}}
                    \stringud{0}{$\cdots$}{$\cdots$}
                    \stringud{1}{\circled{}}{\circled{}}
                    \stringud{2}{\circled{\pm}}{\circled{}}
                    %\stringud{3}{\circleddotted{}}{\circleddotted{}}
                    \ubraceplusstring{-1}{2}{$k-2$}{0}
                    %\commutatorhopu{2}{3}
                \end{tikzpicture}
            };
        \end{tikzpicture}
    \biggr)
    $ are zero for the odd $L$ case and are $(\textit{some constant})\times(-1)^{i}$ for the even $L$ case.
\end{lemma}
\begin{proof}
We give the proof for
$
c_{i}
\biggl(
    \begin{tikzpicture}[baseline=-1.25*\completevertical]
        \node[scale=0.8]{
            \begin{tikzpicture}[baseline=0*\completevertical]
                \stringud{-2}{\circled{\pm}}{\circled{}}
                \stringud{-1}{\circled{}}{\circled{}}
                \stringud{0}{$\cdots$}{$\cdots$}
                \stringud{1}{\circled{}}{\circled{}}
                \stringud{2}{\circled{\pm}}{\circled{}}
                %\stringud{3}{\circleddotted{}}{\circleddotted{}}
                \ubraceplusstring{-1}{2}{$k-2$}{0}
                %\commutatorhopu{2}{3}
            \end{tikzpicture}
        };
    \end{tikzpicture}
\biggr)
$.
%%%%%%%%%%%%%%%%%%%%%%%%%%%%%%
We consider the following cancellation of the $(k+1)$-support basis element~$\widetilde{q}(i)$:
%%%%%%%%%%%%%%%%%%%%%%%%%%%%%%
\begin{align}
    \begin{tikzpicture}[baseline=(current bounding box.center)]
        % Second TikZ picture in a node
        \node[draw, rectangle, rounded corners=2mm, scale=0.8] (B) at (3,3) {
            \begin{tikzpicture}[baseline=-0.5*\completevertical, rounded corners=0mm]
                \stringud{-2}{\circled{\pm}}{\circled{}}
                \stringud{-1}{\circled{}}{\circled{}}
                \stringud{0}{\circled{}}{\circled{}}
                \stringud{1}{$\cdots$}{$\cdots$}
                \stringud{2}{\circled{}}{\circled{}}
                \stringud{3}{\circled{}}{\circled{}}
                \stringud{4}{\circled{\pm}}{\circled{}}
                %\stringud{3}{\circleddotted{}}{\circleddotted{}}
                \ubraceplusstring{-1}{4}{$k-1$}{0}
                %\commutatorhopu{2}{3}
            \end{tikzpicture}
        };
        \node[draw, rectangle, rounded corners=2mm, scale=0.8] (A1) at (0,0) {
            \begin{tikzpicture}[baseline=-0.5*\completevertical, rounded corners=0mm]
                \stringud{-3}{\circleddotted{}}{\circleddotted{}}
                \stringud{-2}{\circled{\pm}}{\circled{}}
                \stringud{-1}{\circled{}}{\circled{}}
                \stringud{0}{$\cdots$}{$\cdots$}
                \stringud{1}{\circled{}}{\circled{}}
                \stringud{2}{\circled{}}{\circled{}}
                \stringud{3}{\circled{\pm}}{\circled{}}
                \ubraceplusstring{-1}{3}{$k-2$}{0}
                \commutatorhopudot{-2}{-3}
                %\commutatorhopu{2}{3}
            \end{tikzpicture}
        };
        % First TikZ picture in a node
        \node[draw, rectangle, rounded corners=2mm, scale=0.8] (A2) at (6,0) {
            \begin{tikzpicture}[baseline=-0.5*\completevertical, rounded corners=0mm]
                \stringud{-3}{\circled{\pm}}{\circled{}}
                \stringud{-2}{\circled{}}{\circled{}}
                \stringud{-1}{\circled{}}{\circled{}}
                \stringud{0}{$\cdots$}{$\cdots$}
                \stringud{1}{\circled{}}{\circled{}}
                \stringud{2}{\circled{\pm}}{\circled{}}
                \stringud{3}{\circleddotted{}}{\circleddotted{}}
                \ubraceplusstring{-2}{2}{$k-2$}{0}
                \commutatorhopudot{2}{3}
            \end{tikzpicture}
        };
        % Arrow from first node to the second
        \draw[->,thick] (A1) -- (B) node[midway, scale=0.7,above, xshift=-5pt] {$\pm1$};
        \draw[->,thick] (A2) -- (B) node[midway, scale=0.7,above, xshift=5pt] {$\pm1$};
        \node[above=0em  of B, scale=0.8] {$\widetilde{q}(i)$};
        \node[above=0em  of A1, scale=0.8] {$q(i+1)$};
        \node[above=0em  of A2, scale=0.8] {$q(i)$};
    \end{tikzpicture}
\end{align}
and we have 
$
c_{i}
\biggl(
    \begin{tikzpicture}[baseline=-1.25*\completevertical]
        \node[scale=0.8]{
            \begin{tikzpicture}[baseline=0*\completevertical]
                \stringud{-2}{\circled{\pm}}{\circled{}}
                \stringud{-1}{\circled{}}{\circled{}}
                \stringud{0}{$\cdots$}{$\cdots$}
                \stringud{1}{\circled{}}{\circled{}}
                \stringud{2}{\circled{\pm}}{\circled{}}
                %\stringud{3}{\circleddotted{}}{\circleddotted{}}
                \ubraceplusstring{-1}{2}{$k-2$}{0}
                %\commutatorhopu{2}{3}
            \end{tikzpicture}
        };
    \end{tikzpicture}
\biggr)
=
-
c_{i+1}
\biggl(
    \begin{tikzpicture}[baseline=-1.25*\completevertical]
        \node[scale=0.8]{
            \begin{tikzpicture}[baseline=0*\completevertical]
                \stringud{-2}{\circled{\pm}}{\circled{}}
                \stringud{-1}{\circled{}}{\circled{}}
                \stringud{0}{$\cdots$}{$\cdots$}
                \stringud{1}{\circled{}}{\circled{}}
                \stringud{2}{\circled{\pm}}{\circled{}}
                %\stringud{3}{\circleddotted{}}{\circleddotted{}}
                \ubraceplusstring{-1}{2}{$k-2$}{0}
                %\commutatorhopu{2}{3}
            \end{tikzpicture}
        };
    \end{tikzpicture}
\biggr)
$. 
%%%%%%%%%%%%%%%%%%%%%%%%%%%%%%
Because of the periodic boundary condition, we have
$
c_{i}
\biggl(
    \begin{tikzpicture}[baseline=-1.25*\completevertical]
        \node[scale=0.8]{
            \begin{tikzpicture}[baseline=0*\completevertical]
                \stringud{-2}{\circled{\pm}}{\circled{}}
                \stringud{-1}{\circled{}}{\circled{}}
                \stringud{0}{$\cdots$}{$\cdots$}
                \stringud{1}{\circled{}}{\circled{}}
                \stringud{2}{\circled{\pm}}{\circled{}}
                %\stringud{3}{\circleddotted{}}{\circleddotted{}}
                \ubraceplusstring{-1}{2}{$k-2$}{0}
                %\commutatorhopu{2}{3}
            \end{tikzpicture}
        };
    \end{tikzpicture}
\biggr)
=
(-1)^{L}
c_{i}
\biggl(
    \begin{tikzpicture}[baseline=-1.25*\completevertical]
        \node[scale=0.8]{
            \begin{tikzpicture}[baseline=0*\completevertical]
                \stringud{-2}{\circled{\pm}}{\circled{}}
                \stringud{-1}{\circled{}}{\circled{}}
                \stringud{0}{$\cdots$}{$\cdots$}
                \stringud{1}{\circled{}}{\circled{}}
                \stringud{2}{\circled{\pm}}{\circled{}}
                %\stringud{3}{\circleddotted{}}{\circleddotted{}}
                \ubraceplusstring{-1}{2}{$k-2$}{0}
                %\commutatorhopu{2}{3}
            \end{tikzpicture}
        };
    \end{tikzpicture}
\biggr)
$.
%%%%%%%%%%%%%%%%%%%%%%%%%%%%%%
Thus, we have proved that 
$
c_{i}
\biggl(
    \begin{tikzpicture}[baseline=-1.25*\completevertical]
        \node[scale=0.8]{
            \begin{tikzpicture}[baseline=0*\completevertical]
                \stringud{-2}{\circled{\pm}}{\circled{}}
                \stringud{-1}{\circled{}}{\circled{}}
                \stringud{0}{$\cdots$}{$\cdots$}
                \stringud{1}{\circled{}}{\circled{}}
                \stringud{2}{\circled{\pm}}{\circled{}}
                %\stringud{3}{\circleddotted{}}{\circleddotted{}}
                \ubraceplusstring{-1}{2}{$k-2$}{0}
                %\commutatorhopu{2}{3}
            \end{tikzpicture}
        };
    \end{tikzpicture}
\biggr)
$
is zero for odd $L$ and is $(\text{some constant})\times(-1)^i$ for even $L$ case.
%%%%%%%%%%%%%%%%%%%%%%%%%%%%%%
The other cases can also be proved in the same way.
\end{proof}

%%%%%%%%%%%%%%%%%%%%%%%%%%%%%%
From Lemma~\ref{lem:fifth} and Lemma~\ref{lem:sixth}, $F_{k}^{ k}$ is written as
\begin{align}
    \label{eq:Fkk}
    F_{k}^{ k}
    &=
    \sum_{\sigma, \mu\in\{\uparrow, \downarrow\}}
    \sum_{s\in\{+,-\}}
    \alpha_{(\sigma, s),(\mu, -s)}
    \sum_{i=1}^{L}
    c_{i, \sigma}^{s}c_{i+k-1, \mu}^{-s}
    \nonumber\\
    &\hspace*{3em}
    +
    \sum_{\sigma, \mu\in\{\uparrow, \downarrow\}}
    \sum_{s\in\{+,-\}}
    \alpha_{(\sigma, s),(\mu, s)}
    \sum_{i=1}^{L}
    (-1)^{i}
    c_{i, \sigma}^{s}c_{i+k-1, \mu}^{s}
    ,
\end{align}
where $c^{+}_{i,\sigma}\equiv c^{\dag}_{i,\sigma}, c^{-}_{i,\sigma}\equiv c_{i,\sigma}$ and $-s\equiv\mp$ for $s=\pm$, and $\alpha_{(\sigma, s),(\mu, -s)}$ is arbitrary constant, and $\alpha_{(\sigma, s),(\mu, s)}$ is arbitrary constant for even $L$ and $\alpha_{(\sigma, s),(\mu, s)}=0$ for odd $L$.

%%%%%%%%%%%%%%%%%%%%%%%%%%%%%%
Note that $F_{k}^{ k}$ commute with $H_0$: 
$
\bck{
        F_{k}^{ k}, H_0
    }
    =0
$.
%%%%%%%%%%%%%%%%%%%%%%%%%%%%%%
Therefore, the equation for the cancellation of the $k$-support operators~\eqref{eq:k-cancellation} can be reduced to 
\begin{align}
    \label{eq:k-cancellation2}
    \bck{F_k^{ k-1},H_0}
    \bigg|_{k}
    +
    \bck{F_k^{  k},H_{\mathrm{int}}}
    =0. 
\end{align}

\subsection{Cancellation of $k$-support operators}
%%%%%%%%%%%%%%%%%%%%%%%%%%%%%%%%%%%%%%%%
In the following, we determine $\alpha_{(\sigma, s),(\mu, s^\prime)}$ to satisfy the cancellation of $k$-support operators~\eqref{eq:k-cancellation2}.
%%%%%%%%%%%%%%%%%%%%%%%%%%%%%%%%%%%%%%%%
We let a $(k-1)$-support configuration denoted by the symbol $p$ below.
%%%%%%%%%%%%%%%%%%%%%%%%%%%%%%%%%%%%%%%%
$F^{ k-1}_k$ is written as
\begin{align}
        \label{eq:k-1-sup-sum}
        F^{ k-1}_k
        =
        \sum_{p\in \mathcal{C}_{k-1}}
        \sum_{i=1}^L
        c_{i}(p)
        p(i)
        ,
\end{align}
where $p \in \mathcal{C}_{k-1}$ is denoted as
$
    p
    =
    \begin{tikzpicture}[baseline=-0.5*\completevertical]
        \stringud{0}{$\Bar{a}_1$}{$\Bar{b}_1$}
        \stringud{1}{$\Bar{a}_2$}{$\Bar{b}_2$}
        \stringud{2}{$\cdots$}{$\cdots$}
        \stringud{3.25}{$\Bar{a}_{k-1}$}{$\Bar{b}_{k-1}$}
    \end{tikzpicture}
$
, and at least one of $\Bar{a}_1$ or $\Bar{b}_1$ and at least one of $\Bar{a}_{k-1}$ or $\Bar{b}_{k-1}$ is not $\circled{}$ because $p\in \mathcal{C}_{k-1}$ is $(k-1)$-support configuration.

%%%%%%%%%%%%%%%%%%%%%%%%%%%%%%%%%%%%%%%%
From Lemma~\ref{lem:fifth} and Lemma~\ref{lem:sixth}, we can calculate explicitly $\bck{F_k^k, H_{\mathrm{int}}}$ as
\begin{align}
    \label{eq:Fkkint}
    &\bck{F_k^k, H_{\mathrm{int}}}=
    \nonumber\\
    &
    -
    U
    \sum_{s, t\in\bce{+, -}}
    \bce{
    \alpha_{(\uparrow, s), (\uparrow, t)}
        \biggl(
            s
            \begin{tikzpicture}[baseline=-1.25*\completevertical]
                \node[scale=0.8]{
                    \begin{tikzpicture}[baseline=0*\completevertical]
                        \stringud{-2}{\circled{s}}{\circled{z}}
                        \stringud{-1}{\circled{}}{\circled{}}
                        \stringud{0}{$\cdots$}{$\cdots$}
                        \stringud{1}{\circled{}}{\circled{}}
                        \stringud{2}{\circled{t}}{\circled{}}
                        %\stringud{3}{\circleddotted{}}{\circleddotted{}}
                        \ubraceplusstring{-1}{2}{$k-2$}{0}
                        %\commutatorhopu{2}{3}
                    \end{tikzpicture}
                };
            \end{tikzpicture}
            +
            t
            \begin{tikzpicture}[baseline=-1.25*\completevertical]
                \node[scale=0.8]{
                    \begin{tikzpicture}[baseline=0*\completevertical]
                        \stringud{-2}{\circled{s}}{\circled{}}
                        \stringud{-1}{\circled{}}{\circled{}}
                        \stringud{0}{$\cdots$}{$\cdots$}
                        \stringud{1}{\circled{}}{\circled{}}
                        \stringud{2}{\circled{t}}{\circled{z}}
                        %\stringud{3}{\circleddotted{}}{\circleddotted{}}
                        \ubraceplusstring{-1}{2}{$k-2$}{0}
                        %\commutatorhopu{2}{3}
                    \end{tikzpicture}
                };
            \end{tikzpicture}
        \biggr)
        %\linebreaklr
        %\hspace*{4em}
        +
        \alpha_{(\downarrow, s), (\downarrow, t)}
        \biggl(
            s
            \begin{tikzpicture}[baseline=-1.25*\completevertical]
                \node[scale=0.8]{
                    \begin{tikzpicture}[baseline=0*\completevertical]
                        \stringud{-2}{\circled{z}}{\circled{s}}
                        \stringud{-1}{\circled{}}{\circled{}}
                        \stringud{0}{$\cdots$}{$\cdots$}
                        \stringud{1}{\circled{}}{\circled{}}
                        \stringud{2}{\circled{}}{\circled{t}}
                        %\stringud{3}{\circleddotted{}}{\circleddotted{}}
                        \ubraceplusstring{-1}{2}{$k-2$}{0}
                        %\commutatorhopu{2}{3}
                    \end{tikzpicture}
                };
            \end{tikzpicture}
            +
            t
            \begin{tikzpicture}[baseline=-1.25*\completevertical]
                \node[scale=0.8]{
                    \begin{tikzpicture}[baseline=0*\completevertical]
                        \stringud{-2}{\circled{}}{\circled{s}}
                        \stringud{-1}{\circled{}}{\circled{}}
                        \stringud{0}{$\cdots$}{$\cdots$}
                        \stringud{1}{\circled{}}{\circled{}}
                        \stringud{2}{\circled{z}}{\circled{t}}
                        %\stringud{3}{\circleddotted{}}{\circleddotted{}}
                        \ubraceplusstring{-1}{2}{$k-2$}{0}
                        %\commutatorhopu{2}{3}
                    \end{tikzpicture}
                };
            \end{tikzpicture}
        \biggr)
    \linebreaklr
    +
    \alpha_{(\uparrow, s), (\downarrow, t)}
        \biggl(
            s
            \begin{tikzpicture}[baseline=-1.25*\completevertical]
                \node[scale=0.8]{
                    \begin{tikzpicture}[baseline=0*\completevertical]
                        \stringud{-2}{\circled{s}}{\circled{z}}
                        \stringud{-1}{\circled{}}{\circled{}}
                        \stringud{0}{$\cdots$}{$\cdots$}
                        \stringud{1}{\circled{}}{\circled{}}
                        \stringud{2}{\circled{}}{\circled{t}}
                        %\stringud{3}{\circleddotted{}}{\circleddotted{}}
                        \ubraceplusstring{-1}{2}{$k-2$}{0}
                        %\commutatorhopu{2}{3}
                    \end{tikzpicture}
                };
            \end{tikzpicture}
            +
            t
            \begin{tikzpicture}[baseline=-1.25*\completevertical]
                \node[scale=0.8]{
                    \begin{tikzpicture}[baseline=0*\completevertical]
                        \stringud{-2}{\circled{s}}{\circled{}}
                        \stringud{-1}{\circled{}}{\circled{}}
                        \stringud{0}{$\cdots$}{$\cdots$}
                        \stringud{1}{\circled{}}{\circled{}}
                        \stringud{2}{\circled{z}}{\circled{t}}
                        %\stringud{3}{\circleddotted{}}{\circleddotted{}}
                        \ubraceplusstring{-1}{2}{$k-2$}{0}
                        %\commutatorhopu{2}{3}
                    \end{tikzpicture}
                };
            \end{tikzpicture}
        \biggr)
        %\linebreaklr
        %\hspace*{4em}
        +
        \alpha_{(\downarrow, s), (\uparrow, t)}
        \biggl(
            s
            \begin{tikzpicture}[baseline=-1.25*\completevertical]
                \node[scale=0.8]{
                    \begin{tikzpicture}[baseline=0*\completevertical]
                        \stringud{-2}{\circled{z}}{\circled{s}}
                        \stringud{-1}{\circled{}}{\circled{}}
                        \stringud{0}{$\cdots$}{$\cdots$}
                        \stringud{1}{\circled{}}{\circled{}}
                        \stringud{2}{\circled{t}}{\circled{}}
                        %\stringud{3}{\circleddotted{}}{\circleddotted{}}
                        \ubraceplusstring{-1}{2}{$k-2$}{0}
                        %\commutatorhopu{2}{3}
                    \end{tikzpicture}
                };
            \end{tikzpicture}
            +
            t
            \begin{tikzpicture}[baseline=-1.25*\completevertical]
                \node[scale=0.8]{
                    \begin{tikzpicture}[baseline=0*\completevertical]
                        \stringud{-2}{\circled{}}{\circled{s}}
                        \stringud{-1}{\circled{}}{\circled{}}
                        \stringud{0}{$\cdots$}{$\cdots$}
                        \stringud{1}{\circled{}}{\circled{}}
                        \stringud{2}{\circled{t}}{\circled{z}}
                        %\stringud{3}{\circleddotted{}}{\circleddotted{}}
                        \ubraceplusstring{-1}{2}{$k-2$}{0}
                        %\commutatorhopu{2}{3}
                    \end{tikzpicture}
                };
            \end{tikzpicture}
        \biggr)
    }
    .
\end{align}
%%%%%%%%%%%%%%%%%%%%%%%%%%%%%%%%%%%%%%%%
We note that the contributions to the cancellation in~\eqref{eq:k-cancellation2} from $F_k^k$ are exhausted by the above operators in~\eqref{eq:Fkkint}.

%%%%%%%%%%%%%%%%%%%%%%%%%%%%%%
We can prove the coefficients are zero if the spin flavors differ between the right and left ends, such as $
c_i\paren{
    \begin{tikzpicture}[baseline=-0.5*\completevertical]
        \stringud{0}{\circled{+}}{\circled{}}
        \stringud{1}{$\cdots$}{$\cdots$}
        \stringud{2}{\circled{}}{\circled{-}}
    \end{tikzpicture}
}
=
0
$, which is stated in the following Lemma.
\begin{lemma}
    \label{lem:seventh}
    $\alpha_{(\sigma, s),(\mu, s^\prime)} = 0$ holds for $\sigma \neq \mu$.
\end{lemma}
\begin{proof}
    We give the proof in the case of $\sigma=\uparrow$, $\mu = \downarrow$, $s = +$, and $s^\prime = -$.
%%%%%%%%%%%%%%%%%%%%%%%%%%%%%%
Let $q\in \mathcal{C}_{k}$ be 
$
q 
= 
\begin{tikzpicture}[baseline=-1.25*\completevertical]
    \node[scale=0.8]{
        \begin{tikzpicture}[baseline=-0.5*\completevertical, rounded corners=0mm]
            \stringud{-2}{\circled{+}}{\circled{}}
            \stringud{-1}{\circled{}}{\circled{}}
            \stringud{0}{$\cdots$}{$\cdots$}
            \stringud{1}{\circled{}}{\circled{}}
            \stringud{2}{\circled{}}{\circled{-}}
            \ubraceplusstring{-1}{2}{$k-2$}{0}
            %\commutatorhopu{3.5}{4.5}
            %\commutatorintdot{2}
        \end{tikzpicture}
    };
\end{tikzpicture}
$
and from Lemma~\ref{lem:fifth}, we have $c_{i}(q) = \alpha_{(\uparrow, +),(\downarrow, -)}$.
%%%%%%%%%%%%%%%%%%%%%%%%%%%%%%
We consider the following cancellation of the $k$-support basis element $\widetilde{q}_l(i+l-1)\ (1\leq l \leq k)$ in~\eqref{eq:k-cancellation2}:
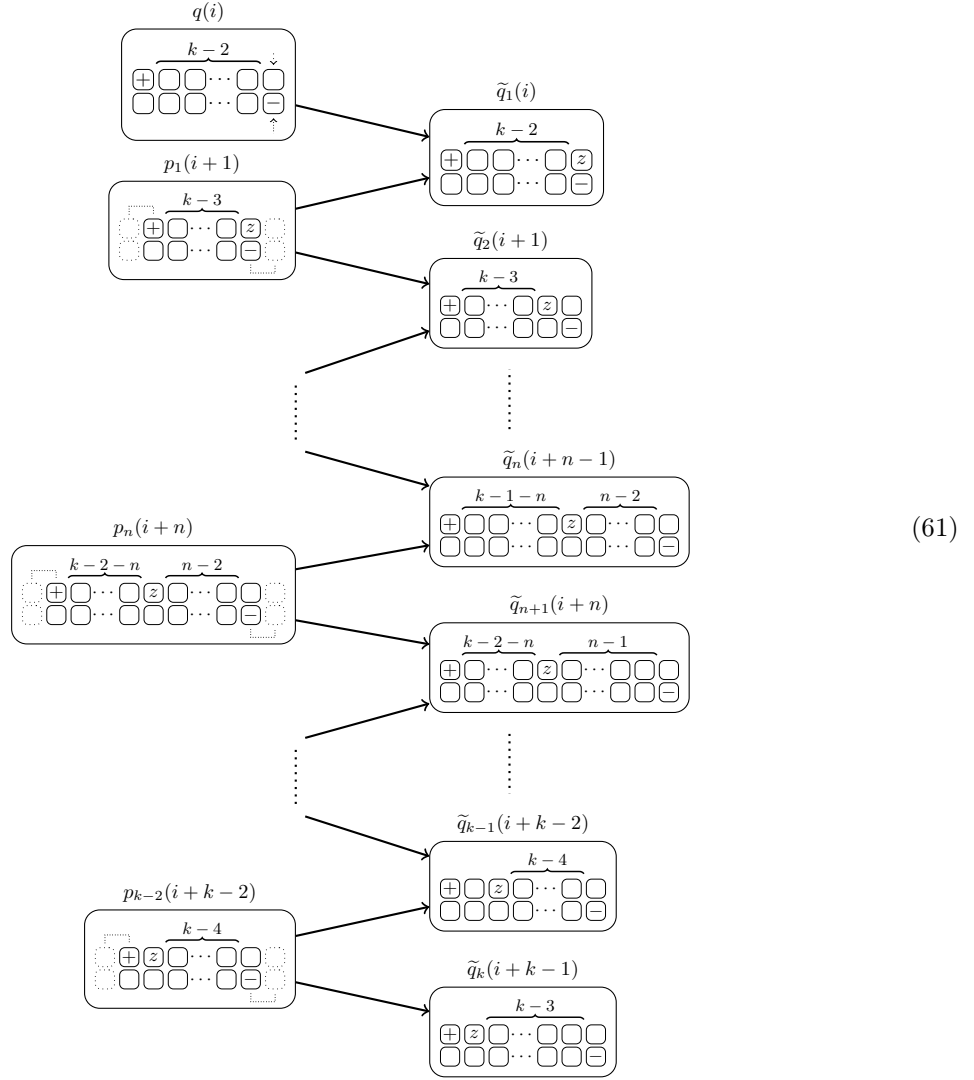
\begin{align}
    \begin{tikzpicture}[baseline=(current bounding box.center), x=1.25em, y=2.75em]
        \node[draw, rectangle, rounded corners=2mm, scale=0.75, anchor=east] (A1) at (0,0) {
            \begin{tikzpicture}[baseline=-0.5*\completevertical, rounded corners=0mm]
                \stringud{-3}{\circled{+}}{\circled{}}
                \stringud{-2}{\circled{}}{\circled{}}
                \stringud{-1}{\circled{}}{\circled{}}
                \stringud{0}{$\cdots$}{$\cdots$}
                \stringud{1}{\circled{}}{\circled{}}
                \stringud{2}{\circled{}}{\circled{-}}
                \ubraceplusstring{-2}{2}{$k-2$}{0}
                %\commutatorhopu{3.5}{4.5}
                \commutatorintdot{2}
            \end{tikzpicture}
        };
        % Second TikZ picture in a node
        \node[draw, rectangle, rounded corners=2mm, scale=0.75, anchor=west] (B1) at (4,-1) {
            \begin{tikzpicture}[baseline=-0.5*\completevertical, rounded corners=0mm]
                \stringud{-3}{\circled{+}}{\circled{}}
                \stringud{-2}{\circled{}}{\circled{}}
                \stringud{-1}{\circled{}}{\circled{}}
                \stringud{0}{$\cdots$}{$\cdots$}
                \stringud{1}{\circled{}}{\circled{}}
                \stringud{2}{\circled{z}}{\circled{-}}
                \ubraceplusstring{-2}{2}{$k-2$}{0}
                %\commutatorhopu{3.5}{4.5}
            \end{tikzpicture}
        };
        \node[draw, rectangle, rounded corners=2mm, scale=0.7, anchor=east] (A2) at (0,-2) {
            \begin{tikzpicture}[baseline=-0.5*\completevertical, rounded corners=0mm]
                %\stringud{-2}{\circled{}}{\circled{}}
                %\circleddottedcol{-3}
                \stringud{-3}{\circleddotted{}}{\circleddotted{}}
                \stringud{-2}{\circled{+}}{\circled{}}
                \circledcol{-1}
                \stringud{0}{$\cdots$}{$\cdots$}
                \stringud{1}{\circled{}}{\circled{}}
                \stringud{2}{\circled{z}}{\circled{-}}
                \stringud{3}{\circleddotted{}}{\circleddotted{}}
                \ubraceplusstring{-1}{2}{$k-3$}{0}
                \commutatorhopudot{-3}{-2}
                \commutatorhopddot{2}{3}
                %\commutatorint{2}
            \end{tikzpicture}
        };
        \node[draw, rectangle, rounded corners=2mm, scale=0.7, anchor=west] (B2) at (4,-3) {
            \begin{tikzpicture}[baseline=-0.5*\completevertical, rounded corners=0mm]
                \stringud{-2}{\circled{+}}{\circled{}}
                \stringud{-1}{\circled{}}{\circled{}}
                \stringud{0}{$\cdots$}{$\cdots$}
                \stringud{1}{\circled{}}{\circled{}}
                \stringud{2}{\circled{z}}{\circled{}}
                \stringud{3}{\circled{}}{\circled{-}}
                \ubraceplusstring{-1}{2}{$k-3$}{0}
                %\commutatorhopu{3.5}{4.5}
            \end{tikzpicture}
        };
        \node (A3a) at (0,-4) {};
        \node (A3b) at (0,-5) {};
        \node[draw, rectangle, rounded corners=2mm, scale=0.7, anchor=west] (B3) at (4,-6) {
            \begin{tikzpicture}[baseline=-0.5*\completevertical, rounded corners=0mm]
                \stringud{-1}{\circled{+}}{\circled{}}
                \stringud{0}{\circled{}}{\circled{}}
                \stringud{1}{\circled{}}{\circled{}}
                \stringud{2}{$\cdots$}{$\cdots$}
                \stringud{3}{\circled{}}{\circled{}}
                \stringud{4}{\circled{z}}{\circled{}}
                \stringud{5}{\circled{}}{\circled{}}
                \stringud{6}{$\cdots$}{$\cdots$}
                \stringud{7}{\circled{}}{\circled{}}
                \stringud{8}{\circled{}}{\circled{-}}
                \ubraceplusstring{0}{4}{$k-1-n$}{0}
                \ubraceplusstring{5}{8}{$n-2$}{0}
                %\stringlongu{1}{3}{$q_2(i-1)$}
            \end{tikzpicture}
        };
        \node[draw, rectangle, rounded corners=2mm, scale=0.7, anchor=east] (A4) at (0,-7) {
            \begin{tikzpicture}[baseline=-0.5*\completevertical, rounded corners=0mm]
                \stringud{-1}{\circleddotted{}}{\circleddotted{}}
                \stringud{0}{\circled{+}}{\circled{}}
                \stringud{1}{\circled{}}{\circled{}}
                \stringud{2}{$\cdots$}{$\cdots$}
                \stringud{3}{\circled{}}{\circled{}}
                \stringud{4}{\circled{z}}{\circled{}}
                \stringud{5}{\circled{}}{\circled{}}
                \stringud{6}{$\cdots$}{$\cdots$}
                \stringud{7}{\circled{}}{\circled{}}
                \stringud{8}{\circled{}}{\circled{-}}
                \stringud{9}{\circleddotted{}}{\circleddotted{}}
                \ubraceplusstring{1}{4}{$k-2-n$}{0}
                \ubraceplusstring{5}{8}{$n-2$}{0}
                %\stringlongu{1}{3}{$q_2(i-1)$}
                \commutatorhopudot{-1}{0}
                \commutatorhopddot{8}{9}
            \end{tikzpicture}
        };
        \node[draw, rectangle, rounded corners=2mm, scale=0.7, anchor=west] (B4) at (4,-8) {
            \begin{tikzpicture}[baseline=-0.5*\completevertical, rounded corners=0mm]
                \stringud{0}{\circled{+}}{\circled{}}
                \stringud{1}{\circled{}}{\circled{}}
                \stringud{2}{$\cdots$}{$\cdots$}
                \stringud{3}{\circled{}}{\circled{}}
                \stringud{4}{\circled{z}}{\circled{}}
                \stringud{5}{\circled{}}{\circled{}}
                \stringud{6}{$\cdots$}{$\cdots$}
                \stringud{7}{\circled{}}{\circled{}}
                \stringud{8}{\circled{}}{\circled{}}
                \stringud{9}{\circled{}}{\circled{-}}
                \ubraceplusstring{1}{4}{$k-2-n$}{0}
                \ubraceplusstring{5}{9}{$n-1$}{0}
                %\stringlongu{1}{3}{$q_2(i-1)$}
            \end{tikzpicture}
        };
        \node (A5a) at (0,-9) {};
        \node (A5b) at (0,-10) {};
        \node[draw, rectangle, rounded corners=2mm, scale=0.7, anchor=west] (B5) at (4,-11) {
            \begin{tikzpicture}[baseline=-0.5*\completevertical, rounded corners=0mm]
                \stringud{1}{\circled{+}}{\circled{}}
                \stringud{2}{\circled{}}{\circled{}}
                \stringud{3}{\circled{z}}{\circled{}}
                \stringud{4}{\circled{}}{\circled{}}
                \stringud{5}{$\cdots$}{$\cdots$}
                \stringud{6}{\circled{}}{\circled{}}
                \stringud{7}{\circled{}}{\circled{-}}
                %\commutatorhopudot{7.5}{8.5}
                %\ubraceplusstring{1}{4}{$l-2-m$}{0}
                \ubraceplusstring{4}{7}{$k-4$}{0}
                %\stringlongu{1}{3}{$q_2(i-1)$}
            \end{tikzpicture}
        };
        \node[draw, rectangle, rounded corners=2mm, scale=0.7, anchor=east] (A6) at (0,-12) {
            \begin{tikzpicture}[baseline=-0.5*\completevertical, rounded corners=0mm]
                \stringud{1}{\circleddotted{}}{\circleddotted{}}
                \stringud{2}{\circled{+}}{\circled{}}
                \stringud{3}{\circled{z}}{\circled{}}
                \stringud{4}{\circled{}}{\circled{}}
                \stringud{5}{$\cdots$}{$\cdots$}
                \stringud{6}{\circled{}}{\circled{}}
                \stringud{7}{\circled{}}{\circled{-}}
                \stringud{8}{\circleddotted{}}{\circleddotted{}}
                %\commutatorhopudot{7.5}{8.5}
                %\ubraceplusstring{1}{4}{$l-2-m$}{0}
                \ubraceplusstring{4}{7}{$k-4$}{0}
                \commutatorhopddot{7}{8}
                \commutatorhopudot{1}{2}
                %\stringlongu{1}{3}{$q_2(i-1)$}
            \end{tikzpicture}
        };
        \node[draw, rectangle, rounded corners=2mm, scale=0.7, anchor=west] (B6) at (4,-13) {
            \begin{tikzpicture}[baseline=-0.5*\completevertical, rounded corners=0mm]
                \stringud{2}{\circled{+}}{\circled{}}
                \stringud{3}{\circled{z}}{\circled{}}
                \stringud{4}{\circled{}}{\circled{}}
                \stringud{5}{$\cdots$}{$\cdots$}
                \stringud{6}{\circled{}}{\circled{}}
                \stringud{7}{\circled{}}{\circled{}}
                \stringud{8}{\circled{}}{\circled{-}}
                %\commutatorhopudot{7.5}{8.5}
                %\ubraceplusstring{1}{4}{$l-2-m$}{0}
                \ubraceplusstring{4}{8}{$k-3$}{0}
                %\stringlongu{1}{3}{$q_2(i-1)$}
            \end{tikzpicture}
        };
        % First TikZ picture in a node
        % Arrow from first node to the second
        \draw[->,thick] (A1) -- (B1);
        \draw[->,thick] (A2) -- (B1);
        \draw[->,thick] (A2) -- (B2);
        \draw[->,thick] (A4) -- (B3);
        \draw[->,thick] (A4) -- (B4);
        \draw[dotted,thick] (A3a) -- (A3b);
        \draw[dotted,thick] (A5a) -- (A5b);
        \draw[dotted,thick] ([xshift=3em, yshift=3.5em]B3.west) -- ([xshift=3em, yshift=-2.5em]B2.west);
        \draw[dotted,thick] ([xshift=3em, yshift=-2.5em]B4.west) -- ([xshift=3em, yshift=3.5em]B5.west);
        \draw[->,thick] (A3a) -- (B2);
        \draw[->,thick] (A3b) -- (B3);
        \draw[->,thick] (A5a) -- (B4);
        \draw[->,thick] (A5b) -- (B5);
        \draw[->,thick] (A6) -- (B5);
        \draw[->,thick] (A6) -- (B6);
        \node[above=0em  of B6, scale=0.75] {$\widetilde{q}_{k}(i+k-1)$};
        \node[above=0em  of A6, scale=0.75] {$p_{k-2}(i+k-2)$};
        \node[above=0em of B5, scale=0.75] {$\widetilde{q}_{k-1}(i+k-2)$};
        \node[above=0em of B4, scale=0.75] {$\widetilde{q}_{n+1}(i+n)$};
        \node[above=0em of A4, scale=0.75] {$p_n(i+n)$};
        \node[above=0em  of B3, scale=0.75] {$\widetilde{q}_{n}(i+n-1)$};
        \node[above=0em  of B2, scale=0.75] {$\widetilde{q}_2(i+1)$};
        \node[above=0em  of B1, scale=0.75] {$\widetilde{q}_1(i)$};
        \node[above=0em  of A1, scale=0.75] {$q(i)$};
        \node[above=0em  of A2, scale=0.75] {$p_1(i+1)$};
    \end{tikzpicture}
\end{align}
where the parts omitted with a dot have a similar structure, and $p_l \in \mathcal{C}_{k-1}$ and $q \in \mathcal{C}_k$, and we have $\alpha_{(\uparrow, +),(\downarrow, -)} = c_{i}(q) \propto c_{i+1}(p_1) \propto \cdots \propto c_{i+n}(p_n) \propto \cdots \propto c_{i+k-2}(p_{k-2}) = 0$ from~\eqref{eq:graphical-cancellation-one-m} and~\eqref{eq:graphical-cancellation-two-m}.
%%%%%%%%%%%%%%%%%%%%%%%%%%%%%%%%%%%%%%%%
Then we have $\alpha_{(\uparrow, +),(\downarrow, -)} = 0$.
%%%%%%%%%%%%%%%%%%%%%%%%%%%%%%%%%%%%%%%%
The other cases can also be proved in the same way.
\end{proof}

%%%%%%%%%%%%%%%%%%%%%%%%%%%%%%
We will derive a relationship between the coefficients in the following Lemma.
\begin{lemma}
    \label{lem:eighth}
    %Let $\sigma\in\bce{\uparrow, \downarrow}$ and $s\in\bce{+, -}$, then 
    $\alpha_{(\uparrow, s),(\uparrow, -s)} =  \alpha_{(\downarrow, s),(\downarrow, -s)}$ and $\alpha_{(\sigma, +),(\sigma, -)} =  
     (-1)^{k-1} \alpha_{(\sigma, -),(\sigma, +)}$.
\end{lemma}
\begin{proof}
We consider the following cancellation of the $k$-support basis element $\widetilde{q}_l(i+l-1)\ (1\leq l \leq k)$ in~\eqref{eq:k-cancellation2}:
\begin{align}
    \label{eq:lemeighth}
    \begin{tikzpicture}[baseline=(current bounding box.center), x=1.25em, y=3.em]
        % First TikZ picture in a node
        \node[draw, rectangle, rounded corners=2mm, scale=0.75, anchor=east] (A1) at (0,0) {
            \begin{tikzpicture}[baseline=-0.5*\completevertical, rounded corners=0mm]
                \stringud{-3}{\circled{+}}{\circled{}}
                \stringud{-2}{\circled{}}{\circled{}}
                \stringud{-1}{\circled{}}{\circled{}}
                \stringud{0}{$\cdots$}{$\cdots$}
                \stringud{1}{\circled{}}{\circled{}}
                \stringud{2}{\circled{-}}{\circled{}}
                \ubraceplusstring{-2}{2}{$k-2$}{0}
                %\commutatorhopu{3.5}{4.5}
                \commutatorintdot{2}
            \end{tikzpicture}
        };
         % Second TikZ picture in a node
         \node[draw, rectangle, rounded corners=2mm, scale=0.75, anchor=west] (B1) at (4,-1) {
            \begin{tikzpicture}[baseline=-0.5*\completevertical, rounded corners=0mm]
                \stringud{-3}{\circled{+}}{\circled{}}
                \stringud{-2}{\circled{}}{\circled{}}
                \stringud{-1}{\circled{}}{\circled{}}
                \stringud{0}{$\cdots$}{$\cdots$}
                \stringud{1}{\circled{}}{\circled{}}
                \stringud{2}{\circled{-}}{\circled{z}}
                \ubraceplusstring{-2}{2}{$k-2$}{0}
                %\commutatorhopu{3.5}{4.5}
            \end{tikzpicture}
        };
        \node[draw, rectangle, rounded corners=2mm, scale=0.7, anchor=east] (A2) at (0,-2) {
            \begin{tikzpicture}[baseline=-0.5*\completevertical, rounded corners=0mm]
                %\stringud{-2}{\circled{}}{\circled{}}
                %\circleddottedcol{-3}
                \stringud{-3}{\circleddotted{}}{\circleddotted{}}
                \stringud{-2}{\circled{+}}{\circled{}}
                \circledcol{-1}
                \stringud{0}{$\cdots$}{$\cdots$}
                \stringud{1}{\circled{}}{\circled{}}
                \stringud{2}{\circled{-}}{\circled{z}}
                \stringud{3}{\circleddotted{}}{\circleddotted{}}
                \ubraceplusstring{-1}{2}{$k-3$}{0}
                \commutatorhopudot{-3}{-2}
                \commutatorhopddot{2}{3}
                %\commutatorint{2}
            \end{tikzpicture}
        };
        \node[draw, rectangle, rounded corners=2mm, scale=0.7, anchor=west] (B2) at (4,-3) {
            \begin{tikzpicture}[baseline=-0.5*\completevertical, rounded corners=0mm]
                \stringud{-3}{\circled{+}}{\circled{}}
                \stringud{-2}{\circled{}}{\circled{}}
                \stringud{-1}{\circled{}}{\circled{}}
                \stringud{0}{$\cdots$}{$\cdots$}
                \stringud{1}{\circled{}}{\circled{}}
                \stringud{2}{\circled{-}}{\circled{\pm}}
                \stringud{3}{\circled{}}{\circled{\mp}}
                \ubraceplusstring{-2}{2}{$k-3$}{0}
                %\commutatorhopu{3.5}{4.5}
            \end{tikzpicture}
        };
        \node[draw, rectangle, rounded corners=2mm, scale=0.7, anchor=east] (A3) at (0,-4) {
            \begin{tikzpicture}[baseline=-0.5*\completevertical, rounded corners=0mm]
                %\stringud{-2}{\circled{}}{\circled{}}
                %\circleddottedcol{-3}
                \stringud{-3}{\circleddotted{}}{\circleddotted{}}
                \stringud{-2}{\circled{+}}{\circled{}}
                \circledcol{-1}
                \stringud{0}{$\cdots$}{$\cdots$}
                \stringud{1}{\circled{}}{\circled{}}
                \stringud{2}{\circled{-}}{\circled{\pm}}
                \stringud{3}{\circled{}}{\circled{\mp}}
                \stringud{4}{\circleddotted{}}{\circleddotted{}}
                \ubraceplusstring{-1}{2}{$k-4$}{0}
                \commutatorhopudot{-3}{-2}
                \commutatorhopddot{3}{4}
                %\commutatorint{2}
            \end{tikzpicture}
        };
        \node (B3a) at (3.75,-5) {};
        \node (B3b) at (3.75,-6) {};
        \node[draw, rectangle, rounded corners=2mm, scale=0.7, anchor=east] (A4) at (0,-7) {
            \begin{tikzpicture}[baseline=-0.5*\completevertical, rounded corners=0mm]
                \stringud{-2}{\circleddotted{}}{\circleddotted{}}
                \stringud{-1}{\circled{+}}{\circled{}}
                \stringud{0}{\circled{}}{\circled{}}
                \stringud{1}{\circled{}}{\circled{}}
                \stringud{2}{$\cdots$}{$\cdots$}
                \stringud{3}{\circled{}}{\circled{}}
                \stringud{4}{\circled{-}}{\circled{\pm}}
                \stringud{5}{\circled{}}{\circled{}}
                \stringud{6}{$\cdots$}{$\cdots$}
                \stringud{7}{\circled{}}{\circled{}}
                \stringud{8}{\circled{}}{\circled{\mp}}
                \stringud{9}{\circleddotted{}}{\circleddotted{}}
                \ubraceplusstring{0}{4}{$k-2-n$}{0}
                \ubraceplusstring{5}{8}{$n-2$}{0}
                %\stringlongu{1}{3}{$q_2(i-1)$}
                \commutatorhopudot{-2}{-1}
                \commutatorhopddot{8}{9}
            \end{tikzpicture}
        };
        \node[draw, rectangle, rounded corners=2mm, scale=0.7, anchor=west] (B4) at (4,-8) {
            \begin{tikzpicture}[baseline=-0.5*\completevertical, rounded corners=0mm]
                \stringud{-1}{\circled{+}}{\circled{}}
                \stringud{0}{\circled{}}{\circled{}}
                \stringud{1}{\circled{}}{\circled{}}
                \stringud{2}{$\cdots$}{$\cdots$}
                \stringud{3}{\circled{}}{\circled{}}
                \stringud{4}{\circled{-}}{\circled{\pm}}
                \stringud{5}{\circled{}}{\circled{}}
                \stringud{6}{$\cdots$}{$\cdots$}
                \stringud{7}{\circled{}}{\circled{}}
                \stringud{8}{\circled{}}{\circled{}}
                \stringud{9}{\circled{}}{\circled{\mp}}
                \ubraceplusstring{0}{4}{$k-2-n$}{0}
                \ubraceplusstring{5}{9}{$n-1$}{0}
                %\stringlongu{1}{3}{$q_2(i-1)$}
            \end{tikzpicture}
        };
        \node[draw, rectangle, rounded corners=2mm, scale=0.7, anchor=east] (A5) at (0,-9) {
            \begin{tikzpicture}[baseline=-0.5*\completevertical, rounded corners=0mm]
                \stringud{-1}{\circleddotted{}}{\circleddotted{}}
                \stringud{0}{\circled{+}}{\circled{}}
                \stringud{1}{\circled{}}{\circled{}}
                \stringud{2}{$\cdots$}{$\cdots$}
                \stringud{3}{\circled{}}{\circled{}}
                \stringud{4}{\circled{-}}{\circled{\pm}}
                \stringud{5}{\circled{}}{\circled{}}
                \stringud{6}{$\cdots$}{$\cdots$}
                \stringud{7}{\circled{}}{\circled{}}
                \stringud{8}{\circled{}}{\circled{}}
                \stringud{9}{\circled{}}{\circled{\mp}}
                \stringud{10}{\circleddotted{}}{\circleddotted{}}
                \ubraceplusstring{1}{4}{$k-3-n$}{0}
                \ubraceplusstring{5}{9}{$n-1$}{0}
                %\stringlongu{1}{3}{$q_2(i-1)$}
                \commutatorhopudot{-1}{0}
                \commutatorhopddot{9}{10}
            \end{tikzpicture}
        };
        \node (B5a) at (3.75,-10) {};
        \node (B5b) at (3.75,-11) {};
        \node[draw, rectangle, rounded corners=2mm, scale=0.7, anchor=east] (A6) at (0,-12) {
            \begin{tikzpicture}[baseline=-0.5*\completevertical, rounded corners=0mm]
                \commutatorhopudot{1}{2}
                \stringud{1}{\circleddotted{}}{\circleddotted{}}
                \stringud{2}{\circled{+}}{\circled{}}
                \stringud{3}{\circled{-}}{\circled{\pm}}
                \stringud{4}{\circled{}}{\circled{}}
                \stringud{5}{$\cdots$}{$\cdots$}
                \stringud{6}{\circled{}}{\circled{}}
                \stringud{7}{\circled{}}{\circled{\mp}}
                \stringud{8}{\circleddotted{}}{\circleddotted{}}
                \commutatorhopddot{7}{8}
                %\ubraceplusstring{1}{4}{$l-2-m$}{0}
                \ubraceplusstring{4}{7}{$k-4$}{0}
                %\stringlongu{1}{3}{$q_2(i-1)$}
            \end{tikzpicture}
        };
        \node[draw, rectangle, rounded corners=2mm, scale=0.7, anchor=west] (B6) at (4,-13) {
            \begin{tikzpicture}[baseline=-0.5*\completevertical, rounded corners=0mm]
                \stringud{2}{\circled{+}}{\circled{}}
                \stringud{3}{\circled{-}}{\circled{\pm}}
                \stringud{4}{\circled{}}{\circled{}}
                \stringud{5}{$\cdots$}{$\cdots$}
                \stringud{6}{\circled{}}{\circled{}}
                \stringud{7}{\circled{}}{\circled{}}
                \stringud{8}{\circled{}}{\circled{\mp}}
                %\commutatorhopudot{7.5}{8.5}
                %\ubraceplusstring{1}{4}{$l-2-m$}{0}
                \ubraceplusstring{4}{8}{$k-3$}{0}
                %\stringlongu{1}{3}{$q_2(i-1)$}
            \end{tikzpicture}
        };
        \node[draw, rectangle, rounded corners=2mm, scale=0.75, anchor=east] (A7) at (0,-14) {
            \begin{tikzpicture}[baseline=-0.5*\completevertical, rounded corners=0mm]
                \commutatorhopudot{-3}{-2}
                \stringud{-3}{\circleddotted{}}{\circleddotted{}}
                \stringud{-2}{\circled{z}}{\circled{\pm}}
                \stringud{-1}{\circled{}}{\circled{}}
                \stringud{0}{$\cdots$}{$\cdots$}
                \stringud{1}{\circled{}}{\circled{}}
                \stringud{2}{\circled{}}{\circled{\mp}}
                \stringud{3}{\circleddotted{}}{\circleddotted{}}
                \ubraceplusstring{-1}{2}{$k-3$}{0}
                %\commutatorhopu{3.5}{4.5}
                \commutatorhopddot{2}{3}
            \end{tikzpicture}
        };
        \node[draw, rectangle, rounded corners=2mm, scale=0.7, anchor=west] (B7) at (4,-15) {
            \begin{tikzpicture}[baseline=-0.5*\completevertical, rounded corners=0mm]
                \stringud{-2}{\circled{z}}{\circled{\pm}}
                \stringud{-1}{\circled{}}{\circled{}}
                \stringud{0}{$\cdots$}{$\cdots$}
                \stringud{1}{\circled{}}{\circled{}}
                \stringud{2}{\circled{}}{\circled{}}
                \stringud{3}{\circled{}}{\circled{\mp}}
                \ubraceplusstring{-1}{3}{$k-2$}{0}
            \end{tikzpicture}
        };
        \node[draw, rectangle, rounded corners=2mm, scale=0.75, anchor=east] (A8) at (0,-16) {
            \begin{tikzpicture}[baseline=-0.5*\completevertical, rounded corners=0mm]
                \stringud{-2}{\circled{}}{\circled{\pm}}
                \stringud{-1}{\circled{}}{\circled{}}
                \stringud{0}{$\cdots$}{$\cdots$}
                \stringud{1}{\circled{}}{\circled{}}
                \stringud{2}{\circled{}}{\circled{}}
                \stringud{3}{\circled{}}{\circled{\mp}}
                \ubraceplusstring{-1}{3}{$k-2$}{0}
                \commutatorintdot{-2}
            \end{tikzpicture}
        };
        % Arrow from first node to the second
        \draw[->,thick] (A1) -- (B1) node[midway, scale=0.7,above, xshift=0pt] {$U$};
        \draw[->,thick] (A2) -- (B1) node[midway, scale=0.7,above, xshift=0pt] {$+1$};
        \draw[->,thick] (A2) -- (B2) node[midway, scale=0.7,above, xshift=0pt] {$-2$};
        \draw[->,thick] (A3) -- (B2) node[midway, scale=0.7,above, xshift=0pt] {$+1$};
        \draw[->,thick] (A3) -- (B3a) node[midway, scale=0.7,above, xshift=0pt] {$\mp1$};
        \draw[->,thick] (A4) -- (B3b) node[midway, scale=0.7,above, xshift=0pt] {$+1$};
        \draw[dotted,thick] (B3a) -- (B3b);
        %\draw[dotted,thick] ([yshift=3em, xshift=4em]B4.west) -- ([yshift=-2em, xshift=4em]B2.west);
        \draw[dotted,thick] ([yshift=3.25em, xshift=-4em]A4.east) -- ([yshift=-2.5em, xshift=-4em]A3.east);
        \draw[->,thick] (A4) -- (B4) node[midway, scale=0.7,above, xshift=0pt] {$\mp1$};
        \draw[->,thick] (A5) -- (B4) node[midway, scale=0.7,above, xshift=0pt] {$+1$};
        \draw[->,thick] (A5) -- (B5a) node[midway, scale=0.7,above, xshift=0pt] {$\mp 1$};
        \draw[dotted,thick] (B5a) -- (B5b);
        \draw[dotted,thick] ([yshift=3.25em, xshift=-4em]A6.east) -- ([yshift=-2.5em, xshift=-4em]A5.east);
        %\draw[dotted,thick] ([yshift=3em, xshift=4em]B6.west) -- ([yshift=-2em, xshift=4em]B4.west);
        \draw[->,thick] (A6) -- (B5b) node[midway, scale=0.7,above, xshift=0pt] {$+1$};
        \draw[->,thick] (A6) -- (B6) node[midway, scale=0.7,above, xshift=0pt] {$\mp1$};
        \draw[->,thick] (A7) -- (B6) node[midway, scale=0.7,above, xshift=0pt] {$+2$};
        \draw[->,thick] (A7) -- (B7) node[midway, scale=0.7,above, xshift=0pt] {$\mp1$};
        \draw[->,thick] (A8) -- (B7) node[midway, scale=0.7,above, xshift=0pt] {$\mp U$};
        \node[above=0em  of A1, scale=0.75] {$q(i)$};
        \node[above=0em  of B1, scale=0.75] {$\widetilde{q}_1(i)$};
        \node[above=0em  of A2, scale=0.75] {$p_1(i+1)$};
        \node[above=0em  of B2, scale=0.75] {$\widetilde{q}_2(i+1)$};
        \node[above=0em  of A3, scale=0.75] {$p_2(i+2)$};
        \node[above=0em of A4, scale=0.75] {$p_n(i+n)$};
        \node[above=0em of B4, scale=0.75] {$\widetilde{q}_{n+1}(i+n)$};
        \node[above=0em  of A5, scale=0.75] {$p_{n+1}(i+n+1)$};
        \node[above=0em of A6, scale=0.75] {$p_{k-2}(i+k-2)$};
        \node[above=0em of B6, scale=0.75] {$\widetilde{q}_{k-1}(i+k-2)$};
        \node[above=0em of A7, scale=0.75] {$p_{k-1}(i+k-1)$};
        \node[above=0em of B7, scale=0.75] {$\widetilde{q}_{k}(i+k-1)$};
        \node[above=0em of A8, scale=0.75] {$q^\prime(i+k-1)$};
    \end{tikzpicture}
\end{align}
where the parts omitted with a dot have a similar structure.
%%%%%%%%%%%%%%%%%%%%%%%%%%%%%%%%%%%%%%%%
From these graphical notations, we have
$
    c_i(q)
    =
    \paren{-\frac{1}{U}}
    \paren{-\frac{1}{-2}}
    \paren{-\frac{1}{\mp1}}^{k-4}
    \paren{-\frac{2}{\mp1}}
    \paren{-\frac{\mp U}{\mp1}}
    c_{i+k-1}(q^\prime)
    =
    (\pm 1)^{k-3}c_{i+k-1}(q^\prime)
$. 
%%%%%%%%%%%%%%%%%%%%%%%%%%%%%%%%%%%%%%%%
From~\eqref{eq:Fkk}, we have $c_i(q) = \alpha_{(\uparrow, +), (\uparrow, -)}$ and $c_{i+k-1}(q^\prime) = \alpha_{(\downarrow, \pm), (\downarrow, \mp)}$.
%%%%%%%%%%%%%%%%%%%%%%%%%%%%%%%%%%%%%%%%
Then we have $\alpha_{(\uparrow, +), (\uparrow, -)} = (\pm 1)^{k-1} \alpha_{(\downarrow, \pm), (\downarrow, \mp)}$, and we can see 
$\alpha_{(\uparrow, +), (\uparrow, -)} = \alpha_{(\downarrow, +), (\downarrow, -)} = (-1)^{k-1}  \alpha_{(\downarrow, -), (\downarrow, +)}$.

%%%%%%%%%%%%%%%%%%%%%%%%%%%%%%%%%%%%%%%%
The remaining case we have to prove is $\alpha_{(\uparrow, +), (\uparrow, -)} = (-1)^{k-1}  \alpha_{(\uparrow, -), (\uparrow, +)}$.
%%%%%%%%%%%%%%%%%%%%%%%%%%%%%%%%%%%%%%%%
This case is also proved in the same way by considering the case where all spin flavors are reversed (the upper row and the lower row are interchanged) in~\eqref{eq:lemeighth}, and we have $\alpha_{(\downarrow, +), (\downarrow, -)} = (\pm 1)^{k-1} \alpha_{(\uparrow, \pm), (\uparrow, \mp)}$, and we can see $\alpha_{(\downarrow, +), (\downarrow, -)} = \alpha_{(\uparrow, +), (\uparrow, -)} = (-1)^{k-1}  \alpha_{(\uparrow, -), (\uparrow, +)}$.
%%%%%%%%%%%%%%%%%%%%%%%%%%%%%%%%%%%%%%%%
This concludes the proof.
\end{proof}

%%%%%%%%%%%%%%%%%%%%%%%%%%%%%%
We can prove the coefficients are zero if the columns at the right and left end are the same, such as $
c_i\paren{
    \begin{tikzpicture}[baseline=-0.5*\completevertical]
        \stringud{0}{\circled{+}}{\circled{}}
        \stringud{1}{$\cdots$}{$\cdots$}
        \stringud{2}{\circled{+}}{\circled{}}
    \end{tikzpicture}
}
=
0
$, which is stated in the following Lemma.
\begin{lemma}
    \label{lem:nineth}
    %Let $\sigma\in\bce{\uparrow, \downarrow}$ and $s\in\bce{+, -}$, then 
     $\alpha_{(\sigma, s),(\sigma, s)} =  0$.
\end{lemma}
\begin{proof}
We consider the following cancellation of the $k$-support basis element $\widetilde{q}_l(i+l-1)\ (1\leq l \leq k-1)$ in~\eqref{eq:k-cancellation2}:
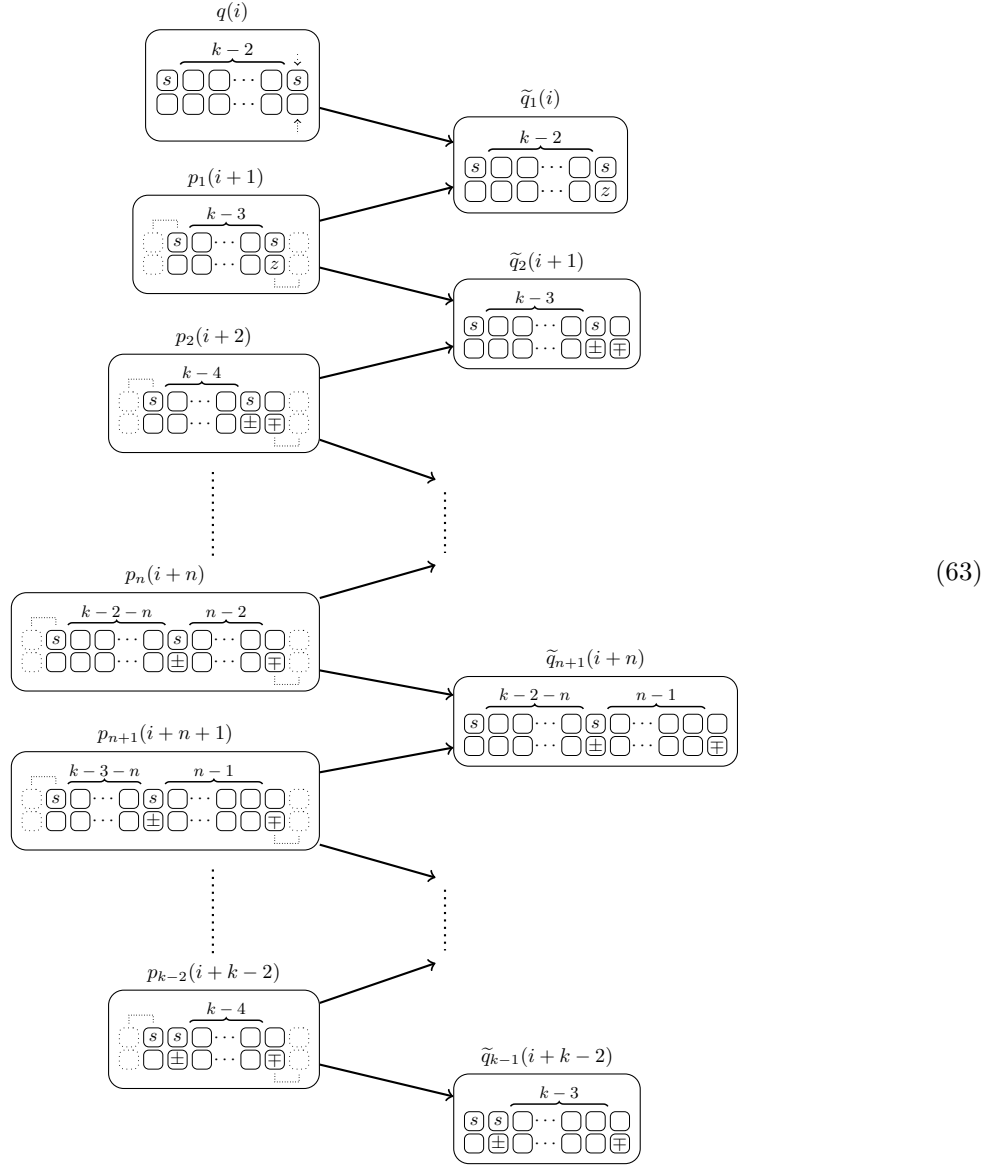
\begin{align}
    \label{eq:nineth}
    \begin{tikzpicture}[baseline=(current bounding box.center), x=1.25em, y=3.em]
        % First TikZ picture in a node
        \node[draw, rectangle, rounded corners=2mm, scale=0.75, anchor=east] (A1) at (0,0) {
            \begin{tikzpicture}[baseline=-0.5*\completevertical, rounded corners=0mm]
                \stringud{-3}{\circled{s}}{\circled{}}
                \stringud{-2}{\circled{}}{\circled{}}
                \stringud{-1}{\circled{}}{\circled{}}
                \stringud{0}{$\cdots$}{$\cdots$}
                \stringud{1}{\circled{}}{\circled{}}
                \stringud{2}{\circled{s}}{\circled{}}
                \ubraceplusstring{-2}{2}{$k-2$}{0}
                %\commutatorhopu{3.5}{4.5}
                \commutatorintdot{2}
            \end{tikzpicture}
        };
         % Second TikZ picture in a node
         \node[draw, rectangle, rounded corners=2mm, scale=0.75, anchor=west] (B1) at (4,-1) {
            \begin{tikzpicture}[baseline=-0.5*\completevertical, rounded corners=0mm]
                \stringud{-3}{\circled{s}}{\circled{}}
                \stringud{-2}{\circled{}}{\circled{}}
                \stringud{-1}{\circled{}}{\circled{}}
                \stringud{0}{$\cdots$}{$\cdots$}
                \stringud{1}{\circled{}}{\circled{}}
                \stringud{2}{\circled{s}}{\circled{z}}
                \ubraceplusstring{-2}{2}{$k-2$}{0}
                %\commutatorhopu{3.5}{4.5}
            \end{tikzpicture}
        };
        \node[draw, rectangle, rounded corners=2mm, scale=0.7, anchor=east] (A2) at (0,-2) {
            \begin{tikzpicture}[baseline=-0.5*\completevertical, rounded corners=0mm]
                %\stringud{-2}{\circled{}}{\circled{}}
                %\circleddottedcol{-3}
                \stringud{-3}{\circleddotted{}}{\circleddotted{}}
                \stringud{-2}{\circled{s}}{\circled{}}
                \circledcol{-1}
                \stringud{0}{$\cdots$}{$\cdots$}
                \stringud{1}{\circled{}}{\circled{}}
                \stringud{2}{\circled{s}}{\circled{z}}
                \stringud{3}{\circleddotted{}}{\circleddotted{}}
                \ubraceplusstring{-1}{2}{$k-3$}{0}
                \commutatorhopudot{-3}{-2}
                \commutatorhopddot{2}{3}
                %\commutatorint{2}
            \end{tikzpicture}
        };
        \node[draw, rectangle, rounded corners=2mm, scale=0.7, anchor=west] (B2) at (4,-3) {
            \begin{tikzpicture}[baseline=-0.5*\completevertical, rounded corners=0mm]
                \stringud{-3}{\circled{s}}{\circled{}}
                \stringud{-2}{\circled{}}{\circled{}}
                \stringud{-1}{\circled{}}{\circled{}}
                \stringud{0}{$\cdots$}{$\cdots$}
                \stringud{1}{\circled{}}{\circled{}}
                \stringud{2}{\circled{s}}{\circled{\pm}}
                \stringud{3}{\circled{}}{\circled{\mp}}
                \ubraceplusstring{-2}{2}{$k-3$}{0}
                %\commutatorhopu{3.5}{4.5}
            \end{tikzpicture}
        };
        \node[draw, rectangle, rounded corners=2mm, scale=0.7, anchor=east] (A3) at (0,-4) {
            \begin{tikzpicture}[baseline=-0.5*\completevertical, rounded corners=0mm]
                %\stringud{-2}{\circled{}}{\circled{}}
                %\circleddottedcol{-3}
                \stringud{-3}{\circleddotted{}}{\circleddotted{}}
                \stringud{-2}{\circled{s}}{\circled{}}
                \circledcol{-1}
                \stringud{0}{$\cdots$}{$\cdots$}
                \stringud{1}{\circled{}}{\circled{}}
                \stringud{2}{\circled{s}}{\circled{\pm}}
                \stringud{3}{\circled{}}{\circled{\mp}}
                \stringud{4}{\circleddotted{}}{\circleddotted{}}
                \ubraceplusstring{-1}{2}{$k-4$}{0}
                \commutatorhopudot{-3}{-2}
                \commutatorhopddot{3}{4}
                %\commutatorint{2}
            \end{tikzpicture}
        };
        \node (B3a) at (3.75,-5) {};
        \node (B3b) at (3.75,-6) {};
        \node[draw, rectangle, rounded corners=2mm, scale=0.7, anchor=east] (A4) at (0,-7) {
            \begin{tikzpicture}[baseline=-0.5*\completevertical, rounded corners=0mm]
                \stringud{-2}{\circleddotted{}}{\circleddotted{}}
                \stringud{-1}{\circled{s}}{\circled{}}
                \stringud{0}{\circled{}}{\circled{}}
                \stringud{1}{\circled{}}{\circled{}}
                \stringud{2}{$\cdots$}{$\cdots$}
                \stringud{3}{\circled{}}{\circled{}}
                \stringud{4}{\circled{s}}{\circled{\pm}}
                \stringud{5}{\circled{}}{\circled{}}
                \stringud{6}{$\cdots$}{$\cdots$}
                \stringud{7}{\circled{}}{\circled{}}
                \stringud{8}{\circled{}}{\circled{\mp}}
                \stringud{9}{\circleddotted{}}{\circleddotted{}}
                \ubraceplusstring{0}{4}{$k-2-n$}{0}
                \ubraceplusstring{5}{8}{$n-2$}{0}
                %\stringlongu{1}{3}{$q_2(i-1)$}
                \commutatorhopudot{-2}{-1}
                \commutatorhopddot{8}{9}
            \end{tikzpicture}
        };
        \node[draw, rectangle, rounded corners=2mm, scale=0.7, anchor=west] (B4) at (4,-8) {
            \begin{tikzpicture}[baseline=-0.5*\completevertical, rounded corners=0mm]
                \stringud{-1}{\circled{s}}{\circled{}}
                \stringud{0}{\circled{}}{\circled{}}
                \stringud{1}{\circled{}}{\circled{}}
                \stringud{2}{$\cdots$}{$\cdots$}
                \stringud{3}{\circled{}}{\circled{}}
                \stringud{4}{\circled{s}}{\circled{\pm}}
                \stringud{5}{\circled{}}{\circled{}}
                \stringud{6}{$\cdots$}{$\cdots$}
                \stringud{7}{\circled{}}{\circled{}}
                \stringud{8}{\circled{}}{\circled{}}
                \stringud{9}{\circled{}}{\circled{\mp}}
                \ubraceplusstring{0}{4}{$k-2-n$}{0}
                \ubraceplusstring{5}{9}{$n-1$}{0}
                %\stringlongu{1}{3}{$q_2(i-1)$}
            \end{tikzpicture}
        };
        \node[draw, rectangle, rounded corners=2mm, scale=0.7, anchor=east] (A5) at (0,-9) {
            \begin{tikzpicture}[baseline=-0.5*\completevertical, rounded corners=0mm]
                \stringud{-1}{\circleddotted{}}{\circleddotted{}}
                \stringud{0}{\circled{s}}{\circled{}}
                \stringud{1}{\circled{}}{\circled{}}
                \stringud{2}{$\cdots$}{$\cdots$}
                \stringud{3}{\circled{}}{\circled{}}
                \stringud{4}{\circled{s}}{\circled{\pm}}
                \stringud{5}{\circled{}}{\circled{}}
                \stringud{6}{$\cdots$}{$\cdots$}
                \stringud{7}{\circled{}}{\circled{}}
                \stringud{8}{\circled{}}{\circled{}}
                \stringud{9}{\circled{}}{\circled{\mp}}
                \stringud{10}{\circleddotted{}}{\circleddotted{}}
                \ubraceplusstring{1}{4}{$k-3-n$}{0}
                \ubraceplusstring{5}{9}{$n-1$}{0}
                %\stringlongu{1}{3}{$q_2(i-1)$}
                \commutatorhopudot{-1}{0}
                \commutatorhopddot{9}{10}
            \end{tikzpicture}
        };
        \node (B5a) at (3.75,-10) {};
        \node (B5b) at (3.75,-11) {};
        \node[draw, rectangle, rounded corners=2mm, scale=0.7, anchor=east] (A6) at (0,-12) {
            \begin{tikzpicture}[baseline=-0.5*\completevertical, rounded corners=0mm]
                \commutatorhopudot{1}{2}
                \stringud{1}{\circleddotted{}}{\circleddotted{}}
                \stringud{2}{\circled{s}}{\circled{}}
                \stringud{3}{\circled{s}}{\circled{\pm}}
                \stringud{4}{\circled{}}{\circled{}}
                \stringud{5}{$\cdots$}{$\cdots$}
                \stringud{6}{\circled{}}{\circled{}}
                \stringud{7}{\circled{}}{\circled{\mp}}
                \stringud{8}{\circleddotted{}}{\circleddotted{}}
                \commutatorhopddot{7}{8}
                %\ubraceplusstring{1}{4}{$l-2-m$}{0}
                \ubraceplusstring{4}{7}{$k-4$}{0}
                %\stringlongu{1}{3}{$q_2(i-1)$}
            \end{tikzpicture}
        };
        \node[draw, rectangle, rounded corners=2mm, scale=0.7, anchor=west] (B6) at (4,-13) {
            \begin{tikzpicture}[baseline=-0.5*\completevertical, rounded corners=0mm]
                \stringud{2}{\circled{s}}{\circled{}}
                \stringud{3}{\circled{s}}{\circled{\pm}}
                \stringud{4}{\circled{}}{\circled{}}
                \stringud{5}{$\cdots$}{$\cdots$}
                \stringud{6}{\circled{}}{\circled{}}
                \stringud{7}{\circled{}}{\circled{}}
                \stringud{8}{\circled{}}{\circled{\mp}}
                %\commutatorhopudot{7.5}{8.5}
                %\ubraceplusstring{1}{4}{$l-2-m$}{0}
                \ubraceplusstring{4}{8}{$k-3$}{0}
                %\stringlongu{1}{3}{$q_2(i-1)$}
            \end{tikzpicture}
        };
        % Arrow from first node to the second
        \draw[->,thick] (A1) -- (B1);
        \draw[->,thick] (A2) -- (B1);
        \draw[->,thick] (A2) -- (B2);
        \draw[->,thick] (A3) -- (B2);
        \draw[->,thick] (A3) -- (B3a);
        \draw[->,thick] (A4) -- (B3b);
        \draw[dotted,thick] (B3a) -- (B3b);
        %\draw[dotted,thick] ([yshift=3em, xshift=4em]B4.west) -- ([yshift=-2em, xshift=4em]B2.west);
        \draw[dotted,thick] ([yshift=3.25em, xshift=-4em]A4.east) -- ([yshift=-2.5em, xshift=-4em]A3.east);
        \draw[->,thick] (A4) -- (B4);
        \draw[->,thick] (A5) -- (B4);
        \draw[->,thick] (A5) -- (B5a);
        \draw[dotted,thick] (B5a) -- (B5b);
        \draw[dotted,thick] ([yshift=3.25em, xshift=-4em]A6.east) -- ([yshift=-2.5em, xshift=-4em]A5.east);
        %\draw[dotted,thick] ([yshift=3em, xshift=4em]B6.west) -- ([yshift=-2em, xshift=4em]B4.west);
        \draw[->,thick] (A6) -- (B5b);
        \draw[->,thick] (A6) -- (B6) ;
        \node[above=0em  of A1, scale=0.75] {$q(i)$};
        \node[above=0em  of B1, scale=0.75] {$\widetilde{q}_1(i)$};
        \node[above=0em  of A2, scale=0.75] {$p_1(i+1)$};
        \node[above=0em  of B2, scale=0.75] {$\widetilde{q}_2(i+1)$};
        \node[above=0em  of A3, scale=0.75] {$p_2(i+2)$};
        \node[above=0em of A4, scale=0.75] {$p_n(i+n)$};
        \node[above=0em of B4, scale=0.75] {$\widetilde{q}_{n+1}(i+n)$};
        \node[above=0em  of A5, scale=0.75] {$p_{n+1}(i+n+1)$};
        \node[above=0em of A6, scale=0.75] {$p_{k-2}(i+k-2)$};
        \node[above=0em of B6, scale=0.75] {$\widetilde{q}_{k-1}(i+k-2)$};
    \end{tikzpicture}
\end{align}
where the parts omitted with a dot have a similar structure.
%%%%%%%%%%%%%%%%%%%%%%%%%%%%%%%%%%%%%%%%
From these graphical notations, we have
$
    c_i(q) \propto c_{i+1}(p_1) \propto \cdots \propto c_{n}(p_n) \propto \cdots \propto c_{k-2}(p_{k-2}) =0
$. 
%%%%%%%%%%%%%%%%%%%%%%%%%%%%%%%%%%%%%%%%
From~\eqref{eq:Fkk}, $c_i(q) = \alpha_{(\uparrow, s), (\uparrow, s)}$ and we have $\alpha_{(\uparrow, s), (\uparrow, s)} = 0$.

%%%%%%%%%%%%%%%%%%%%%%%%%%%%%%%%%%%%%%%%
In the same way, we can prove $\alpha_{(\downarrow, s), (\downarrow, s)} = 0$.
%%%%%%%%%%%%%%%%%%%%%%%%%%%%%%%%%%%%%%%%
This concludes the proof.
\end{proof}

%%%%%%%%%%%%%%%%%%%%%%%%%%%%%%%%%%%%%%%%
From Lemma~\ref{lem:seventh}--Lemma~\ref{lem:nineth}, denoting $\alpha_{(\uparrow, +),(\uparrow, -)}=2c_{k}$, we have
\begin{align}
    F_{k}^{k}
    &=
    \sum_{\sigma=\uparrow, \downarrow}
    \sum_{i=1}^{L}
    2c_{k}
    \paren{
        c_{i, \sigma}c_{i+k-1, \sigma}^{\dag}
        +
        (-1)^{k-1}
        c_{i, \sigma}^{\dag}c_{i+k-1, \sigma}
    }
    \nonumber\\
    &=
    c_{k}
    Q_k^{0}
    .
\end{align}
%%%%%%%%%%%%%%%%%%%%%%%%%%%%%%%%%%%%%%%%
We can see the $k$-support operator in $F_k$ is proportional to $Q_k^0$.

%%%%%%%%%%%%%%%%%%%%%%%%%%%%%%%%%%%%%%%%
We define $\Delta \equiv F_k - c_k Q_k$, which is also a conserved quantity: $\bck{\Delta, H} = \bck{F_k, H} - c_k \bck{Q_k, H} = 0$.
%%%%%%%%%%%%%%%%%%%%%%%%%%%%%%%%%%%%%%%%
We can prove $\Delta$ is a less-than-or-equal-to-$(k-1)$-local conserved quantity as follows:
\begin{align}
    \Delta 
    &\equiv 
    F_k - c_k Q_k
    \nonumber\\
    &=
    c_{k}
    Q_k^{0} + F_k^{k-1} + (\text{rest})  - c_k \paren{Q_k^0 + \delta Q_{k-1}(U)}
    \nonumber\\
    &=
    F_k^{k-1} + (\text{rest}) -c_k \delta Q_{k-1}(U)
    ,
\end{align}
where the last line is a linear combination of less-than-or-equal-to-$(k-1)$-support operators, and then we can see $\Delta$ is a less-than-or-equal-to-$(k-1)$-local conserved quantity, and $(\text{rest})$ is the same as that in~\eqref{eq:F_k_decomp}.
%%%%%%%%%%%%%%%%%%%%%%%%%%%%%%%%%%%%%%%%
Then we have proved $F_k = c_k Q_k + \Delta$ where $\Delta$ is a less-than-or-equal-to-$(k-1)$-local conserved quantity, which concludes the proof of Theorem~\ref{theorem}.

%%%%%%%%%%%%%%%%%%%%%%%%%%%%%%%%%%%%%%%%
It should be noted that the aforementioned proof does not hold in the non-interacting case where $U = 0$. 
%%%%%%%%%%%%%%%%%%%%%%%%%%%%%%%%%%%%%%%%
In this case, there is no contribution from $F_k^k$ to the cancellation of $k$-support operators in~\eqref{eq:k-cancellation2}, and~\eqref{eq:Fkk} itself becomes a local charge.

\subsection{Case of $k\geq \floor{\frac{L}{2}}$}
%%%%%%%%%%%%%%%%%%%%%%%%%%%%%%%%%%%%%%%%
We briefly explain why the above proof breaks in the case of $k\geq \floor{\frac{L}{2}}$.
%%%%%%%%%%%%%%%%%%%%%%%%%%%%%%
For example, in the case of $k = L/2$ (where we assume even $L$), there is another contribution to the cancellation of the basis element with $(L/2+1)$-support configuration in~\eqref{eq:cancellation_eg_pmmp} from the $(L/2)$-support basis element:
%%%%%%%%%%%%%%%%%%%%%%%%%%%%%%
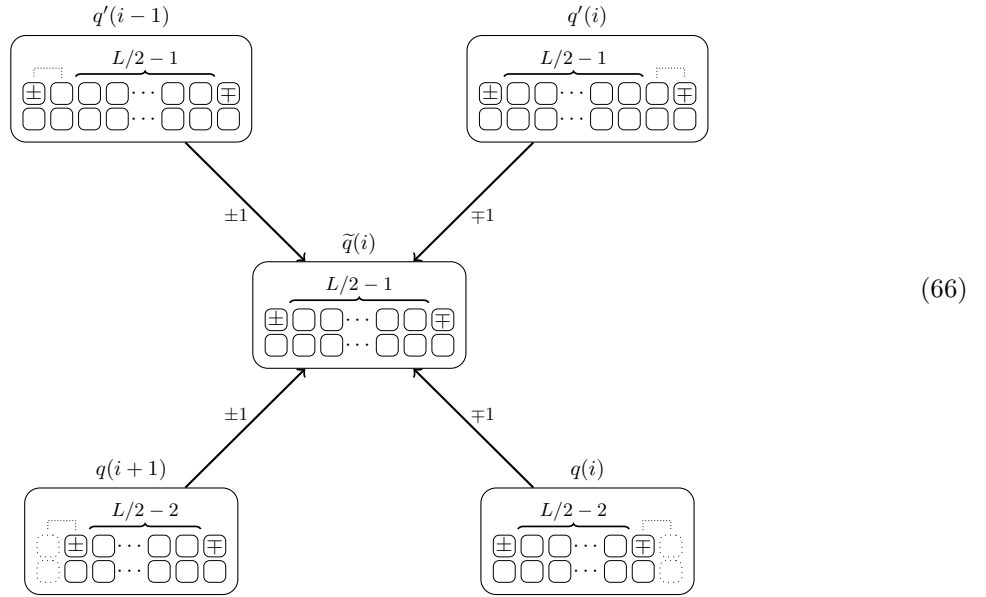
\begin{align}
    \label{eq:cancellation_over_L2}
    \begin{tikzpicture}[baseline=(current bounding box.center)]
        % Second TikZ picture in a node
        \node[draw, rectangle, rounded corners=2mm, scale=0.8] (B) at (3,3) {
            \begin{tikzpicture}[baseline=-0.5*\completevertical, rounded corners=0mm]
                \stringud{-2}{\circled{\pm}}{\circled{}}
                \stringud{-1}{\circled{}}{\circled{}}
                \stringud{0}{\circled{}}{\circled{}}
                \stringud{1}{$\cdots$}{$\cdots$}
                \stringud{2}{\circled{}}{\circled{}}
                \stringud{3}{\circled{}}{\circled{}}
                \stringud{4}{\circled{\mp}}{\circled{}}
                %\stringud{3}{\circleddotted{}}{\circleddotted{}}
                \ubraceplusstring{-1}{4}{$L/2-1$}{0}
                %\commutatorhopu{2}{3}
            \end{tikzpicture}
        };
        \node[draw, rectangle, rounded corners=2mm, scale=0.8] (A1) at (0,0) {
            \begin{tikzpicture}[baseline=-0.5*\completevertical, rounded corners=0mm]
                \stringud{-3}{\circleddotted{}}{\circleddotted{}}
                \stringud{-2}{\circled{\pm}}{\circled{}}
                \stringud{-1}{\circled{}}{\circled{}}
                \stringud{0}{$\cdots$}{$\cdots$}
                \stringud{1}{\circled{}}{\circled{}}
                \stringud{2}{\circled{}}{\circled{}}
                \stringud{3}{\circled{\mp}}{\circled{}}
                \ubraceplusstring{-1}{3}{$L/2-2$}{0}
                \commutatorhopudot{-2}{-3}
                %\commutatorhopu{2}{3}
            \end{tikzpicture}
        };
        % First TikZ picture in a node
        \node[draw, rectangle, rounded corners=2mm, scale=0.8] (A2) at (6,0) {
            \begin{tikzpicture}[baseline=-0.5*\completevertical, rounded corners=0mm]
                \stringud{-3}{\circled{\pm}}{\circled{}}
                \stringud{-2}{\circled{}}{\circled{}}
                \stringud{-1}{\circled{}}{\circled{}}
                \stringud{0}{$\cdots$}{$\cdots$}
                \stringud{1}{\circled{}}{\circled{}}
                \stringud{2}{\circled{\mp}}{\circled{}}
                \stringud{3}{\circleddotted{}}{\circleddotted{}}
                \ubraceplusstring{-2}{2}{$L/2-2$}{0}
                \commutatorhopudot{2}{3}
            \end{tikzpicture}
        };
        \node[draw, rectangle, rounded corners=2mm, scale=0.8] (A3) at (0,6) {
            \begin{tikzpicture}[baseline=-0.5*\completevertical, rounded corners=0mm]
                %\stringud{-3}{\circleddotted{}}{\circleddotted{}}
                \stringud{-4}{\circled{\pm}}{\circled{}}
                \stringud{-3}{\circled{}}{\circled{}}
                \stringud{-2}{\circled{}}{\circled{}}
                \stringud{-1}{\circled{}}{\circled{}}
                \stringud{0}{$\cdots$}{$\cdots$}
                \stringud{1}{\circled{}}{\circled{}}
                \stringud{2}{\circled{}}{\circled{}}
                \stringud{3}{\circled{\mp}}{\circled{}}
                \ubraceplusstring{-2}{3}{$L/2-1$}{0}
                \commutatorhopudot{-3}{-4}
                %\commutatorhopu{2}{3}
            \end{tikzpicture}
        };
        \node[draw, rectangle, rounded corners=2mm, scale=0.8] (A4) at (6,6) {
            \begin{tikzpicture}[baseline=-0.5*\completevertical, rounded corners=0mm]
                %\stringud{-3}{\circleddotted{}}{\circleddotted{}}
                \stringud{-3}{\circled{\pm}}{\circled{}}
                %\stringud{-3}{\circled{}}{\circled{}}
                \stringud{-2}{\circled{}}{\circled{}}
                \stringud{-1}{\circled{}}{\circled{}}
                \stringud{0}{$\cdots$}{$\cdots$}
                \stringud{1}{\circled{}}{\circled{}}
                \stringud{2}{\circled{}}{\circled{}}
                \stringud{3}{\circled{}}{\circled{}}
                \stringud{4}{\circled{\mp}}{\circled{}}
                \ubraceplusstring{-2}{3}{$L/2-1$}{0}
                \commutatorhopudot{3}{4}
                %\commutatorhopu{2}{3}
            \end{tikzpicture}
        };
        % Arrow from first node to the second
        \draw[->,thick] (A1) -- (B) node[midway, scale=0.7,above, xshift=-5pt] {$\pm1$};
        \draw[->,thick] (A3) -- (B) node[midway, scale=0.7,below, xshift=-5pt] {$\pm1$};
        \draw[->,thick] (A2) -- (B) node[midway, scale=0.7,above, xshift=5pt] {$\mp1$};;
        \draw[->,thick] (A4) -- (B) node[midway, scale=0.7,below, xshift=5pt] {$\mp1$};;
        \node[above=0em  of B, scale=0.8] {$\widetilde{q}(i)$};
        \node[above=0em  of A1, scale=0.8] {$q(i+1)$};
        \node[above=0em  of A3, scale=0.8] {$q^\prime(i-1)$};
        \node[above=0em  of A2, scale=0.8] {$q(i)$};
        \node[above=0em  of A4, scale=0.8] {$q^\prime(i)$};
    \end{tikzpicture}
\end{align}
where the configuration $q^\prime$ has $(L/2)$-support because 
\begin{align}
    \begin{tikzpicture}[baseline=-0.5*\completevertical, rounded corners=0mm]
        \stringud{-2}{\circled{\pm}}{\circled{}}
        \stringud{-1}{\circled{}}{\circled{}}
        \stringud{0}{\circled{}}{\circled{}}
        \stringud{1}{$\cdots$}{$\cdots$}
        \stringud{2}{\circled{}}{\circled{}}
        \stringud{3}{\circled{}}{\circled{}}
        \stringud{4}{\circled{\mp}}{\circled{}}
        %\stringud{3}{\circleddotted{}}{\circleddotted{}}
        \ubraceplusstring{-1}{4}{$L/2$}{0}
        %\commutatorhopu{2}{3}
    \end{tikzpicture}
    (i)
    =
    \begin{tikzpicture}[baseline=-0.5*\completevertical, rounded corners=0mm]
        \stringud{-2}{\circled{\mp}}{\circled{}}
        \stringud{-1}{\circled{}}{\circled{}}
        \stringud{0}{\circled{}}{\circled{}}
        \stringud{1}{$\cdots$}{$\cdots$}
        \stringud{2}{\circled{}}{\circled{}}
        \stringud{3}{\circled{}}{\circled{}}
        \stringud{4}{\circled{\pm}}{\circled{}}
        %\stringud{3}{\circleddotted{}}{\circleddotted{}}
        \ubraceplusstring{-1}{4}{$L/2-2$}{0}
        %\commutatorhopu{2}{3}
    \end{tikzpicture}
    (i+L/2+1)
    ,
\end{align}
where we used the periodic boundary condition, and we can see this is also a $(L/2)$-support basis element and included in the linear combination in $F_k^k$.

%%%%%%%%%%%%%%%%%%%%%%%%%%%%%%%%%%%%%%%%
In the same way, in the case of $k\geq \floor{\frac{L}{2}}$, there are other contributions to the cancellation apart from what we have considered in the proof of Theorem~\ref{theorem}.

\section{Summary}
\label{sec:summary}
%%%%%%%%%%%%%%%%%%%%%%%%%%%%%%%%%%%%%%%%
In this article, we proved the $k$-local conserved quantities in the one-dimensional Hubbard model are uniquely determined up to the choice of the linear combination of the lower-order local charges.
%%%%%%%%%%%%%%%%%%%%%%%%%%%%%%%%%%%%%%%%
This proof is based on Theorem~\ref{theorem}, which asserts that a maximal support component of the local conserved quantity is uniquely determined.
%%%%%%%%%%%%%%%%%%%%%%%%%%%%%%%%%%%%%%%%
This strategy has also been employed in the proof of completeness of local charges in the case of the spin-1/2 XYZ chain without a magnetic field~\cite{Nozawa2020} building upon the insights regarding the maximal support components of its local charges~\cite{Shiraishi2019}.
%%%%%%%%%%%%%%%%%%%%%%%%%%%%%%%%%%%%%%%%
The same strategy may also be immediately applicable to the other quantum integrable lattice models, such as the $\mathrm{SU}(N)$ generalization of the isotropic Heisenberg chain~\cite{GRABOWSKI1995299, yamada2023matrix}, and the Temperley-Lieb models~\cite{Nienhuis2021}.

%%%%%%%%%%%%%%%%%%%%%%%%%%%%%%%%%%%%%%%%
We note that as for the non-interacting quantum integrable systems, such as the spin-$1/2$ $XY$ chain or free fermion models, $k$-local charges are not uniquely determined and there are more than two independent families of local charges~\cite{barouch1984XY,barouch1985master, Araki1990, GRABOWSKI1995299, Fagotti_2014}.

%%%%%%%%%%%%%%%%%%%%%%%%%%%%%%%%%%%%%%%%
This is the first time that the completeness of the family of local charges in integrable electron systems has been rigorously proven.

\section*{Acknowledgments}
%%%%%%%%%%%%%%%%%%%%%%%%%%%%%%%%%%%%%%%%
This work is dedicated to my grandfather, Daihachi Yoshizawa.
%%%%%%%%%%%%%%%%%%%%%%%%%%%%%%%%%%%%%%%%
This work was supported by FoPM, WINGS Program, JSR Fellowship, the University of Tokyo, and KAKENHI Grants No. JP21J20321 from the Japan Society for the Promotion of Science (JSPS).

\vspace*{1em}
\textbf{Conflicts of interest.}
\quad The author declares no conflicts of interest.

\textbf{Data availability statement.}
\quad No new data were created or analyzed in this study.

\begin{appendix}

\section{One-local conserved quantities}
\label{app:one-sup}
%%%%%%%%%%%%%%%%%%%%%%%%%%%%%%%%%%%%%%%%
In this appendix, we prove there are no one-local conserved quantities other than the $\mathrm{SU}(2)$ charges and $\mathrm{U}(1)$ charge and $\eta$-pairing charges, which is the another $\mathrm{SU}(2)$ symmetry for even $L$~\cite{PhysRevLett.63.2144, doi:10.1142/S0217984990000933}.

%%%%%%%%%%%%%%%%%%%%%%%%%%%%%%%%%%%%%%%%
A one-local conserved quantity $F_1$ is written as
\begin{align}
    F_1
    =
    \sum_{
        \substack{
            \Bar{a}, \Bar{b}\in \bce{ \circled{}, \circled{+}, \circled{-}, \circled{z} }
            \\
            \bce{\Bar{a}, \Bar{b}} \neq \bce{\circled{}, \circled{}}
        }
    }
    \sum_{i=1}^{L}
    c_i\paren{
        \begin{tikzpicture}[baseline=-0.5*\completevertical]
            \stringud{0}{$\Bar{a}$}{$\Bar{b}$}
        \end{tikzpicture}
    }
    \begin{tikzpicture}[baseline=-0.5*\completevertical]
        \stringud{0}{$\Bar{a}$}{$\Bar{b}$}
    \end{tikzpicture}
    (i)
    ,
\end{align}
where $
c_i\paren{
        \begin{tikzpicture}[baseline=-0.5*\completevertical]
            \stringud{0}{$\Bar{a}$}{$\Bar{b}$}
        \end{tikzpicture}
    }
$ is the coefficient of $
\begin{tikzpicture}[baseline=-0.5*\completevertical]
    \stringud{0}{$\Bar{a}$}{$\Bar{b}$}
\end{tikzpicture}
(i)
$ and $\bce{\Bar{a}, \Bar{b}} \neq \bce{\circled{}, \circled{}}$ is the normalization where $F_1$ is traceless.

%%%%%%%%%%%%%%%%%%%%%%%%%%%%%%%%%%%%%%%%
We determine $F_1$ so that $\bck{F_1, H} = \bck{F_1, H_0} + \bck{F_1, H_1} = 0$.
%%%%%%%%%%%%%%%%%%%%%%%%%%%%%%%%%%%%%%%%
Because $\bck{F_1, H_0} $ is a two-support operator and $\bck{F_1, H_1}$ is a one-support operator, the two terms must independently vanish, then we have
\begin{align}
    \label{eq:one1}
    & \bck{F_1, H_0} =0,
    \\
    & \bck{F_1, H_1} = 0.
    \label{eq:one2}
\end{align}
%%%%%%%%%%%%%%%%%%%%%%%%%%%%%%%%%%%%%%%%
From~\eqref{eq:one2}, we can see:
\begin{align}
    \bck{F_1, H_1}
    =
    \sum_{
        \substack{
            \Bar{a}, \Bar{b}\in \bce{ \circled{}, \circled{+}, \circled{-}, \circled{z} }
            \\
            \bce{\Bar{a}, \Bar{b}} \neq \bce{\circled{}, \circled{}}
        }
    }
    \sum_{i=1}^{L}
    c_i\paren{
        \begin{tikzpicture}[baseline=-0.5*\completevertical]
            \stringud{0}{$\Bar{a}$}{$\Bar{b}$}
        \end{tikzpicture}
    }
    \begin{tikzpicture}[baseline=-0.5*\completevertical]
        \stringud{0}{$\Bar{a}$}{$\Bar{b}$}
        \commutatorint{0}
    \end{tikzpicture}
    (i)
    =0
    .
\end{align}
$
\begin{tikzpicture}[baseline=-0.5*\completevertical]
    \stringud{0}{$\Bar{a}$}{$\Bar{b}$}
    \commutatorint{0}
\end{tikzpicture}
(i)
$ is the one-support basis element at the $i$ th site.
%%%%%%%%%%%%%%%%%%%%%%%%%%%%%%%%%%%%%%%%
$
\begin{tikzpicture}[baseline=-0.5*\completevertical]
    \stringud{0}{$\Bar{a}_1$}{$\Bar{b}_1$}
    \commutatorint{0}
\end{tikzpicture}
(i)
$
and 
$
\begin{tikzpicture}[baseline=-0.5*\completevertical]
    \stringud{0}{$\Bar{a}_2$}{$\Bar{b}_2$}
    \commutatorint{0}
\end{tikzpicture}
(j)
$
are independent if $i\neq j$ or $\Bar{a}_1 \neq \Bar{a}_2$ or $\Bar{b}_1 \neq  \Bar{b}_2$.
%%%%%%%%%%%%%%%%%%%%%%%%%%%%%%%%%%%%%%%%]
Thus, we have
\begin{align}
    c_i\paren{
        \begin{tikzpicture}[baseline=-0.5*\completevertical]
            \stringud{0}{$\Bar{a}$}{$\Bar{b}$}
        \end{tikzpicture}
    }
    \begin{tikzpicture}[baseline=-0.5*\completevertical]
        \stringud{0}{$\Bar{a}$}{$\Bar{b}$}
        \commutatorint{0}
    \end{tikzpicture}
    (i)
    =0
    .
\end{align}
%%%%%%%%%%%%%%%%%%%%%%%%%%%%%%%%%%%%%%%%
If 
$
\begin{tikzpicture}[baseline=-0.5*\completevertical]
    \stringud{0}{$\Bar{a}$}{$\Bar{b}$}
\end{tikzpicture}
\in
\bce{
    \begin{tikzpicture}[baseline=-0.5*\completevertical]
        \stringud{0}{\circled{\pm}}{\circled{}}
    \end{tikzpicture}
    ,
    \begin{tikzpicture}[baseline=-0.5*\completevertical]
        \stringud{0}{\circled{\pm}}{\circled{z}}
    \end{tikzpicture}
    ,
    \begin{tikzpicture}[baseline=-0.5*\completevertical]
        \stringud{0}{\circled{}}{\circled{\pm}}
    \end{tikzpicture}
    ,
    \begin{tikzpicture}[baseline=-0.5*\completevertical]
        \stringud{0}{\circled{z}}{\circled{\pm}}
    \end{tikzpicture}
}
$, we can see $
\begin{tikzpicture}[baseline=-0.5*\completevertical]
    \stringud{0}{$\Bar{a}$}{$\Bar{b}$}
    \commutatorint{0}
\end{tikzpicture}
\neq
0
$ and then we have $
c_i\paren{
        \begin{tikzpicture}[baseline=-0.5*\completevertical]
            \stringud{0}{$\Bar{a}$}{$\Bar{b}$}
        \end{tikzpicture}
    }
    =0
$.
%%%%%%%%%%%%%%%%%%%%%%%%%%%%%%%%%%%%%%%%
Therefore $\bck{F_1, H_1} = 0$ is satisfied if
\begin{align}
    F_1
    =
    \sum_{i=1}^{L}
    \bck{
        \alpha_i^{+-}
        \begin{tikzpicture}[baseline=-0.5*\completevertical]
            \stringud{0}{$\circled{+}$}{$\circled{-}$}
        \end{tikzpicture}
        (i)
        +
        \alpha_i^{-+}
        \begin{tikzpicture}[baseline=-0.5*\completevertical]
            \stringud{0}{$\circled{-}$}{$\circled{+}$}
        \end{tikzpicture}
        (i)
        +
        \alpha_i^{++}
        \begin{tikzpicture}[baseline=-0.5*\completevertical]
            \stringud{0}{$\circled{+}$}{$\circled{+}$}
        \end{tikzpicture}
        (i)
        +
        \alpha_i^{--}
        \begin{tikzpicture}[baseline=-0.5*\completevertical]
            \stringud{0}{$\circled{-}$}{$\circled{-}$}
        \end{tikzpicture}
        (i)
        +
        \alpha_i^{z \uparrow}
        \begin{tikzpicture}[baseline=-0.5*\completevertical]
            \stringud{0}{$\circled{z}$}{$\circled{}$}
        \end{tikzpicture}
        (i)
        +
        \alpha_i^{z \downarrow}
        \begin{tikzpicture}[baseline=-0.5*\completevertical]
            \stringud{0}{$\circled{}$}{$\circled{z}$}
        \end{tikzpicture}
        (i)
        +
        \alpha_i^{z z}
        \begin{tikzpicture}[baseline=-0.5*\completevertical]
            \stringud{0}{$\circled{z}$}{$\circled{z}$}
        \end{tikzpicture}
        (i)
    }
    .
\end{align}
%%%%%%%%%%%%%%%%%%%%%%%%%%%%%%%%%%%%%%%%
In the following, by considering the cancellation of two-support operators ($\bck{F_1, H_0}=0$), we prove $\alpha_i^{\pm\mp}$, $\alpha_i^{z \uparrow}$ and $\alpha_i^{z \downarrow}$ are constants independent of $i$ and $\alpha_i^{\pm\mp}$ is proportional to ${(-1)}^{i}$ for even $L$ and is zero for odd $L$ and $\alpha_i^{z z}=0$.

%%%%%%%%%%%%%%%%%%%%%%%%%%%%%%%%%%%%%%%%
Considering the following  cancellation of the two-support basis element:
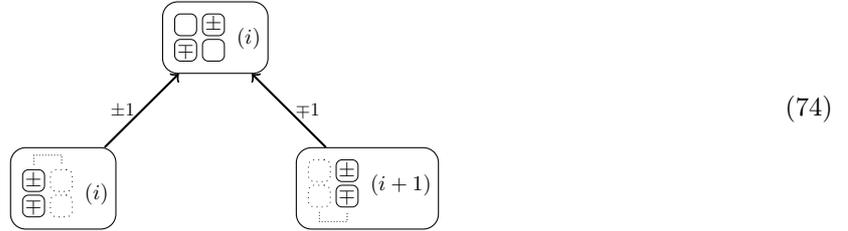
\begin{align}
    \begin{tikzpicture}[baseline=(current bounding box.center)]
        % Second TikZ picture in a node
        \node[draw, rectangle, rounded corners=2mm, scale=0.8] (B) at (2,2) {
            \begin{tikzpicture}[baseline=-0.5*\completevertical, rounded corners=0mm]
                \stringud{0}{\circled{}}{\circled{\mp}}
                \stringud{1}{\circled{\pm}}{\circled{}}
            \end{tikzpicture}
            $(i)$
        };
        \node[draw, rectangle, rounded corners=2mm, scale=0.8] (A1) at (0,0) {
            \begin{tikzpicture}[baseline=-0.5*\completevertical, rounded corners=0mm]
                \stringud{0}{\circled{\pm}}{\circled{\mp}}
                \stringud{1}{\circleddotted{}}{\circleddotted{}}
                \commutatorhopudot{0}{1}
            \end{tikzpicture}
            $(i)$
        };
        % First TikZ picture in a node
        \node[draw, rectangle, rounded corners=2mm, scale=0.8] (A2) at (4,0) {
            \begin{tikzpicture}[baseline=-0.5*\completevertical, rounded corners=0mm]
                \stringud{-1}{\circleddotted{}}{\circleddotted{}}
                \stringud{0}{\circled{\pm}}{\circled{\mp}}
                \commutatorhopddot{-1}{0}
            \end{tikzpicture}
            $(i+1)$
        };
        % Arrow from first node to the second
        \draw[->,thick] (A1) -- (B) node[midway, scale=0.7,left, xshift=0pt] {$\pm1$};
        \draw[->,thick] (A2) -- (B) node[midway, scale=0.7,right, xshift=0pt] {$\mp1$};
    \end{tikzpicture}
\end{align}
and we have $\alpha_i^{\pm\mp}=\alpha_{i+1}^{\pm\mp}$.
%%%%%%%%%%%%%%%%%%%%%%%%%%%%%%%%%%%%%%%%
Then we can see  $\alpha_i^{\pm\mp}$ is independent of $i$.

%%%%%%%%%%%%%%%%%%%%%%%%%%%%%%%%%%%%%%%%
Considering the following  cancellation of the two-support basis element:
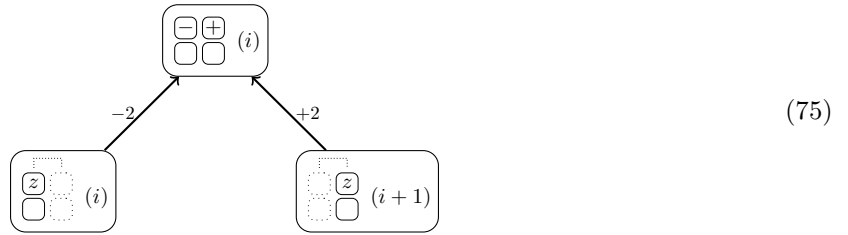
\begin{align}
    \begin{tikzpicture}[baseline=(current bounding box.center)]
        % Second TikZ picture in a node
        \node[draw, rectangle, rounded corners=2mm, scale=0.8] (B) at (2,2) {
            \begin{tikzpicture}[baseline=-0.5*\completevertical, rounded corners=0mm]
                \stringud{0}{\circled{-}}{\circled{}}
                \stringud{1}{\circled{+}}{\circled{}}
            \end{tikzpicture}
            $(i)$
        };
        \node[draw, rectangle, rounded corners=2mm, scale=0.8] (A1) at (0,0) {
            \begin{tikzpicture}[baseline=-0.5*\completevertical, rounded corners=0mm]
                \stringud{0}{\circled{z}}{\circled{}}
                \stringud{1}{\circleddotted{}}{\circleddotted{}}
                \commutatorhopudot{0}{1}
            \end{tikzpicture}
            $(i)$
        };
        % First TikZ picture in a node
        \node[draw, rectangle, rounded corners=2mm, scale=0.8] (A2) at (4,0) {
            \begin{tikzpicture}[baseline=-0.5*\completevertical, rounded corners=0mm]
                \stringud{-1}{\circleddotted{}}{\circleddotted{}}
                \stringud{0}{\circled{z}}{\circled{}}
                \commutatorhopudot{-1}{0}
            \end{tikzpicture}
            $(i+1)$
        };
        % Arrow from first node to the second
        \draw[->,thick] (A1) -- (B) node[midway, scale=0.7,left, xshift=0pt] {$-2$};
        \draw[->,thick] (A2) -- (B) node[midway, scale=0.7,right, xshift=0pt] {$+2$};
    \end{tikzpicture}
\end{align}
and we have $\alpha_i^{z\uparrow}=\alpha_{i+1}^{z\uparrow}$.
%%%%%%%%%%%%%%%%%%%%%%%%%%%%%%%%%%%%%%%%
Then we can see  $\alpha_i^{z\uparrow}$ is independent of $i$.
%%%%%%%%%%%%%%%%%%%%%%%%%%%%%%%%%%%%%%%%
In the same way, we can prove $\alpha_i^{z\downarrow}$ is independent of $i$.

%%%%%%%%%%%%%%%%%%%%%%%%%%%%%%%%%%%%%%%%
Considering the following cancellation of the two support basis elements:
\begin{align}
    \begin{tikzpicture}[baseline=(current bounding box.center)]
        % Second TikZ picture in a node
        \node[draw, rectangle, rounded corners=2mm, scale=0.8] (B) at (0,2) {
            \begin{tikzpicture}[baseline=-0.5*\completevertical, rounded corners=0mm]
                \stringud{0}{\circled{-}}{\circled{z}}
                \stringud{1}{\circled{+}}{\circled{}}
            \end{tikzpicture}
            $(i)$
        };
        \node[draw, rectangle, rounded corners=2mm, scale=0.8] (A1) at (0,0) {
            \begin{tikzpicture}[baseline=-0.5*\completevertical, rounded corners=0mm]
                \stringud{0}{\circled{z}}{\circled{z}}
                \stringud{1}{\circleddotted{}}{\circleddotted{}}
                \commutatorhopudot{0}{1}
            \end{tikzpicture}
            $(i)$
        };
        \draw[->,thick] (A1) -- (B) node[midway, scale=0.7,left, xshift=0pt] {$-2$};
    \end{tikzpicture}
\end{align}
and we have $\alpha_i^{zz}=0$.

%%%%%%%%%%%%%%%%%%%%%%%%%%%%%%%%%%%%%%%%
Considering the following  cancellation of the two-support basis element:
\begin{align}
    \begin{tikzpicture}[baseline=(current bounding box.center)]
        % Second TikZ picture in a node
        \node[draw, rectangle, rounded corners=2mm, scale=0.8] (B) at (2,2) {
            \begin{tikzpicture}[baseline=-0.5*\completevertical, rounded corners=0mm]
                \stringud{0}{\circled{}}{\circled{\pm}}
                \stringud{1}{\circled{\pm}}{\circled{}}
            \end{tikzpicture}
            $(i)$
        };
        \node[draw, rectangle, rounded corners=2mm, scale=0.8] (A1) at (0,0) {
            \begin{tikzpicture}[baseline=-0.5*\completevertical, rounded corners=0mm]
                \stringud{0}{\circled{\pm}}{\circled{\pm}}
                \stringud{1}{\circleddotted{}}{\circleddotted{}}
                \commutatorhopudot{0}{1}
            \end{tikzpicture}
            $(i)$
        };
        % First TikZ picture in a node
        \node[draw, rectangle, rounded corners=2mm, scale=0.8] (A2) at (4,0) {
            \begin{tikzpicture}[baseline=-0.5*\completevertical, rounded corners=0mm]
                \stringud{-1}{\circleddotted{}}{\circleddotted{}}
                \stringud{0}{\circled{\pm}}{\circled{\pm}}
                \commutatorhopddot{-1}{0}
            \end{tikzpicture}
            $(i+1)$
        };
        % Arrow from first node to the second
        \draw[->,thick] (A1) -- (B) node[midway, scale=0.7,left, xshift=0pt] {$\pm1$};
        \draw[->,thick] (A2) -- (B) node[midway, scale=0.7,right, xshift=0pt] {$\pm1$};
    \end{tikzpicture}
\end{align}
and we have $\alpha_i^{\pm\pm}=-\alpha_{i+1}^{\pm\pm}$.
%%%%%%%%%%%%%%%%%%%%%%%%%%%%%%%%%%%%%%%%
From the periodic boundary condition, we have $\alpha_i^{\pm\pm} = (-1)^{L}\alpha_{i+L}^{\pm\pm} = (-1)^{L}\alpha_{i}^{\pm\pm}$.
%%%%%%%%%%%%%%%%%%%%%%%%%%%%%%%%%%%%%%%%
Then we can see $\alpha_i^{\pm\pm} = (-1)^{i} \alpha^{\pm\pm} $ where $\alpha^{\pm\pm}$ is a constant and $\alpha^{\pm\pm} = 0$ for odd $L$.

%%%%%%%%%%%%%%%%%%%%%%%%%%%%%%%%%%%%%%%%
From the above argument, $\bck{F_1, H_0} = 0$ is satisfied if 
\begin{align}
    F_1
    &=
    \alpha^{+-}
    \sum_{i=1}^{L}
    \begin{tikzpicture}[baseline=-0.5*\completevertical]
        \stringud{0}{$\circled{+}$}{$\circled{-}$}
    \end{tikzpicture}
    (i)
    +
    \alpha^{-+}
    \sum_{i=1}^{L}
    \begin{tikzpicture}[baseline=-0.5*\completevertical]
        \stringud{0}{$\circled{-}$}{$\circled{+}$}
    \end{tikzpicture}
    (i)
    +
    \alpha^{z \uparrow}
    \sum_{i=1}^{L}
    \begin{tikzpicture}[baseline=-0.5*\completevertical]
        \stringud{0}{$\circled{z}$}{$\circled{}$}
    \end{tikzpicture}
    (i)
    +
    \alpha^{z \downarrow}
    \sum_{i=1}^{L}
    \begin{tikzpicture}[baseline=-0.5*\completevertical]
        \stringud{0}{$\circled{}$}{$\circled{z}$}
    \end{tikzpicture}
    (i)
    \nonumber\\
    &
    \hspace*{4em}
    +
    \alpha^{++}
    \sum_{i=1}^{L}
    (-1)^i
    \begin{tikzpicture}[baseline=-0.5*\completevertical]
        \stringud{0}{$\circled{+}$}{$\circled{+}$}
    \end{tikzpicture}
    (i)
    -
    \alpha^{--}
    \sum_{i=1}^{L}
    (-1)^i
    \begin{tikzpicture}[baseline=-0.5*\completevertical]
        \stringud{0}{$\circled{-}$}{$\circled{-}$}
    \end{tikzpicture}
    (i)
    \nonumber\\
    &=
    \alpha^{+-} S^+
    -
    \alpha^{-+} S^-
    +
    2(\alpha^{z \uparrow} - \alpha^{z \downarrow})S^z
    +
    (\alpha^{z \uparrow} + \alpha^{z \downarrow})N
    -
    \alpha^{++} \eta^+
    +
    \alpha^{--} \eta^-
    \nonumber\\
    &=
    c_+ S^+
    +
    c_- S^-
    +
    c_z S^z
    +
    c_N N
    +
    a_+ \eta^+
    +
    a_- \eta^-
    ,
\end{align}
where we redefine each coefficient in the last line and 
\begin{align}
    S^+ = \sum_{i=1}^{L} c_{j, \uparrow}^{\dag} c_{j, \downarrow}
    ,\quad 
    S^- = \sum_{i=1}^{L} c_{j, \downarrow}^{\dag} c_{j, \uparrow}
    ,\quad 
    S^z = \frac{1}{2}\sum_{i=1}^{L} \paren{n_{j, \uparrow} - n_{j, \downarrow}} 
    ,
\end{align}
are the $\mathrm{SU}(2)$ charges and $N=\sum_{i=1}^{L}\paren{n_{j, \uparrow} + n_{j, \downarrow}-1}$ is the $\mathrm{U}(1)$ charge, and 
\begin{align}
    \eta^+
    =
    \sum_{i=1}^{L} (-1)^{i+1} c_{j, \uparrow}^{\dag} c_{j, \downarrow}^{\dag}
    ,\quad 
    \eta^-
    =
    \sum_{i=1}^{L} (-1)^{i+1} c_{j, \downarrow} c_{j, \uparrow} 
    ,
\end{align}
are the $\eta$-pairing charges and $a_{\pm} = 0$ for odd $L$.

%%%%%%%%%%%%%%%%%%%%%%%%%%%%%%%%%%%%%%%%
Therefore, we have proved that one-local conserved quantities are always written as a linear combination of $\mathrm{SU}(2)$ charges and $\mathrm{U}(1)$ charge and $\eta$-pairing charges.

\end{appendix}

%===========================================================================================%%
%% If you are submitting to one of the Nature Portfolio journals, using the eJP submission   %%
%% system, please include the references within the manuscript file itself. You may do this  %%
%% by copying the reference list from your .bbl file, paste it into the main manuscript .tex %%
%% file, and delete the associated \verb+\bibliography+ commands.                            %%
%%===========================================================================================%%

%\bibliography{sn-bibliography}% common bib file
%% if required, the content of .bbl file can be included here once bbl is generated
%\input sn-article.bbl
%\input{sn-article.bbl}

\end{document}